\documentclass[10pt,journal,compsoc]{IEEEtran}

\usepackage[T1]{fontenc}
\ifCLASSOPTIONcompsoc

\usepackage{cite}

\usepackage{bm}
\usepackage{amsmath}
\usepackage{amsthm}
\usepackage{float}
\usepackage{setspace}
\usepackage{amssymb}
\usepackage{stfloats}
\usepackage{cite}
\usepackage{ragged2e}
\usepackage[ruled,vlined,linesnumbered]{algorithm2e}
\usepackage{amsfonts}
\usepackage{mathrsfs}
\usepackage{amsmath,amsthm}
\usepackage{array,booktabs}
\usepackage{subfigure}
\usepackage{multirow}
\usepackage{cuted}
\usepackage{multicol}
\usepackage{graphicx}
\usepackage{subfigure}
\usepackage{graphicx,xcolor,bm}
\usepackage{hyperref}
\usepackage{threeparttable}
\usepackage{dcolumn}
\usepackage{setspace}
\usepackage{makecell}
\usepackage{lipsum}
\usepackage{enumerate}
\usepackage{mathrsfs}
\usepackage{bbm}
\usepackage[bottom]{footmisc}

\newtheorem{Defn}{Definition}
\newtheorem{lem}{Lemma}

\newtheorem{Prop}{Proposition}
\newtheorem{thm}{Theorem}

\usepackage{cite,bm}
\graphicspath{{figures/}}
\begin{document}
	\title{Masking Intent, Sustaining Equilibrium: Risk-Aware Potential-Game-Based Service Provision in Dynamic Mobile Crowdsensing}
	
\author{Houyi Qi, \IEEEmembership{Graduate Student Member}, \IEEEmembership{IEEE},
	 Minghui Liwang, \IEEEmembership{Senior Member}, \IEEEmembership{IEEE}, 
	Kaiwen Tan,
	 Wenyong Wang, \IEEEmembership{Senior Member}, \IEEEmembership{IEEE},
	 Sai Zou,  \IEEEmembership{Senior Member}, \IEEEmembership{IEEE}, 
	 Yiguang Hong, \IEEEmembership{Fellow}, \IEEEmembership{IEEE},
	\\Xianbin Wang, \IEEEmembership{Fellow}, \IEEEmembership{IEEE},
and	Wei Ni, \IEEEmembership{Fellow}, \IEEEmembership{IEEE}

	\thanks{H. Qi (houyiqi@tongji.edu.cn), M. Liwang (minghuiliwang@tongji.edu.cn), and Y. Hong (yghong@iss.ac.cn) are with the Shanghai Research Institute for Intelligent Autonomous Systems, the State Key Laboratory of Autonomous Intelligent Unmanned Systems, Department of Control Science and Engineering, Tongji University, Shanghai, China. 
	K. Tan (2450031@tongji.edu.cn) is with the Guohao School, Tongji University, Shanghai, China. 
	W. Wang (wywang@must.edu.mo) is with the International Institute of Next Generation Internet, Macau University of Science and Technology, Macau, China.
	S. Zou (dr-zousai@foxmail.com) is with College of Big Data and Information Engineering, Guizhou University, China. 
	X. Wang (xianbin.wang@uwo.ca) is with the Department of Electrical and Computer Engineering, Western University, Ontario, Canada. 
	W. Ni (Wei.Ni@ieee.org) is with the School of Engineering, Edith Cowan University, Perth, Western Australia, Australia. 
	}}

	\IEEEtitleabstractindextext{\vspace{-3.5mm}
		\begin{abstract}
			\justifying
Mobile crowdsensing (MCS) is evolving from basic data collection toward dynamic service provisioning, where platforms must ensure reliable task completion, budget feasibility, and sensing quality under fluctuating worker availability. Beyond raw-data and location privacy, workers’ long-term intent traces, such as task-selection tendencies and participation histories, may be exploited by an honest-but-curious platform to infer private preferences through one- and multi-snapshot intent inference. Meanwhile, worker dropouts and execution uncertainty may destabilize sensing coverage and create redundant sensing, while frequent global re-optimization incurs high interaction overhead and increases observable exposure.
To address these challenges, we propose \textit{iParts}, an intent-preserving and risk-aware two-stage service provisioning framework for dynamic MCS. In the offline stage, workers report perturbed intent vectors via personalized local differential privacy (LDP) with memoization-based permanent randomized response, which suppresses long-term frequency-based inference while preserving report utility for decision making. Using only perturbed intents, the platform constructs a redundancy-aware quality model and performs risk-aware pre-planning under budget availability, individual rationality, quality-failure risk, and intent-mismatch risk constraints. The offline problem is formulated as an exact potential game whose potential function is the expected social welfare, ensuring constrained pure-strategy equilibrium existence and finite-step convergence under exact feasible improvement dynamics. In the online stage, emerging quality deficits are addressed through bounded-round temporary recruitment over idle or standby workers, enabling feasibility-preserving remedial adjustment with limited observable exposure.
Extensive experiments show that iParts achieves a favorable privacy--utility--efficiency trade-off, improving welfare and task completion while reducing redundancy and communication overhead compared with representative benchmark methods.
\end{abstract}

		\vspace{-1mm}
}
	
\maketitle
\IEEEdisplaynontitleabstractindextext
%
\IEEEpeerreviewmaketitle

\section{Introduction}
\IEEEPARstart{B}{y} opportunistically leveraging distributed sensing resources and mobile user capabilities, mobile crowdsensing (MCS) has become an effective paradigm for large-scale data acquisition and sensing-intensive applications. Owing to its flexibility and low deployment cost, MCS has supported many emerging applications, such as urban environment monitoring, traffic congestion inference, and public safety management \cite{wang2022bsif,cai2025towards}. As these applications increasingly require timely, reliable, and high-quality sensing services, MCS platforms are expected to evolve beyond conventional data collection toward service-oriented provisioning. In this emerging paradigm, platforms must jointly ensure budget feasibility, robust task completion under participant dropout, sensing-quality satisfaction, and low interaction overhead, which are essential for supporting dynamic and large-scale MCS applications\cite{cai2025towards,an2024privacy,qi2026Accelerating}.

Reliable MCS service provision relies on incentive-driven recruitment and task scheduling among heterogeneous workers (e.g., mobile users)\cite{an2024privacy,qi2026Accelerating}. However, privacy protection, redundancy-aware quality control, and worker dynamics are tightly coupled in dynamic MCS, making mechanism design challenging. Existing studies often address these issues in isolation, making it difficult to simultaneously ensure sensing reliability, budget feasibility, and privacy protection under uncertain worker participation. Motivated by these gaps, we distill a set of key research questions (RQs) that guide our MCS design and clarify the core motivations behind this work.

\noindent $\bullet$ \textit{RQ1:}
Most privacy-preserving MCS studies focus on location perturbation\cite{cai2025towards,xu2022personalized} or raw-data encryption/anonymization\cite{xiao2020privacy,azhar2022privacy}, while paying less attention to {intent/preference privacy}. In particular, the tasks a worker is willing (or unwilling) to undertake and the long-term patterns embedded in task-selection and participation histories constitute sensitive information that can be readily inferred. Under the honest-but-curious platform assumption, even without direct access to raw sensing data, the platform may still perform one- or multi-snapshot inference\cite{cai2025towards} by exploiting intent reports, historical task allocations, and interaction logs. This auxiliary information may reveal workers' behavioral profiles and risk preferences, discouraging participation and threatening the sustainability of MCS.
Therefore, the first RQ is \textit{how to support effective mapping between tasks and workers without explicitly revealing workers' true intents, while quantifying inference risks and characterizing constraint-aware allocation quality loss under intent perturbation?}

\noindent $\bullet$ \textit{RQ2: }
In practical MCS, tasks may be redundantly sampled by multiple workers. While redundancy improves sensing quality, its marginal gain often diminishes or saturates due to spatio-temporal correlations.
Ignoring redundancy effects may lead to structural imbalance: some ``popular'' tasks become overcrowded with redundant reports, wasting budget, whereas others remain undersampled, resulting in severe quality and coverage gaps that fail to meet required thresholds.
Nevertheless, most prior work overlooks sensing redundancy, implicitly presuming a linear quality gain with increased worker recruitment\cite{qi2026Accelerating,liwang2025long}.
Therefore, our second RQ is: \textit{How to model redundancy-aware sensing quality while ensuring attainable quality, controlled redundancy, and budget feasibility?}

\noindent $\bullet$ \textit{RQ3:}
Dynamic MCS is inherently uncertain due to worker mobility, availability fluctuation, participant dropout, and network variations~\cite{qi2026Accelerating,qi2023matching}. Purely online optimization requires frequent preference disclosure and communication, leading to high interaction overhead~\cite{liwang2025long,qi2023matching} and increased inference exposure. In contrast, purely offline planning may become infeasible when actual worker participation deviates from the planned assignments. Hence, a practical MCS service provisioning mechanism should combine risk-aware offline planning with lightweight online adaptation, so that quality deficits can be mitigated with bounded interactions and stable adjustment. Therefore, our third RQ is: 
\textit{How to design a unified framework integrating risk-aware offline planning with lightweight online adaptation, enabling temporary recruitment within a few interaction rounds under explicit budget and quality requirements?}

To answer the above-discussed RQs, we propose \textbf{iParts} (\underline{i}ntent-\underline{p}rivate and \underline{r}isk-aware \underline{t}wo-\underline{s}tage service provision), an offline-online coordinated paradigm that integrates \textit{risk-aware offline pre-planning} with \textit{lightweight online adjustment}. Our key contributions are summarized as follows:

\noindent $\bullet$ \textit{We develop a unified intent-preserving and risk-aware framework for dynamic MCS provisioning with low interaction overhead.}
By decoupling privacy-sensitive structural decisions into offline planning and confining the online stage to lightweight online adjustment, our framework simultaneously enhances social welfare (SW) and task reliability, mitigates redundant sensing, and reduces the inference surface exploitable for intent profiling.

\noindent $\bullet$ \textit{We formalize and protect workers' intent/preference information under repeated participation via personalized local differential privacy (LDP) with memoization.}
We model workers' intent vectors and their locally perturbed reports, and introduce an \emph{expected intent-report distortion} metric to quantify report-level utility loss induced by perturbation.
To characterize inference risk, we establish one-/multi-snapshot attack models and incorporate an inference-resistance requirement, where memoization mitigates frequency-based profiling within memo-epochs under repeated observations.

\noindent $\bullet$ \textit{We design an offline redundancy-aware and risk-aware pre-planning mechanism via an exact potential game.}
To capture the diminishing marginal gain of redundant sensing, we develop a redundancy-discounted quality aggregation model. Under budget feasibility, individual rationality, quality-failure risk, and intent-mismatch risk constraints, we formulate the offline pre-planning problem as an exact potential game whose potential function is the expected SW. This formulation ensures the existence of a constrained pure-strategy Nash equilibrium and finite-step convergence under feasible improvement dynamics, based on which efficient budget-feasible update algorithms are developed.

\noindent $\bullet$ \textit{We investigate lightweight online adjustment for execution deviations using temporary potential-game dynamics.}
When dynamic execution deviations lead to quality gaps, the platform initiates temporary recruitment over idle/standby workers and coordinates bounded-round remedial decisions through potential-guided updates.
The online interaction is restricted to task-level aggregates and feasibility-screened adjustments, reducing communication overhead while preserving the low-observability design goal.

\noindent $\bullet$ \textit{We provide theoretical properties and extensive evaluations on both service performance and intent privacy.}
We prove the personalized-LDP guarantee of intent reporting, establish mechanism-level post-processing invariance for platform-side quantities derived from perturbed reports, and characterize inference resistance under single- and multi-snapshot adversaries. Experiments show that iParts achieves higher SW and task reliability, better redundancy control, and lower interaction overhead, while offering stronger resistance to intent inference than representative game-theoretic and privacy-preserving benchmarks.

\section{Literature Review}
This section reviews related work from two perspectives: game-theoretic service provisioning (focusing on interaction efficiency and decision dynamics) and privacy-preserving MCS (emphasizing protection scope and inference exposure). We further highlight the key distinctions between our design and existing approaches.

\noindent  $\bullet$ \textbf{Game-theoretic service provision over MCS: interaction efficiency and privacy gaps.}
Game-theoretic paradigms have been extensively explored as a principled means to characterize strategic interactions in MCS, particularly for incentive design in user recruitment, pricing, and task assignment. For example, \textit{Qi et al.}~\cite{qi2026Accelerating} developed a stagewise trading mechanism to achieve stable task-worker matching under task diversity and worker dynamics. \textit{Ouyang et al.}~\cite{Ouyang2025MWRS} coupled a multi-armed bandit model with a tripartite stackelberg game for worker recruitment under hierarchical strategic interactions. \textit{Guang et al.}~\cite{Guang2025Game} designed a collaborative game-based task allocation framework with hierarchical information sharing. \textit{Zhang et al.}~\cite{Zhang2025An} investigated a reverse affine maximizer auction mechanism to improve provider utility in MCS service provisioning.

Despite their effectiveness in capturing strategic behaviors, existing game-theoretic frameworks still face limitations in large-scale and dynamic sensing markets: First, many schemes are interaction-intensive, relying on repeated bidding or pricing, iterative equilibrium computation, and frequent global information exchange, which incurs substantial latency and communication overhead. Second, the schemes often assume that key system parameters are publicly known or can be accurately learned online, which is an assumption that is fragile under the incomplete and uncertain information prevalent in practical crowdsensing. These approaches primarily optimize allocation or equilibrium efficiency, but often overlook intent privacy: \textit{frequent interactions and explicit preference revelation may expose fine-grained task-selection traces, increasing the observable exposure for intent inference.}

\noindent  $\bullet$ \textbf{Privacy-preserving MCS: from data and location privacy to intent privacy.}
Prior studies on privacy-preserving MCS have predominantly concentrated on protecting raw sensing data, such as sensor readings and multimedia reports as well as worker locations.
For instance, \textit{Cai et al.}~\cite{cai2025towards} developed a personalized location-privacy trading framework that jointly balances privacy protection and task allocation efficiency.
\textit{An et al.}~\cite{an2024privacy} combined privacy-preserving recruitment with quality-aware sensing using deviation- and variance-based metrics.
\textit{Zhao et al.}~\cite{zhao2025privacy} proposed a privacy-preserving truth discovery (PPTD) framework for evolving truths, protecting both user and requester privacy while maintaining high estimation accuracy.
\textit{Liang et al.}~\cite{Liang2025Privacy} proposed a fog-assisted PPTD scheme for sensory data while considering server collusion, server dropout, and heterogeneous worker efficiency.
\textit{You et al.}~\cite{You2025Location} combined federated learning with reinforcement learning to support decentralized task selection while preserving location privacy.
\textit{Zhou et al.}~\cite{Zhou2025Unknown} designed a shuffle differential privacy-based auction mechanism to reduce privacy-related utility loss.

Despite these advances, most existing mechanisms primarily focus on protecting information/data privacy (what to sense) and location privacy (where to sense). In dynamic MCS, however, an equally critical yet underexplored dimension is intent or preference privacy, namely, which tasks a worker is willing or unwilling to undertake and how such preferences evolve over time. Many existing mechanisms rely on explicit preference disclosure, repeated bidding, or fine-grained interaction records, which may expose task-selection traces and enable one- or multi-snapshot inference of long-term worker behaviors, even under the honest-but-curious platform assumption. Consequently, while prior studies mainly safeguard data, location, or identity privacy, \textit{intent privacy under repeated observations remains insufficiently addressed, especially when dynamic MCS requires low-interaction service provisioning.} Table~\ref{Tab. RW} compares iParts with representative existing approaches.

\begin{table}[t!]
	{\scriptsize
		\caption{\footnotesize{A summary of related studies\\(Red.: Redundancy controllability, Ris.: Risk controllability)}\label{Tab. RW}}
		\vspace{-0.51cm}
		\begin{center}
			\setlength{\tabcolsep}{1.2mm}{
				\begin{tabular}{|c|c|c|c|c|c|c|c|}
					\hline
					\multirow{2}{*}{\textbf{Reference}} &\multirow{2}{*}{\makecell[c]{ \textbf{Privacy}\\ \textbf{Type}}}
					& \multicolumn{2}{c|}{\makecell[c]{\textbf{Network}\\ \textbf{Characteristics}}} 
					& \multicolumn{2}{c|}{\textbf{Trading Mode}} 
					& \multicolumn{2}{c|}{\makecell[c]{\textbf{Innovative}\\ \textbf{Attributes}}} \\
					\cline{3-8}
					& &Static & Dynamic   
					& Offline & Online  
					& Red.  & Ris. \\ 
					\hline
					
						\makecell[c]{\cite{qi2026Accelerating}}
					&---  &   & $\surd$
					& $\surd$ &$\surd$& &  $\surd$
					\\
					\hline
					
						\makecell[c]{\cite{Ouyang2025MWRS}}
					&---  & $\surd$  & 
					&  &$\surd$& $\surd$ &  
					\\
					\hline
					
						\makecell[c]{\cite{Guang2025Game}}
					&---  &   & $\surd$
					&  &$\surd$&  &  
					\\
					\hline
					
						\makecell[c]{\cite{Zhang2025An}}
					&---  & $\surd$  & 
					&  &$\surd$&  &  
					\\
					\hline
					
					\makecell[c]{\cite{cai2025towards,You2025Location}}
					&location  &   & $\surd$
					&  &$\surd$&  &  $\surd$
				  \\
					\hline
					
					\makecell[c]{\cite{an2024privacy}}
					& data&  & $\surd$
					&  &$\surd$&   &  
					 \\
					\hline
					
					\makecell[c]{\cite{zhao2025privacy}}
					& data&  & $\surd$
					&  &$\surd$& $\surd$  &$\surd$  
					\\
					\hline

					\makecell[c]{\cite{Liang2025Privacy}}
					& data& $\surd$ & 
					& $\surd$ &&    &$\surd$  
					\\
					\hline
					
					\makecell[c]{\cite{Zhou2025Unknown}}
					& data&  & $\surd$
					&  &$\surd$&    &$\surd$  
					\\
					\hline

					\textbf{iParts (Ours)}
					& Intent &  & $\surd$
					&$\surd$ &$\surd$ & $\surd$ & $\surd$
					 \\
					\hline
			\end{tabular}}
	\end{center}}
	\vspace{-1.5mm}
\end{table}

\section{Preliminaries and System Model}
\subsection{Overview of iParts}
\label{sec:overview_method}

We consider a dynamic MCS system consisting of three main components: \textit{(i)} a sensing platform, \textit{(ii)} a set of sensing tasks $\bm{\mathcal{S}}=\{s_1,\ldots,s_{|\bm{\mathcal{S}}|}\}$, and \textit{(iii)} a set of mobile workers $\bm{\mathcal{W}}=\{w_1,\ldots,w_{|\bm{\mathcal{W}}|}\}$. The platform is assumed to be \emph{honest-but-curious}: it faithfully follows the prescribed protocol, but may infer workers' sensitive intent/preference information from reported signals, task allocation histories, and interaction logs.
\begin{figure}[t!]
	\centering
	\includegraphics[width=1\linewidth]{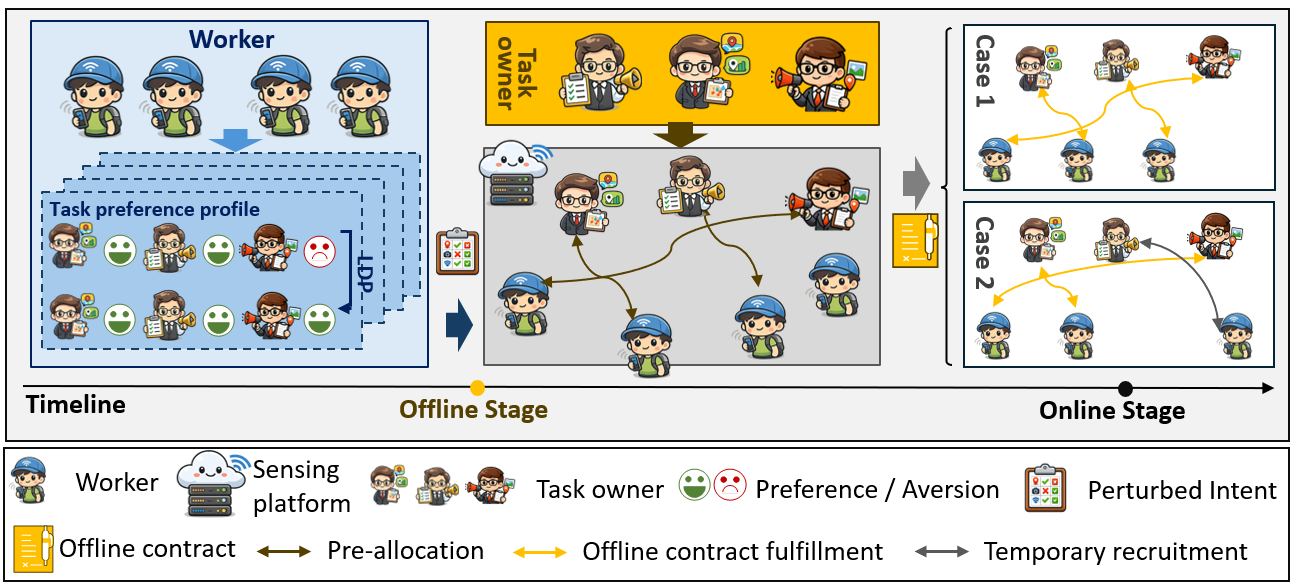}
	\caption{Timeline and workflow of the proposed iParts framework in dynamic MCS.}
	\label{fig:system}
\end{figure}
In this setting, the platform aims to maximize practical social welfare (pSW) while satisfying several system-level requirements, including \textit{(i)} intent privacy preservation, \textit{(ii)} redundancy-aware sensing-quality assurance, \textit{(iii)} budget feasibility, and \textit{(iv)} low online interaction overhead. To achieve these goals, iParts adopts an offline-online coordinated design, which separates privacy-sensitive structural planning from lightweight online remediation, as illustrated in Fig.~\ref{fig:system}.

\noindent
$\bullet$ \textit{Offline stage: intent-preserving and risk-aware pre-planning.}
Before task execution, the platform performs risk-aware pre-planning based on workers' perturbed intent reports and public task/worker parameters. This stage addresses two key challenges in dynamic MCS: worker participation is uncertain due to mobility and dropouts, and the platform cannot directly access workers' true intents due to privacy constraints. To this end, we introduce \underline{r}isk-\underline{a}ware \underline{p}re-planning with long-term \underline{co}ntr\underline{a}ct \underline{d}esign (RAPCoD), which jointly accounts for execution uncertainty, redundancy-aware quality requirements, and budget feasibility. To support privacy-preserving planning, PRIMER and MIRROR are developed to generate personalized-LDP intent reports, while ASPIRE-Off solves the resulting offline planning problem through an exact potential game that maximizes expected social welfare (eSW). The offline stage produces intent-preserving and risk-aware long-term contracts, thereby reducing the need for frequent online preference disclosure and global re-optimization.

\noindent
$\bullet$ \textit{Online stage: lightweight remediation via temporary recruitment.}
During task execution, worker availability fluctuations and execution uncertainty may cause deviations from the offline plan and lead to sensing-quality deficits. To address these deficits under limited interaction budgets, iParts activates a lightweight online adjustment mechanism based on a transient potential game for temporary recruitment among idle or standby workers. With potential-guided remedial updates, the online stage performs bounded-round adjustment to mitigate quality deficits when feasible, while improving pSW with low overhead and limited observable exposure.

Together, the two stages move privacy-sensitive structural decisions to offline planning under personalized LDP and risk constraints, while restricting online decisions to lightweight and temporary remediation.

\subsection{Basic Modeling}
We introduce the basic models of tasks, workers, and intent preferences with perturbed reports.

\noindent\textbf{\textit{(i)} Tasks:}
The attributes of each sensing task $s_i \in \bm{\mathcal{S}}$ are characterized by the triple
$\{\, loc_i^\mathrm{S},\, B_i,\,  Q_i^\mathrm{D}\,\}$,
where $loc_i^\mathrm{S}$ denotes its spatial location;
$B_i$ is the nominal budget;
and $Q_i^\mathrm{D}$ represents the minimum sensing-quality requirement (i.e., quality-demand threshold).

\noindent\textbf{\textit{(ii)} Workers:}
The attributes of each worker $w_j \in \bm{\mathcal{W}}$ are modeled by a 4-tuple
$
\{\, loc_j^\mathrm{W},\, \varepsilon_j,\, \theta_j,\, \alpha_j \,\}$,
where $loc_j^\mathrm{W}$ denotes the worker's location,
$\varepsilon_j$ is the personalized intent-privacy budget used in personalized LDP,
$\theta_j>0$ is the sensing capability parameter (a larger $\theta_j$ implies a smaller observation error),
while $\alpha_j\in\{0,1\}$ indicates whether $w_j$ is practically available during the online stage, capturing the arrival/dropout uncertainty.
Moreover, we model $\alpha_j$ as a Bernoulli random variable, i.e.,
$\alpha_j\sim \mathbf{B}(\pi_j)$, with
$\Pr(\alpha_j=1)=\pi_j$ and $\Pr(\alpha_j=0)=1-\pi_j$.

\noindent\textbf{\textit{(iii)} Intent preference and perturbed intent report:}
We denote the true intent vector of worker $w_j$ by
$\mathbf{b}_j=[b_{i,j}]_{s_i\in\bm{\mathcal{S}}}$,
where $b_{i,j}=1$ indicates that $w_j$ is willing to undertake task $s_i$ (and $b_{i,j}=0$ otherwise).
To prevent direct disclosure of task preferences while preserving participatory autonomy, each worker locally generates perturbed intent report:
$\tilde{\mathbf{b}}_j=[\tilde b_{i,j}]_{s_i\in\bm{\mathcal{S}}}$. The platform can observe only $\tilde{\mathbf{b}}_j$, whereas the true intent vector $\mathbf{b}_j$ remains private.

\subsection{Intent-Report Distortion Model}\label{subsec:intent_quality_loss}
To quantify the report-level utility loss induced by intent perturbation, we introduce the expected intent-report distortion (EIRD). Since the platform can only observe perturbed intent reports, the reported preferences used for offline pre-planning may deviate from workers' true intents. Such deviation may reduce the informativeness of intent reports and affect the quality of subsequent task-worker mapping. Different from allocation-level mismatch risk, which depends on the final assignment decision, EIRD characterizes the mechanism-dependent distortion between the true intent vector and the perturbed report, and can be directly computed from the perturbation rule.

\begin{Defn}[Expected intent-report distortion (EIRD)]\label{def:expected_report_distortion}
	For worker $w_j$, the EIRD is defined as\footnote{Eq.  (\ref{eq:Qloss_intent_rev}) takes expectation over both the prior distribution and IPM.}
	\begin{equation}
		\begin{aligned}
			Q&^{\mathrm{loss}}_j(\phi,f,\Delta_j)
			=
			\mathbb{E}_{\mathbf{b}_j\sim\phi}\Big[
			\mathbb{E}_{\tilde{\mathbf{b}}_j\sim f(\cdot\mid \mathbf{b}_j)}\big[
			\Delta_j(\tilde{\mathbf{b}}_j,\mathbf{b}_j)
			\big]\Big]
			\\&=
			\sum_{\mathbf{b}_j\in\mathcal{B}_j}\phi(\mathbf{b}_j)
			\sum_{\tilde{\mathbf{b}}_j\in\tilde{\mathcal{B}}_j}
			f(\tilde{\mathbf{b}}_j\mid\mathbf{b}_j)\,
			\Delta_j(\tilde{\mathbf{b}}_j,\mathbf{b}_j),
		\end{aligned}
		\label{eq:Qloss_intent_rev}
	\end{equation}
	where $\phi(\mathbf{b}_j)$ denotes the prior distribution of the true intent vector, and
	$f(\cdot)$ denotes the intent perturbation mechanism (IPM).
	Moreover, $\mathcal{B}_j$ and $\tilde{\mathcal{B}}_j$ describe the input and output spaces of IPM, respectively.
\end{Defn}

To measure the entry-wise deviation between the perturbed report and the true intent vector, we adopt a weighted Hamming distortion:
\begin{equation}
	\Delta_j(\tilde{\mathbf{b}}_j,\mathbf{b}_j)
	\triangleq
	\sum_{s_i\in\bm{\mathcal{S}}}
	\gamma_j \big|b_{i,j}-\tilde b_{i,j}\big|,
	\label{eq:intent_distortion_rev}
\end{equation}
where $\big|b_{i,j}-\tilde b_{i,j}\big|=1$ if the intent entry associated with task $s_i$ is flipped by the IPM, and equals $0$ otherwise. The coefficient $\gamma_j\ge 0$ is a worker-specific weighting factor that measures the normalized cost of one flipped intent entry. Accordingly, $\Delta_j(\tilde{\mathbf{b}}_j,\mathbf{b}_j)$ aggregates the entry-wise perturbation errors into a report-level distortion score, providing a simple and mechanism-dependent proxy for quantifying how intent perturbation changes the information observed by the platform.

\subsection{Intent Inference Model under Repeated Observations}\label{subsec:intent_attack_model}
We adopt an \emph{honest-but-curious} platform assumption, under which the platform faithfully executes the prescribed offline and online protocols, but may infer workers' true intents or long-term preferences from observable intent reports, allocation histories, and interaction logs. Accordingly, we analyze two representative intent-inference attacks: \emph{one-snapshot} attack and \emph{multi-snapshot} attack.

\subsubsection{One-snapshot attack}\label{subsubsec:one_snapshot_intent}
The adversary is assumed to know the prior distribution $\phi(\mathbf{b}_j)$ of worker $w_j$'s true intent vector\footnote{Following a conservative inference setting, we assume that $\phi(\mathbf{b}_j)$ is known to the adversary. In practice, such a prior can be estimated from historical participation statistics, public information, or long-term observations.}~\cite{cai2025towards}, and observes the perturbed intent report $\tilde{\mathbf{b}}_j$ generated by the IPM, i.e., $f(\tilde{\mathbf{b}}_j\mid\mathbf{b}_j)$.
By Bayes' rule, the posterior distribution of the true intent vector is
\begin{equation}
	\Pr(\mathbf{b}_j\mid\tilde{\mathbf{b}}_j)=
	\frac{f(\tilde{\mathbf{b}}_j\mid\mathbf{b}_j)\,\phi(\mathbf{b}_j)}
	{\sum_{\mathbf{b}'_j\in\mathcal{B}_j } f(\tilde{\mathbf{b}}_j\mid\mathbf{b}'_j)\,\phi(\mathbf{b}'_j)}.
	\label{eq:posterior_intent_rev}
\end{equation}
where $\mathbf{b}'_j$ is a dummy summation variable over the feasible set $\mathcal{B}_j$, and the denominator is the marginal probability $\Pr(\tilde{\mathbf{b}}_j)$ for normalization.
Given the posterior, the adversary performs optimal inference by choosing an estimate $\hat{\mathbf{b}}_j$ that minimizes the expected inference error (eIE)~\cite{cai2025towards,shokri2012protecting}:
\begin{equation}
	\hat{\mathbf{b}}_j=
	\arg\min_{\hat{\mathbf{b}}_j}
\sum_{\mathbf{b}_j\in\mathcal{B}_j}
\Pr(\mathbf{b}_j\mid\tilde{\mathbf{b}}_j)\,
d(\hat{\mathbf{b}}_j,\mathbf{b}_j).
\label{eq:optimal_inference_intent_rev}
\end{equation}

To measure the discrepancy between the inferred and true intent vectors, we adopt a task-level \emph{weighted Hamming distance} as the inference error metric:
\begin{equation}
d(\hat{\mathbf{b}}_j,\mathbf{b}_j)
=\sum_{s_i\in\bm{\mathcal{S}}} \omega_i\,\big|\hat{b}_{i,j}-b_{i,j}\big|, \qquad \omega_i\ge 0,
\label{eq:distance_intent_rev}
\end{equation}
where $\omega_i$ denotes the sensitivity/importance weight of task $s_i$.
The adversary's eIE under one-snapshot observation is defined as
\begin{equation}
\xi_j
=
\sum_{\mathbf{b}_j\in\mathcal{B}_j}\phi(\mathbf{b}_j)
\sum_{\tilde{\mathbf{b}}_j\in\tilde{\mathcal{B}}_j} f(\tilde{\mathbf{b}}_j\mid\mathbf{b}_j)\,
d\!\left(\hat{\mathbf{b}}_j(\tilde{\mathbf{b}}_j),\,\mathbf{b}_j\right),
\label{eq:inference_error_intent_rev}
\end{equation}
where $\hat{\mathbf{b}}_j(\tilde{\mathbf{b}}_j)$ is the adversary's optimal estimator after observing $\tilde{\mathbf{b}}_j$.

To limit one-snapshot intent inference, we require the inference error to be no smaller than a threshold $\beta^{0}$:
\begin{equation}
\xi_j \ge \beta^{0}.
\label{eq:error_lower_bound_intent_rev}
\end{equation}
Following representative studies\cite{cai2025towards,wang2019mobile}, to facilitate computation and constraint decomposition, we introduce an auxiliary function $h(\tilde{\mathbf{b}}_j)$, and rewrite \eqref{eq:error_lower_bound_intent_rev} equivalently as the following two constraints:
\begin{equation}
\sum_{\mathbf{b}_j\in\mathcal{B}_j}\phi(\mathbf{b}_j)\,
f(\tilde{\mathbf{b}}_j\mid\mathbf{b}_j)\,
d\!\left(\hat{\mathbf{b}}_j(\tilde{\mathbf{b}}_j),\,\mathbf{b}_j\right)
\ge h(\tilde{\mathbf{b}}_j),\quad \forall \tilde{\mathbf{b}}_j,
\label{eq:h_lower_bound_intent_1_rev}
\end{equation}
\begin{equation}
\sum_{\tilde{\mathbf{b}}_j\in\tilde{\mathcal{B}}_j} h(\tilde{\mathbf{b}}_j)\ge \beta^{0}.
\label{eq:h_lower_bound_intent_2_rev}
\end{equation}

\subsubsection{Multi-snapshot attack}\label{subsubsec:multi_snapshot_intent}
In multi-snapshot settings, the adversary continuously observes a sequence of perturbed intent reports submitted by the same worker $w_j$ over multiple rounds,
$\{\tilde{\mathbf{b}}^{(1)}_j,\tilde{\mathbf{b}}^{(2)}_j,\ldots,\tilde{\mathbf{b}}^{(T)}_j\}$,
and performs frequency-based profiling on each task dimension.

For task dimension $s_i\in\bm{\mathcal{S}}$, the empirical frequency of observing output 1 over $T$ reports is defined as
\begin{equation}
F_{i,j} \triangleq \frac{1}{T}\sum_{\tau=1}^{T}\tilde b^{(\tau)}_{i,j},
\label{eq:freq_stat_intent_rev}
\end{equation}
where $\tilde b^{(\tau)}_{i,j}\in\{0,1\}$ is the perturbed response of worker $w_j$ on dimension $s_i$ at service round $\tau$.
When IPM is independent across rounds and has a stable expectation, the law of large numbers implies that $F_{i,j}$ converges to its mathematical expectation with high probability, thereby providing a stable statistical signal for the adversary.

Based on $F_{i,j}$, the adversary may apply a likelihood-based frequency-threshold rule to infer the underlying intent.
For the standard \emph{symmetric} binary randomized response, this reduces to a simple majority-vote test:
\begin{equation}
	\hat b_{i,j} = \mathbbm{1}\{F_{i,j}\ge 1/2\},\qquad \forall s_i\in\bm{\mathcal{S}},
	\label{eq:majority_inference_intent_rev}
\end{equation}
whereas in general the optimal threshold depends on the perturbation parameters and the prior.

The above attack shows that if per-round perturbations are independent and identically distributed (i.i.d.) with a stable mean, averaging $\{\tilde b^{(\tau)}_{i,j}\}_{\tau=1}^{T}$ progressively reduces the estimation variance. As $T$ increases, the adversary's inference accuracy may improve, i.e., the eIE decreases, leading to \emph{privacy degradation} under long-term observations.

To suppress frequency convergence and profiling inference under repeated observations,
we incorporate two complementary designs into IPM:
\textit{(i)}~\emph{$\varepsilon_j$-personalized LDP} to bound per-report distinguishability; and
\textit{(ii)}~\emph{memoization-based permanent randomized response} to suppress the adversary's variance-reduction gain from repeated observations within an epoch.
Specifically, within each memo-epoch $e$, worker $w_j$ first generates an epoch-stable permanent perturbed intent vector
$\tilde{\mathbf{b}}^{\mathrm{perm}}_j(e)$ via entry-wise personalized randomized response under $\varepsilon_j$,
stores it locally, and then {reuses the same} $\tilde{\mathbf{b}}^{\mathrm{perm}}_j(e)$ for all reporting rounds in epoch $e$.
Therefore, repeated observations within the same epoch reveal only the same epoch-stable perturbed vector, preventing the adversary from progressively reducing uncertainty through intra-epoch frequency statistics.
To accommodate slowly drifting intents and limit long-horizon linkability, we further enforce epoch-level refresh:
when entering $e{+}1$, worker $w_j$ discards $\tilde{\mathbf{b}}^{\mathrm{perm}}_j(e)$ and regenerates a new permanent report
$\tilde{\mathbf{b}}^{\mathrm{perm}}_j(e{+}1)$ using the same personalized randomized response rule with budget $\varepsilon_j$.

\subsection{Personalized LDP Guarantee}\label{subsec:personalized_ldp_intent}
\begin{Defn}[$\varepsilon_j$-personalized LDP~\cite{yu2017dynamic}]
\label{def:personalized_ldp_intent}
Given the task set $\bm{\mathcal{S}}$ and the output space $\tilde{\mathcal{B}}$, two intent vectors $\mathbf{b}_j^{(1)}$ and $\mathbf{b}_j^{(2)}$ are adjacent if they differ in exactly one task dimension, i.e., $\|\mathbf{b}_j^{(1)}-\mathbf{b}_j^{(2)}\|_0=1$. A randomized mechanism $f$ is said to satisfy $\varepsilon_j$-personalized LDP for worker $w_j$ if, for any adjacent $\mathbf{b}_j^{(1)}$ and $\mathbf{b}_j^{(2)}$ and any output $\tilde{\mathbf{b}}^{*}\in\tilde{\mathcal{B}}$, it holds that
\begin{equation}
\frac{f(\tilde{\mathbf{b}}^{*}\mid \mathbf{b}_j^{(1)})}{f(\tilde{\mathbf{b}}^{*}\mid \mathbf{b}_j^{(2)})} \le e^{\varepsilon_j}.
\label{eq:personalized_ldp_def_rev}
\end{equation}
The privacy budget $\varepsilon_j$ is specified according to worker $w_j$'s personalized privacy preference, with a smaller value indicating stronger privacy protection.
\end{Defn}

Moreover, our multi-snapshot attack relies on the fact that intent vectors are repeatedly observable on discrete task dimensions, which differs from location-privacy settings where the exact same coordinate is less likely to recur in a continuous space.
Therefore, multi-snapshot statistical inference is particularly relevant for intent privacy in dynamic MCS.
Condition~\eqref{eq:personalized_ldp_def_rev} provides a mechanism-level distinguishability guarantee for intent reporting, which facilitates privacy parameterization and downstream post-processing analysis. In contrast, constraints~\eqref{eq:h_lower_bound_intent_1_rev}--\eqref{eq:h_lower_bound_intent_2_rev} characterize inference resistance from the perspective of the adversary's optimal estimation error.

\section{Offline Stage: Intent-Preserving and Risk-Aware Service Pre-Planning}
\subsection{Redundancy-Aware Sensing Quality}
To capture the common phenomenon that repeated sampling of the same task by multiple workers yields diminishing marginal quality gains, we build a redundancy-aware task-level quality aggregation model.

For task $s_i$, we directly characterize the sensing reliability of worker $w_j$ by an \emph{effective error variance} $\sigma_{i,j}^2$, which abstracts the overall uncertainty of the data produced by $w_j$ when serving $s_i$, as given by
\begin{equation}
	\sigma_{i,j}^2
=\frac{\sigma_{0,i}^2}{\theta_j}
+(\sigma_{i,j}^\mathrm{data})^2,
\qquad \theta_j>0,
	\label{eq:sigma_decompose}
\end{equation}
where $\sigma_{0,i}^2$ represents the task-dependent reference variance under standard sensing conditions, and $\theta_j$ denotes the sensing capability factor of worker $w_j$.
A larger $\theta_j$ implies higher sensing capability and thus a smaller capability-related variance.
The data-quality term $(\sigma_{i,j}^\mathrm{data})^2$ can reflect factors such as distance, delay, and network conditions.
Accordingly, we define the single-shot sensing quality of $w_j$ for task $s_i$ as the inverse of the error variance (i.e., the precision):
\begin{equation}
	q_{i,j}= \frac{1}{\sigma_{i,j}^2}.
	\label{eq:qij_precision}
\end{equation}

Let the realized number of participating workers for task $s_i$ be $n_i = \sum_{w_j\in\bm{\mathcal{W}}} x_{i,j}\alpha_j$, where $x_{i,j}\in\{0,1\}$ is the offline assignment decision and $\alpha_j$ indicates whether worker $w_j$ actually arrives/participates online.
To capture the inherent correlation among sensing contributions arising from spatial and temporal redundancy, we introduce a redundancy correlation coefficient $\zeta_i\in[0,1)$ for each task.
This characterizes the diminishing marginal utility of multiple worker contributions caused by repeated or correlated sampling. Accordingly, we define the redundancy-aware aggregated quality as
\begin{equation}
	Q_i=
	\frac{\sum_{w_j\in\bm{\mathcal{W}}} x_{i,j}\alpha_j q_{i,j}}
	{1+\big(n_i-1\big)\zeta_i},
	\label{eq:Qi_redundancy}
\end{equation}
This formulation discounts correlated sensing contributions as the number of participating workers increases, thereby capturing the diminishing effective gain under redundant sensing.

\subsection{Utilities and Risks of Workers and Tasks}
\textit{(i) Worker utility:}
The utility of worker $w_j$ consists of payment and cost:
\begin{equation}
	u_j^\mathrm{W}
	=
	\sum_{s_i\in\bm{\mathcal{S}}} x_{i,j}\alpha_j\big(p_{i,j}-c^\mathrm{W}_{i,j}\big),
	\label{eq:uj_worker}
\end{equation}
where $p_{i,j}$ is the payment from the platform to worker $w_j$ for executing task $s_i$, and $c^\mathrm{W}_{i,j}$ is the worker cost.
We decompose the total cost of worker $w_j$ serving task $s_i$ into an execution cost and a privacy cost, i.e.,
\begin{equation}
	c^\mathrm{W}_{i,j}
	=
	c^{\text{exe}}_{i,j}
	+
	c^{\text{priv}}_{j}.
	\label{eq:cw_decompose}
\end{equation}
The execution cost models the physical effort (e.g., traveling) required for $w_j$ to accomplish task $s_i$, defined as
\begin{equation}
	c^{\text{exe}}_{i,j}
	=
	\mu_j \cdot \mathrm{dist}\!\left(loc_j^\mathrm{W},\,loc_i^\mathrm{S}\right),
	\label{eq:cexe_def}
\end{equation}
where $\mu_j\ge 0$ is a worker-specific cost coefficient that converts travel distance into monetary cost, and
$\mathrm{dist}(\cdot)$ denotes the Euclidean distance.

To capture the compensation demand associated with intent-privacy protection, we model the privacy cost of worker $w_j$ as
\begin{equation}
	c^{\text{priv}}_{j}
	=
	\frac{\lambda_j}{\varepsilon_j},
	\label{eq:cpriv_def}
\end{equation}
where $\lambda_j\ge 0$ is a scaling coefficient converting privacy preference into monetary cost, and
$\varepsilon_j$ is the worker's personalized privacy budget in the IPM.
A smaller $\varepsilon_j$ implies stronger privacy protection and therefore a larger privacy cost, reflecting the opportunity cost induced by the worker's privacy preference.

\noindent\textit{(ii) Task utility:}
The utility of task $s_i$ involves the value gained from sensing quality and payment cost.
Using $\omega_3>0$ to map the aggregated quality into a utility gain, we have
\begin{equation}
	u_i^\mathrm{S}
	=
	\omega_3\,Q_i
	-\sum_{w_j\in\bm{\mathcal{W}}} x_{i,j}\alpha_j p_{i,j}.
	\label{eq:ui_task}
\end{equation}

\noindent\textit{(iii) Social welfare:}
The overall pSW is defined as the sum of task utilities and worker utilities:
\begin{equation}
	\mathbb{SW}=\sum_{s_i\in\bm{\mathcal{S}}}u_i^\mathrm{S}+\sum_{w_j\in\bm{\mathcal{W}}}u_j^\mathrm{W}.
	\label{eq:SW_def}
\end{equation}

\noindent\textit{(iv) Risk analysis:}
Due to the IPM and the underlying dynamism of MCS, we introduce two types of risks:

\noindent $\bullet$ \textbf{Unsatisfactory service quality risk.}
Owing to stochastic worker participation and the redundancy-discount effect in quality aggregation, task $s_i$ may receive an insufficient realized sensing quality, i.e., the aggregated sensing quality falls below the required minimum level.
Based on \eqref{eq:Qi_redundancy}, we define the risk of task $s_i$ receiving unsatisfactory service quality as
\begin{equation}
	R_i^\mathrm{Qual}
	=
	\Pr\left(
	\frac{\sum_{w_j\in\bm{\mathcal{W}}} x_{i,j}\alpha_j q_{i,j}}
	{1+\big(n_i-1\big)\zeta_i}
	<
	Q_i^\mathrm{D}
	\right).
	\label{eq:R_qual}
\end{equation}

\noindent $\bullet$ \textbf{Intent mismatch risk.}
Since the platform can observe only perturbed intent reports, a worker may be assigned to a task that it is truly unwilling to undertake, which may reduce worker satisfaction and discourage long-term participation.
We define the intent-mismatch ratio of worker $w_j$ as
\begin{equation}
	\mathrm{mis}_j(\bm{x},\mathbf{b}_j)
	\triangleq
	\frac{1}{|\bm{\mathcal{S}}|}\sum_{s_i\in\bm{\mathcal{S}}} x_{i,j}\,(1-b_{i,j}),
	\label{eq:mis_ratio}
\end{equation}
where $\bm{x}=[x_{i,j}]$ denotes the assignment matrix, and a mismatch occurs if and only if $x_{i,j}=1$ (i.e., the platform assigns $s_i$ to $w_j$) while $b_{i,j}=0$ (i.e., the worker is truly unwilling to execute $s_i$).
Since the platform is unaware of $\mathbf{b}_j$ and can only evaluate mismatch under the prior distribution $\phi(\mathbf{b}_j)$, we define the expected mismatch ratio as an intent mismatch risk:
\begin{equation}
	\begin{aligned}
		R_j^\mathrm{Pref}
		&=
		\mathbb{E}_{\mathbf{b}_j\sim \phi}\!\left[\mathrm{mis}_j(\bm{x},\mathbf{b}_j)\right]
		\\&=
		\sum_{\mathbf{b}_j\in\mathcal{B}_j}\phi(\mathbf{b}_j)\,
		\frac{1}{|\bm{\mathcal{S}}|}\sum_{s_i\in\bm{\mathcal{S}}} x_{i,j}(1-b_{i,j}).
	\end{aligned}
	\label{eq:R_pref}
\end{equation}
A smaller $R_j^\mathrm{Pref}$ indicates that, in a statistical sense, the tasks assigned to worker $w_j$ better align with its true intent, and thus the worker is more likely to remain willing to participate in subsequent service execution.

\subsection{Joint Optimization in the Offline Stage}\label{subsec:offline_problem}
In the offline stage, the platform performs risk-aware service pre-planning based on perturbed intent reports.
Accordingly, we formulate an eSW maximization problem:
\begin{align}
	\bm{\mathcal{F}^{\mathrm{off}}}:~&\underset{\bm{x}}{\max}~ \mathbb{E}_{\bm{\alpha}}\!\left[ \mathbb{SW} \right]
	\tag{25}\label{ProES_rev}\\
	\text{s.t.}~~~&
	\sum_{s_i\in\bm{\mathcal{S}}}x_{i,j}\le 1,\ \forall w_j\in \bm{\mathcal{W}}
	\tag{25a}\label{equ.25a_rev}\\
	& p_{i,j}\ge c_{i,j}^\mathrm{W} ,\ \forall s_i\in\bm{\mathcal{S}},\ \forall w_j\in\bm{\mathcal{W}}
	\tag{25b}\label{equ.25b_rev}\\
	& \sum_{w_j\in\bm{\mathcal{W}}}x_{i,j}p_{i,j}\le B_i,\ \forall s_i\in\bm{\mathcal{S}}
	\tag{25c}\label{equ.25c_rev}\\
	& f~\text{satisfies } \varepsilon_j\text{-personalized LDP in \eqref{eq:personalized_ldp_def_rev}},\ \forall w_j\in\bm{\mathcal{W}}
	\tag{25d}\label{equ.25d_rev}\\
	& Q^{\text{loss}}_j(\phi,f,\Delta_j)\le Q^{\text{loss}}_{\max},\ \forall w_j\in\bm{\mathcal{W}}
	\tag{25e}\label{equ.25e_rev}\\
	& \eqref{eq:h_lower_bound_intent_1_rev}\ \text{and}\ \eqref{eq:h_lower_bound_intent_2_rev}\ \text{hold},\ \forall w_j\in\bm{\mathcal{W}}
	\tag{25f}\label{equ.25f_rev}\\
	& R_i^\mathrm{Qual}\le \rho_1,\ \forall s_i\in\bm{\mathcal{S}}
	\tag{25g}\label{equ.25g_rev}\\
	& R_j^\mathrm{Pref}\le \rho_2,\ \forall w_j\in\bm{\mathcal{W}}.
	\tag{25h}\label{equ.25h_rev}
\end{align}
\setcounter{equation}{25} where $\rho_1$ and $\rho_2$ are risk thresholds in $(0,1]$,
constraint \eqref{equ.25a_rev} limits each worker to serve at most one task. Individual rationality is guaranteed by \eqref{equ.25b_rev}, requiring the payment to cover the worker's cost, while budget feasibility is enforced by \eqref{equ.25c_rev}, which bounds the total payment of each task owner within its budget. Privacy protection is imposed through constraint \eqref{equ.25d_rev}, where workers report perturbed intents under $\varepsilon_j$-personalized LDP. Meanwhile, \eqref{equ.25e_rev} restricts the expected intent-report distortion of each worker to be no greater than 
$Q^{\text{loss}}_{\max}$, and \eqref{equ.25f_rev} ensures a minimum inference-error level against intent inference attacks\footnote{These privacy-side requirements are implemented by the intent-privacy module in Sec.~\ref{Sec:Solution_IPM}, where workers that fail the feasibility check are excluded from subsequent candidate construction.}. Finally, constraints \eqref{equ.25g_rev} and \eqref{equ.25h_rev} explicitly bound the risks of service-quality violation and intent mismatch, respectively.

Problem $\bm{\mathcal{F}^{\mathrm{off}}}$ is a mixed-integer stochastic optimization problem, where elements in $\bm{x}$ are coupled with budget constraints and further complicated by risk constraints, redundancy-aware nonlinear quality, and privacy/inference requirements.
To address this problem, we develop RAPCoD, which consists of two key modules, as illustrated in Fig.~\ref{fig:offline}. We first equip each worker with an intent-privacy module (Module A) comprising two algorithms (see Sec. \ref{Sec:Solution_IPM}): \textit{(i)} \underline{p}e\underline{r}sonalized \underline{i}ntent-report deter\underline{m}ination for privacy-budg\underline{e}t calib\underline{r}ation (PRIMER) , and \textit{(ii)} \underline{m}emo\underline{i}zation-based personalized LDP perturbation for gene\underline{r}ating intent \underline{r}ep\underline{or}ts (MIRROR).
Then, we further cast the offline pre-planning problem as an exact potential game (Module B), for which we develop an algorithm called \underline{a}synchronous \underline{s}elf-organized \underline{p}otential \underline{i}mprovement for \underline{re}liable \underline{off}line pre-contract (ASPIRE-Off).
ASPIRE-Off performs exact feasible best-response updates over the unilateral feasible strategy space of each task, where each accepted update is evaluated by the exact offline potential rather than by sampling-based estimation.
As a result, ASPIRE-Off terminates at a constraint-satisfying pure-strategy NE, and the resulting equilibrium profile directly yields the task--worker assignment for long-term service contracts.

\subsection{Design of Module A}\label{Sec:Solution_IPM}
Our intent-privacy module jointly balances privacy protection and decision utility via two components. First, PRIMER calibrates each worker's personalized privacy budget so that the expected intent-report distortion remains below $Q^{\mathrm{loss}}_{\max}$ while the one-snapshot expected inference error remains no smaller than $\beta^0$ (see Alg.~\ref{Alg:PIRD}). Second, MIRROR applies memoization-based personalized LDP to generate an epoch-stable perturbed intent vector that is reused within each memo-epoch, avoiding independent per-round perturbation samples and reducing the variance-reduction gain of frequency-based profiling under multi-snapshot observations (see Alg.~\ref{Alg:MSIPM}).

\begin{figure}[t!]
	\centering
	\includegraphics[width=1\linewidth]{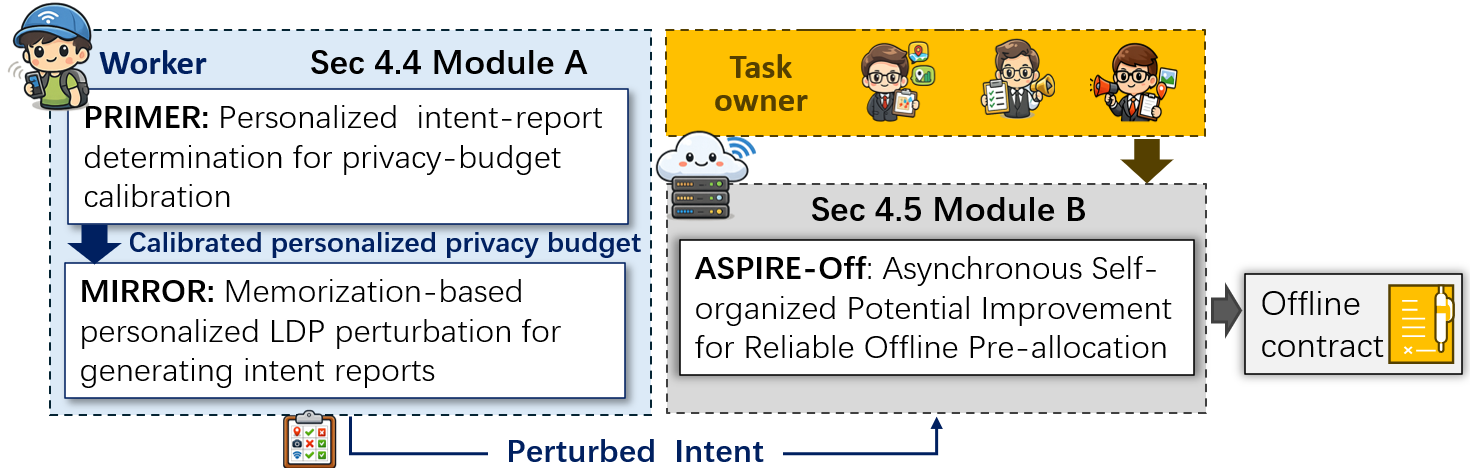}
	\caption{A flow chart regarding our proposed RAPCoD in the offline stage.}
	\label{fig:offline}
\end{figure}

\subsubsection{PRIMER: Personalized intent-report determination for privacy-budget calibration}
\noindent
PRIMER calibrates a deployable personalized privacy budget $\varepsilon_j^\star$ for each worker by jointly considering report usefulness and one-snapshot inference resistance, as shown in Alg.~\ref{Alg:PIRD}, incrementally searching for a \emph{feasible} personalized budget that satisfies both a usefulness constraint and a privacy-robustness constraint.
The main steps are summarized as follows.

\begin{algorithm}[b!]
	{\scriptsize
		\caption{Proposed PRIMER}
		\label{Alg:PIRD}
		\LinesNumbered
		
		\textbf{Input:} 
		task set $\bm{\mathcal{S}}$;
		worker $w_j$'s candidate privacy budget range $[\varepsilon_j^{\min},\varepsilon_j^{\max}]$;
		budget step $\Delta\varepsilon$;
		IPM $f(\cdot\mid \cdot;\varepsilon)$;
		prior $\phi(\mathbf b_j)$;
		distortion metric $\Delta_j(\tilde{\mathbf b}_j,\mathbf b_j)$ in \eqref{eq:intent_distortion_rev};
		distortion cap $Q^{\mathrm{loss}}_{\max}$;
		(one-snapshot) inference-error threshold $\beta^0$.
		
		\textbf{Output:} calibrated privacy budget $\varepsilon_j^\star$.
		
		\BlankLine
		\textbf{Initialization:} $\varepsilon \leftarrow \varepsilon_j^{\min}$.
		
		\BlankLine
		\While{$\varepsilon \le \varepsilon_j^{\max}$}{
			Compute the expected intent-report distortion
			$Q^{\mathrm{loss}}_j(\phi,f,\Delta_j)$ using \eqref{eq:Qloss_intent_rev} with $f(\cdot\mid\cdot;\varepsilon)$;\\
			
			Compute the one-snapshot eIE $\xi_j$ using \eqref{eq:inference_error_intent_rev}
			and equivalently check \eqref{eq:h_lower_bound_intent_1_rev}--\eqref{eq:h_lower_bound_intent_2_rev};\\
			
			\If{$Q^{\mathrm{loss}}_j(\phi,f,\Delta_j)\le Q^{\mathrm{loss}}_{\max}$ \textbf{and} $\xi_j \ge \beta^0$}{
				\Return $\varepsilon_j^\star \leftarrow \varepsilon$;
			}
			$\varepsilon \leftarrow \varepsilon + \Delta\varepsilon$;
		}
		
		\BlankLine
		\Return \textsf{INFEASIBLE}; 
	}
\end{algorithm}

\vspace{0.25em}
\noindent
\textbf{Step 1. Personalized search initialization and design goal} (line 3, Alg. \ref{Alg:PIRD}):
For each worker $w_j$, we specify a candidate budget interval $[\varepsilon_j^{\min},\,\varepsilon_j^{\max}]$ and a step size $\Delta\varepsilon$,
and then initialize $\varepsilon\leftarrow \varepsilon_j^{\min}$.
The objective is to find a deployable privacy budget $\varepsilon_j^\star$ that simultaneously satisfies
\textit{(i)} the distortion cap $Q^{\mathrm{loss}}_{\max}$ and
\textit{(ii)} the one-snapshot inference-error requirement $\beta^0$.
Since $\varepsilon$ controls the privacy--utility trade-off, this step establishes a principled search region for balancing the two.

\noindent
\textbf{Step 2. IPM instantiation under a trial budget} (line 5, Alg. \ref{Alg:PIRD}):
Given the trial value $\varepsilon$, we instantiate IPM $f(\cdot\mid\cdot;\varepsilon)$,
which induces the perturbed intent distribution $\tilde{\mathbf b}_j\sim f(\cdot\mid \mathbf b_j;\varepsilon)$ under the prior $\phi(\mathbf b_j)$.
This instantiation enables evaluating both distortion and inference-risk metrics consistently under the same mechanism.

\noindent
\textbf{Step 3. Distortion-cap verification for decision usefulness} (line 5, Alg. \ref{Alg:PIRD}):
Under the instantiated $f$, we compute the expected intent-report distortion $Q^{\mathrm{loss}}_j(\phi,f,\Delta_j)$ according to \eqref{eq:Qloss_intent_rev}
(and the distortion metric $\Delta_j(\tilde{\mathbf b}_j,\mathbf b_j)$ in \eqref{eq:intent_distortion_rev}).
The trial budget is utility-feasible only if
\[
	Q^{\mathrm{loss}}_j(\phi,f,\Delta_j)\ \le\ Q^{\mathrm{loss}}_{\max},
\]
which ensures that the perturbed intent remains sufficiently informative for downstream pre-plan decisions.

\noindent
\textbf{Step 4. One-snapshot inference-robustness verification} (line 6, Alg. \ref{Alg:PIRD}):
For the same $\varepsilon$, we evaluate the adversary's one-snapshot eIE $\xi_j$ using \eqref{eq:inference_error_intent_rev}
and equivalently verify \eqref{eq:h_lower_bound_intent_1_rev}--\eqref{eq:h_lower_bound_intent_2_rev}.
The trial budget is privacy-feasible only if
\[
	\xi_j\ \ge\ \beta^0,
\]
meaning that a single observation of $\tilde{\mathbf b}_j$ keeps the adversary's expected estimation error above the required threshold.

\noindent
\textbf{Step 5. Feasibility decision and early stopping} (lines 7-8, Alg. \ref{Alg:PIRD}):
If both feasibility conditions hold, i.e., $Q^{\mathrm{loss}}_j(\phi,f,\Delta_j)\le Q^{\mathrm{loss}}_{\max}$ and $\xi_j\ge \beta^0$,
we accept the current budget and output $\varepsilon_j^\star\leftarrow \varepsilon$.
Since the scan proceeds from $\varepsilon_j^{\min}$ upward, the first feasible solution favors stronger privacy among the feasible candidates.

\noindent
\textbf{Step 6. Infeasibility handling via conservative opt-out} (lines 4--12, Alg.~\ref{Alg:PIRD}):
If no privacy budget within $[\varepsilon_j^{\min},\varepsilon_j^{\max}]$ satisfies
$Q^{\mathrm{loss}}_j(\phi,f,\Delta_j)\le Q^{\mathrm{loss}}_{\max}$ and $\xi_j\ge \beta^0$ simultaneously,
we return \textsf{INFEASIBLE}.
In this case, worker $w_j$ is not included in the offline candidate set for service pre-planning.
Thus, constraints \eqref{equ.25e_rev}--\eqref{equ.25f_rev} are ensured for all workers that participate in RAPCoD.

\subsubsection{MIRROR: Memoization-based personalized LDP perturbation for generating intent reports}
Given the calibrated budget from Alg.~\ref{Alg:PIRD}, we design MIRROR (see Alg.~\ref{Alg:MSIPM}), which specifies how each worker generates and reuses perturbed intent reports under repeated observations, with key steps below.
\begin{algorithm}[t!]
	{\scriptsize
		\caption{Proposed MIRROR}
		\label{Alg:MSIPM}
		\LinesNumbered
		
		\textbf{Input:}
		task set $\bm{\mathcal{S}}$;
		worker $w_j$'s true intent vector $\mathbf b_j=[b_{i,j}]_{s_i\in\bm{\mathcal{S}}}\in\{0,1\}^{|\bm{\mathcal{S}}|}$;
		calibrated privacy budget $\varepsilon_j^\star$ (from Alg.~\ref{Alg:PIRD});
		memo-epoch index $e$.
		
		\textbf{Output:}
		perturbed intent report $\tilde{\mathbf b}_j=[\tilde b_{i,j}]_{s_i\in\bm{\mathcal{S}}}$ to the platform.
		
		\BlankLine
		\textbf{// Permanent randomization to resist multi-snapshot frequency attacks}\\
		\If{$\tilde{\mathbf b}^{\mathrm{perm}}_j(e)$ is not stored}{
			\For{$\forall s_i\in\bm{\mathcal{S}}$}{
				Sample $\tilde b^{\mathrm{perm}}_{i,j}(e) \sim \textsc{RR}\!\left(b_{i,j},\varepsilon_j^\star\right)$;
			}
			Store $\tilde{\mathbf b}^{\mathrm{perm}}_j(e)$ locally for epoch $e$;
		}
		
		\BlankLine
		\textbf{// Report generation in each round within epoch $e$}\\
		Set $\tilde{\mathbf b}_j \leftarrow \tilde{\mathbf b}^{\mathrm{perm}}_j(e)$;
		
		\Return $\tilde{\mathbf b}_j$.
	}
\end{algorithm}

\vspace{0.25em}
\noindent
\textbf{Step 1. Personalized LDP instantiation from the calibrated budget} (line 1, Alg.~\ref{Alg:MSIPM}):
Each worker $w_j$ takes $\varepsilon_j^\star$ and instantiates an $\varepsilon_j^\star$-personalized LDP perturbation rule for each binary intent entry
(e.g., RR in \eqref{eq:RR_def}), which directly provides the personalized LDP guarantee in \eqref{eq:personalized_ldp_def_rev}.

\noindent
\textbf{Step 2. Epoch-level memoization trigger } (line 4, Alg.~\ref{Alg:MSIPM}):
We partition time into memo-epochs $e=1,2,\ldots$. Upon its first reporting instance in epoch $e$, worker $w_j$ verifies the existence of a locally cached epoch-specific permanent report $\tilde{\mathbf b}^{\mathrm{perm}}_j(e)$. If $\tilde{\mathbf b}^{\mathrm{perm}}_j(e)$ has not been previously generated, the worker applies permanent randomization to construct $\tilde{\mathbf b}^{\mathrm{perm}}_j(e)$ and stores it locally. Otherwise, it reuses the cached vector without modification.

\noindent
\textbf{Step 3. Permanent randomization to form an epoch-stable perturbed intent vector} (lines 5-7, Alg.~\ref{Alg:MSIPM}):
For each task $s_i\in\bm{\mathcal{S}}$, worker $w_j$ perturbs the true binary intent indicator $b_{i,j}\in\{0,1\}$
via personalized randomized response (RR)\cite{warner1965randomized} with budget $\varepsilon_j^\star$.
Specifically, for any $z\in\{0,1\}$ and privacy budget $\varepsilon>0$, $\textsc{RR}(z,\varepsilon)$ generates $\tilde z\in\{0,1\}$ according to
\begin{equation}
	\Pr(\tilde z=z)=\frac{e^{\varepsilon}}{e^{\varepsilon}+1},
	\qquad
	\Pr(\tilde z=1-z)=\frac{1}{e^{\varepsilon}+1}.
	\label{eq:RR_def}
\end{equation}
Accordingly, the epoch-$e$ permanent perturbed intent entry is sampled as
\begin{equation}
	\tilde b^{\mathrm{perm}}_{i,j}(e)\ \sim\ \textsc{RR}\!\left(b_{i,j},\,\varepsilon_j^\star\right),
	\qquad \forall s_i\in\bm{\mathcal{S}}.
\end{equation}
The resulting vector $\tilde{\mathbf b}^{\mathrm{perm}}_j(e)$ is then stored locally and treated as fixed throughout epoch $e$,
so that repeated reports within this epoch do not provide additional independent samples to an observer (thus suppressing frequency-based multi-snapshot attacks).

\noindent
\textbf{Step 4. Round-wise report reuse to suppress multi-snapshot frequency attacks} (lines 8-9, Alg.~\ref{Alg:MSIPM}):
In each interaction round $\tau$ within epoch $e$, worker $w_j$ reports
$\tilde{\mathbf b}^{(\tau)}_j \leftarrow \tilde{\mathbf b}^{\mathrm{perm}}_j(e)$.
This memoization strategy prevents repeated reports within the same memo-epoch from forming independent samples, thereby reducing the variance-reduction gain of multi-snapshot inference.

The platform treats $\tilde b_{i,j}$ only as an eligibility signal for constructing candidate sets and performs downstream pre-planning and pricing based on the perturbed intent vector $\tilde{\mathbf b}_j$. It does not require reconstructing the true intent vector $\mathbf b_j$.
Therefore, the downstream pipeline can be viewed as post-processing of $\tilde{\mathbf b}_j$. By the post-processing property of personalized LDP~\cite{dong2022gaussian}, platform-side quantities derived from perturbed reports and public parameters do not consume additional privacy budget beyond the calibrated $\varepsilon_j^\star$. Moreover, to accommodate slowly drifting intents in long-term participation and to limit long-horizon linkability, we partition time into memo-epochs and enforce epoch-level refresh: when entering epoch $e{+}1$, worker $w_j$ discards the stored permanent report $\tilde{\mathbf b}^{\mathrm{perm}}_j(e)$ and regenerates a new permanent vector $\tilde{\mathbf b}^{\mathrm{perm}}_j(e{+}1)$ by reapplying the same entry-wise randomized response rule with budget $\varepsilon_j^\star$. This refresh avoids indefinite reuse of the same permanent report, reduces cross-epoch linkability, and preserves the memoization benefit within each epoch against frequency-based profiling\cite{wang2016using}.

The formal proofs of \textit{(i)} the $\varepsilon_j$-personalized LDP guarantee of MIRROR, 
\textit{(ii)} the post-processing invariance for all platform-side outputs, and 
\textit{(iii)} the memoization-based robustness against multi-snapshot frequency attacks 
are provided in Appx.~C.

\subsection{Design of Module B}
\label{subsec:task_side_pg}
We now move to Module B and formulate the offline task-worker pre-planning problem as a task-dominated non-cooperative game,
where each task selects a subset of candidate workers under local budget, privacy, and risk constraints.
We then show that the resulting game is an \emph{exact potential game} on the \emph{joint feasible set}, guaranteeing the existence of at least one pure-strategy constrained NE and the finite improvement property (FIP)\cite{monderer1996potential,ding2021potential}.
Leveraging these properties, we develop \textit{ASPIRE-Off} (Alg.~\ref{Alg:TaskDPAsyncOff}), an asynchronous exact feasible best-response algorithm.
In each round, every task enumerates its unilateral feasible strategy space under the current strategies of the other tasks, evaluates the exact offline potential of each feasible unilateral deviation, and reports its best feasible response.
Since the game is an exact potential game over a finite feasible set and each accepted update is an exact feasible best-response improvement, ASPIRE-Off terminates in finite steps at a constrained NE, as detailed after Definition~\ref{def:NE_off} and Theorem~\ref{thm:task_exact_potential_en}.

\subsubsection{Game formulation}
As noted, each task $s_i$ seeks high redundancy-aware sensing quality under its budget and risk boundaries.
However, tasks are coupled through a shared worker pool. Because each worker can be included in the pre-plan of at most one task, the long-term contracting decision of a given task induces an immediate externality on others by restricting their feasible worker sets and limiting their attainable service quality. Hence, we formulate the task-side interaction in the offline stage as a non-cooperative game\cite{wang2022cooperative}:
\begin{equation}
	\mathcal{G}^{\mathrm{off}}
	=
	\Big(
	\bm{\mathcal{S}},\
	\{\mathcal{A}_i\}_{s_i\in\bm{\mathcal{S}}},\
	\{U_i\}_{s_i\in\bm{\mathcal{S}}}
	\Big),
	\label{eq:Game_task_off_def_en}
\end{equation}
where core elements are detailed below:

\vspace{0.2em}
\noindent$\bullet$ \textbf{Players.}
Each sensing task is represented by a platform-side virtual player.
Each task $s_i\in\bm{\mathcal{S}}$ selects a set of workers to form its offline plan.

\vspace{0.2em}
\noindent$\bullet$ \textbf{Worker-side acceptance filtering (non-strategic).}
Workers do not participate in the strategy updating process.
Whether a worker can be considered by a task is determined by the following acceptance screening:
\emph{(i) perturbed intent visibility}: $\tilde b_{i,j}=1$; and
\emph{(ii) individual rationality}: $p_{i,j}\ge c^\mathrm{W}_{i,j}$.
Accordingly, the candidate worker set of task $s_i$ is
\begin{equation}
	\tilde{\bm{\mathcal{W}}}_i
	\triangleq
	\Big\{
	w_j\in\bm{\mathcal{W}}\ \big|\ 
	\tilde b_{i,j}=1,\ p_{i,j}\ge c^\mathrm{W}_{i,j},\ \varepsilon_j^\star \neq \textsf{INFEASIBLE}
	\Big\}.
	\label{eq:candidate_workers_task_en}
\end{equation}
Namely, $s_i$ can select workers only from $\tilde{\bm{\mathcal{W}}}_i$.

\vspace{0.2em}
\noindent$\bullet$ \textbf{Strategy space.}
The strategy of task $s_i$ is defined as choosing a pre-plan set from its candidate worker pool, i.e., selecting a subset of eligible workers:
\begin{equation}
	a_i \in \mathcal{A}_i
	\triangleq
	\Big\{
	\bm{\mathcal{W}}_i \subseteq \tilde{\bm{\mathcal{W}}}_i\
	\big|\
	\sum_{w_j\in \bm{\mathcal{W}}_i} p_{i,j}\le B_i
	\Big\},
	\label{eq:Ai_task_en}
\end{equation}
where the budget feasibility is embedded into $\mathcal{A}_i$.
Denote the joint strategy by $\bm{a}=(a_1,\ldots,a_{|\bm{\mathcal{S}}|})$.
It induces the binary selection variable
\begin{equation}
	x_{i,j}=y_{i,j}(\bm{a})
	=
	\mathbbm{1}\{w_j\in a_i\},
	\qquad
	\bm{y}(\bm{a})=[y_{i,j}(\bm{a})].
	\label{eq:y_from_a_task_en}
\end{equation}

\vspace{0.2em}
\noindent$\bullet$ \textbf{Feasible set.}
Due to the exclusivity constraint that each worker can be pre-planned by at most one task, together with the system-wide risk constraints, the task strategies are coupled through shared feasibility requirements. We thus define the joint feasible strategy set~\cite{huang2024distributed} as
\begin{equation}
	\mathcal{A}^{\mathrm{feas}}
	\triangleq
	\Big\{
	\bm{a}\in\prod_{s_i\in\bm{\mathcal{S}}}\mathcal{A}_i
	\ \Big|\
	\sum_{s_i\in\bm{\mathcal{S}}} y_{i,j}(\bm{a})\le 1,\ 
	\eqref{equ.25g_rev},\eqref{equ.25h_rev}\ \text{hold}
	\Big\},
	\label{eq:A_feas_task_en}
\end{equation}
where the worker-exclusivity condition holds for all $w_j\in\bm{\mathcal{W}}$.
Moreover, $\prod_{s_i\in\bm{\mathcal{S}}}\mathcal{A}_i$ imposes local task-side feasibility, while worker exclusivity and risk constraints define the joint safety boundary.

\vspace{0.2em}
\noindent$\bullet$ \textbf{Offline potential function.}
To characterize the global effect of an offline pre-plan under the stochastic arrival $\bm{\alpha}$,
we define the offline potential function using eSW:
\begin{equation}
	\Phi^{\mathrm{off}}(\bm{a})
	\triangleq
	\mathbb{E}_{\bm{\alpha}}\!\left[
	\mathbb{SW}\big(\bm{y}(\bm{a}),\bm{\alpha}\big)
	\right].
	\label{eq:Phi_off_task_en}
\end{equation}
Note that payments cancel out when aggregating task and worker utilities in $\mathbb{SW}$; thus $\Phi^{\mathrm{off}}$ can be equivalently written in a ``quality gain minus execution cost'' form.
Moreover, since worker arrivals are modeled as independent Bernoulli random variables, $\Phi^{\mathrm{off}}$ and the quality-violation risk can be evaluated exactly over the finite arrival-state space.
For a given profile $\bm a$, let
\[
	\bm{\mathcal{W}}(\bm a)
	\triangleq
	\Big\{
	w_j\in\bm{\mathcal{W}}\ \big|\ \sum_{s_i\in\bm{\mathcal{S}}}y_{i,j}(\bm a)=1
	\Big\}
\]
be the set of pre-planned workers, and let
$\Omega(\bm a)=\{0,1\}^{|\bm{\mathcal{W}}(\bm a)|}$ denote all possible arrival realizations of these workers.
Then,
\begin{equation}
	\Phi^{\mathrm{off}}(\bm a)
	=
	\sum_{\bm\alpha\in\Omega(\bm a)}
	\Pr(\bm\alpha)\,
	\mathbb{SW}\big(\bm y(\bm a),\bm\alpha\big),
	\label{eq:Phi_off_exact_eval}
\end{equation}
where
\begin{equation}
	\Pr(\bm\alpha)
	=
	\prod_{w_j\in\bm{\mathcal{W}}(\bm a)}
	\pi_j^{\alpha_j}(1-\pi_j)^{1-\alpha_j}.
	\label{eq:arrival_prob_exact}
\end{equation}
Similarly, for each task $s_i$, $R_i^{\mathrm{Qual}}$ is computed by summing the probabilities of all arrival realizations under which the realized redundancy-aware quality is below $Q_i^{\mathrm{D}}$.
Therefore, the offline potential and risk constraints used by ASPIRE-Off are evaluated exactly rather than through Monte-Carlo sampling.

\vspace{0.2em}
\noindent$\bullet$ \textbf{Task payoff: marginal contribution.}
To strictly align unilateral task updates with the global objective, we define each task's payoff as its \emph{marginal contribution}\cite{konda2024optimal} to the potential, as given by:
\begin{equation}
	U_i(\bm{a})
	\triangleq
	\Phi^{\mathrm{off}}(\bm{a})
	-
	\Phi^{\mathrm{off}}(\varnothing,\bm{a}_{-i}),
	\qquad \forall s_i\in\bm{\mathcal{S}},
	\label{eq:Ui_marginal_task_en}
\end{equation}
where $(\varnothing,\bm{a}_{-i})$ means that $s_i$ selects an empty set while others keep unchanged.
This construction allows each task-side virtual player to act as a distributed decision agent, so that local unilateral updates are exactly aligned with global eSW improvement.
Besides, we define the constrained NE as follows:
\begin{Defn}[Constrained NE]\label{def:NE_off}
	A joint strategy $\bm{a}^\star\in\mathcal{A}^{\mathrm{feas}}$ is a (pure-strategy) constrained NE if for any task $s_i\in\bm{\mathcal{S}}$,
	\begin{equation}
		U_i(a_i^\star,\bm{a}_{-i}^\star)
		\ge
		U_i(a_i,\bm{a}_{-i}^\star),
		\quad
		\forall a_i\ \text{s.t.}\ (a_i,\bm{a}_{-i}^\star)\in\mathcal{A}^{\mathrm{feas}}.
		\label{eq:NE_task_def_cn_en}
	\end{equation}
\end{Defn}

When an NE is achieved, no individual task can further improve its marginal contribution payoff by unilaterally changing its pre-plan set within the joint feasible region, and thus the offline solution is stable.

\subsubsection{Potential property and existence of NE}

We show that the above task-side game forms an exact potential game on $\mathcal{A}^{\mathrm{feas}}$, which implies the existence of a pure-strategy NE and the FIP.

\begin{Defn}[Exact potential game]
	On the joint feasible set $\mathcal{A}^{\mathrm{feas}}$, if there exists a function $P(\bm{a})$ such that for any task $s_i$ and any two feasible joint strategies
	$\bm{a}=(a_i,\bm{a}_{-i})\in\mathcal{A}^{\mathrm{feas}}$ and
	$\bm{a}'=(a_i',\bm{a}_{-i})\in\mathcal{A}^{\mathrm{feas}}$,
	\begin{equation}
		U_i(\bm{a}')-U_i(\bm{a})
		=
		P(\bm{a}')-P(\bm{a}),
		\label{eq:exact_potential_task_def_en}
	\end{equation}
	then the game is an exact potential game and $P(\bm{a})$ is a potential function.
\end{Defn}

\begin{thm}[Potential structure and existence of NE]
	\label{thm:task_exact_potential_en}
	Given a nonempty $\mathcal{A}^{\mathrm{feas}}$, game $\mathcal{G}^{\mathrm{off}}$ is an exact potential game with potential function
	\begin{equation}
		P(\bm{a})=\Phi^{\mathrm{off}}(\bm{a}).
		\label{eq:P_equals_Phi_task_en}
	\end{equation}
	Hence, the game admits at least one pure-strategy constrained NE and satisfies the FIP\cite{zhu2022nash}.
\end{thm}

\begin{proof}
	For given task $s_i$ and any two feasible joint strategies
	$\bm{a}=(a_i,\bm{a}_{-i})\in\mathcal{A}^{\mathrm{feas}}$ and
	$\bm{a}'=(a_i',\bm{a}_{-i})\in\mathcal{A}^{\mathrm{feas}}$.
	Employing \eqref{eq:Ui_marginal_task_en}, we have
	\begin{equation}
		\begin{aligned}
			U_i(\bm{a}')&-U_i(\bm{a})
			=
			\Big(\Phi^{\mathrm{off}}(\bm{a}')-\Phi^{\mathrm{off}}(\varnothing,\bm{a}_{-i})\Big)
			\\&-
			\Big(\Phi^{\mathrm{off}}(\bm{a})-\Phi^{\mathrm{off}}(\varnothing,\bm{a}_{-i})\Big)\\
			&=
			\Phi^{\mathrm{off}}(\bm{a}')-\Phi^{\mathrm{off}}(\bm{a})
			=
			P(\bm{a}')-P(\bm{a}),
		\end{aligned}
	\end{equation}
	which verifies \eqref{eq:exact_potential_task_def_en}. Therefore, the game is an exact potential game on $\mathcal{A}^{\mathrm{feas}}$ with $P=\Phi^{\mathrm{off}}$.
	
	Since the joint feasible set $\mathcal{A}^{\mathrm{feas}}$ is finite and nonempty, the potential function admits at least one maximizer, which is a pure-strategy constrained NE. Under any asynchronous feasible better-response update, the potential function $P$ strictly increases, and thus the convergence to an NE can be guaranteed within a finite number of steps, implying the FIP.
\end{proof}
The above structure establishes an exact alignment between unilateral feasible task improvements and the increase of the global eSW. This alignment naturally leads to a low-overhead distributed solution via potential-driven asynchronous feasible updates.

\subsubsection{ASPIRE-Off: Asynchronous self-organized potential improvement for offline pre-planning}
By Theorem~\ref{thm:task_exact_potential_en}, asynchronous feasible best-response updates with exact potential evaluation converge to a constrained NE in a finite number of iterations.
The pseudo-code of ASPIRE-Off is provided in Alg.~\ref{Alg:TaskDPAsyncOff}, where each task enumerates all feasible unilateral deviations under the current strategies of the other tasks and selects the one with the largest exact potential value.
\begin{algorithm}[t!]
	{\scriptsize \setstretch{0.8}
		\caption{Proposed ASPIRE-Off}
		\label{Alg:TaskDPAsyncOff}
		\LinesNumbered
		
		\textbf{Input:}
		task set $\bm{\mathcal{S}}$, worker set $\bm{\mathcal{W}}$, candidate sets $\{\tilde{\bm{\mathcal{W}}}_i\}$;
		given payment menu $\bm p=[p_{i,j}]$, budgets $\{B_i\}$;
		reliability parameters $\{\pi_j\}$, quality parameters $\{q_{i,j}\}$, redundancy factors $\{\zeta_i\}$.
		
		\textbf{Initialization:}
		$t\leftarrow 0$; initialize a feasible pre-plan profile $\bm a(0)=\{a_i(0)\}\in\mathcal{A}^{\mathrm{feas}}$ using a feasibility-repair initialization;
		set $y_{i,j}(0)=\mathbbm{1}\{w_j\in a_i(0)\}$;
		$\bm{\mathcal{W}}^{\mathrm{idle}}(0)\leftarrow \{w_j\in\bm{\mathcal{W}}\mid \sum_{s_i\in\bm{\mathcal{S}}} y_{i,j}(0)=0\}$.
		
		\While{\textbf{true}}{
			Platform broadcasts the current profile $\bm a(t)$, equivalently $\bm y(t)$, and the idle worker set $\bm{\mathcal{W}}^{\mathrm{idle}}(t)$.
			
			\For{$\forall s_i\in\bm{\mathcal{S}}$}{
				$\bm{\mathcal{W}}^{\mathrm{avail}}_i(t)
				\leftarrow
				\big(a_i(t)\cup \bm{\mathcal{W}}^{\mathrm{idle}}(t)\big)\cap \tilde{\bm{\mathcal{W}}}_i$.
				
				Construct the exact unilateral feasible set
				\[
				\mathcal{A}^{\mathrm{uni}}_i(t)
				\triangleq
				\Big\{
				a_i'\subseteq \bm{\mathcal{W}}^{\mathrm{avail}}_i(t)
				\ \Big|\
				(a_i',\bm a_{-i}(t))\in\mathcal{A}^{\mathrm{feas}}
				\Big\}.
				\]
				
				Find the exact feasible best response
				\[
				a_i^{\mathrm{br}}(t)
				\in
				\arg\max_{a_i'\in\mathcal{A}^{\mathrm{uni}}_i(t)}
				\Phi^{\mathrm{off}}\big(a_i',\bm a_{-i}(t)\big).
				\]
				
				Compute the exact potential gain
				\[
				\Delta_i(t)
				\leftarrow
				\Phi^{\mathrm{off}}\big(a_i^{\mathrm{br}}(t),\bm a_{-i}(t)\big)
				-
				\Phi^{\mathrm{off}}\big(\bm a(t)\big).
				\]
			}
			
			$\mathcal{I}(t)\leftarrow \{i\in\bm{\mathcal{S}}\mid \Delta_i(t)>0\}$.
			
			\If{$\mathcal{I}(t)=\varnothing$}{
				\textbf{break}
			}
			
			Select one task index
			$i^\star\in\arg\max_{i\in\mathcal{I}(t)}\Delta_i(t)$.
			
			$a_{i^\star}(t+1)\leftarrow a_{i^\star}^{\mathrm{br}}(t)$.
			
			\For{$\forall i\in\bm{\mathcal{S}},\, i\neq i^\star$}{
				$a_i(t+1)\leftarrow a_i(t)$;
			}
			
			Update $y_{i,j}(t+1)=\mathbbm{1}\{w_j\in a_i(t+1)\}$.
			
			$\bm{\mathcal{W}}^{\mathrm{idle}}(t+1)
			\leftarrow
			\{w_j\in\bm{\mathcal{W}}\mid \sum_{s_i\in\bm{\mathcal{S}}} y_{i,j}(t+1)=0\}$.
			
			$t\leftarrow t+1$.
		}
		
		\textbf{Return:} $\bm a^\star=\bm a(t)$ and the induced pre-plan decision $\bm y(\bm a^\star)$.
	}
\end{algorithm}

\noindent\textbf{Step 1. Candidate construction and feasible initialization} (line 1, Alg.~\ref{Alg:TaskDPAsyncOff}):
During the offline stage, workers submit perturbed intent indicators $\tilde b_{i,j}$ via the personalized LDP mechanism.
By combining the perturbed-intent acceptance rule with individual-rationality screening, the platform constructs the candidate worker set $\tilde{\bm{\mathcal{W}}}_i$ for each task according to \eqref{eq:candidate_workers_task_en}.
The platform then initializes a feasible pre-plan profile $\bm a(0)=\{a_i(0)\}\in\mathcal{A}^{\mathrm{feas}}$, for example by using a feasibility-repair initialization.
Based on $\bm a(0)$, it derives the initial binary assignment map
$y_{i,j}(0)=\mathbbm{1}\{w_j\in a_i(0)\}$
and the initial idle-worker set
$\bm{\mathcal{W}}^{\mathrm{idle}}(0)=\{w_j\in\bm{\mathcal{W}}\mid \sum_{s_i\in\bm{\mathcal{S}}}y_{i,j}(0)=0\}$.

\vspace{0.35em}
\noindent\textbf{Step 2. Round-wise broadcast and available-set formation} (lines 3--5, Alg.~\ref{Alg:TaskDPAsyncOff}):
We discretize the offline improvement process into rounds $t=0,1,2,\ldots$.
At the beginning of each round, the platform broadcasts the current pre-plan profile $\bm a(t)$, equivalently $\bm y(t)$, and the idle-worker set $\bm{\mathcal{W}}^{\mathrm{idle}}(t)$.
Given these aggregate states, each task $s_i$ forms its available worker set as
\begin{equation}
	\bm{\mathcal{W}}^{\mathrm{avail}}_i(t)
=
\big(a_i(t)\cup\bm{\mathcal{W}}^{\mathrm{idle}}(t)\big)\cap \tilde{\bm{\mathcal{W}}}_i.
\end{equation}
This definition allows $s_i$ to retain its currently selected workers or switch to other available workers, while preventing it from selecting workers already occupied by other tasks.

\vspace{0.35em}
\noindent\textbf{Step 3. Exact unilateral feasible-set enumeration} (lines 6--8, Alg.~\ref{Alg:TaskDPAsyncOff}):
Instead of generating only one heuristic candidate set, ASPIRE-Off enumerates the full unilateral feasible strategy set of each task under the current strategies of the other tasks:
\begin{equation}
	\mathcal{A}^{\mathrm{uni}}_i(t)
	\triangleq
	\Big\{
	a_i'\subseteq \bm{\mathcal{W}}^{\mathrm{avail}}_i(t)
	\ \Big|\
	(a_i',\bm a_{-i}(t))\in\mathcal{A}^{\mathrm{feas}}
	\Big\}.
	\label{eq:offline_unilateral_feasible_set}
\end{equation}
This set contains all feasible unilateral deviations of task $s_i$ that satisfy local budget feasibility, worker exclusivity, quality-violation risk, and intent-mismatch risk under the current profile of other tasks.

\vspace{0.35em}
\noindent\textbf{Step 4. Exact best-response selection and potential-gain evaluation} (lines 9--13, Alg.~\ref{Alg:TaskDPAsyncOff}):
For each task $s_i$, ASPIRE-Off selects the exact feasible best response by solving
\begin{equation}
	a_i^{\mathrm{br}}(t)
	\in
	\arg\max_{a_i'\in\mathcal{A}^{\mathrm{uni}}_i(t)}
	\Phi^{\mathrm{off}}\big(a_i',\bm a_{-i}(t)\big).
	\label{eq:offline_exact_best_response}
\end{equation}
The corresponding exact potential gain is then computed as
\begin{equation}
	\Delta_i(t)
	=
	\Phi^{\mathrm{off}}\big(a_i^{\mathrm{br}}(t),\bm a_{-i}(t)\big)
	-
	\Phi^{\mathrm{off}}\big(\bm a(t)\big).
	\label{eq:offline_exact_gain}
\end{equation}
Both $\Phi^{\mathrm{off}}\big(a_i^{\mathrm{br}}(t),\bm a_{-i}(t)\big)$ and $\Phi^{\mathrm{off}}\big(\bm a(t)\big)$ are evaluated exactly according to \eqref{eq:Phi_off_exact_eval}--\eqref{eq:arrival_prob_exact}, rather than estimated through Monte-Carlo sampling.

\vspace{0.35em}
\noindent\textbf{Step 5. Asynchronous best-response update} (lines 14--22, Alg.~\ref{Alg:TaskDPAsyncOff}):
Let
\begin{equation}
	\mathcal{I}(t)
	=
	\{i\in\bm{\mathcal{S}}\mid \Delta_i(t)>0\}
\end{equation}
denote the set of tasks that admit a strictly improving feasible best response in round $t$.
If $\mathcal{I}(t)\neq\varnothing$, the platform selects one task
$i^\star\in\arg\max_{i\in\mathcal{I}(t)}\Delta_i(t)$
and only allows this task to update:
\begin{equation}
	a_{i^\star}(t+1)\leftarrow a_{i^\star}^{\mathrm{br}}(t),
	\qquad
	a_i(t+1)\leftarrow a_i(t),\ \forall i\neq i^\star.
\end{equation}
The platform then refreshes $\bm y(t+1)$ and $\bm{\mathcal{W}}^{\mathrm{idle}}(t+1)$ accordingly.
Because only one feasible unilateral update is committed in each round, worker exclusivity and all risk constraints are preserved.

\vspace{0.35em}
\noindent\textbf{Step 6. Termination and constrained-NE output} (line 23, Alg.~\ref{Alg:TaskDPAsyncOff}):
If $\mathcal{I}(t)=\varnothing$, no task can improve the exact potential by any feasible unilateral deviation.
By the exact potential property in Theorem~\ref{thm:task_exact_potential_en}, this is equivalent to saying that no task can improve its marginal-contribution payoff by unilaterally changing its strategy within $\mathcal{A}^{\mathrm{feas}}$.
Therefore, the terminal profile $\bm a^\star=\bm a(t)$ is a constrained NE in the sense of Definition~\ref{def:NE_off}, and the induced pre-plan decision $\bm y(\bm a^\star)$ is feasible and deliverable by construction.
Since $\mathcal{A}^{\mathrm{feas}}$ is finite and each accepted update strictly increases $\Phi^{\mathrm{off}}$, ASPIRE-Off terminates in a finite number of rounds.

Due to space limitations, the formal proofs of the main properties of ASPIRE-Off are provided in Appx.~C.

\section{Online Stage: Lightweight Quality Remediation via Temporary Recruitment}
\label{subsec:stage2_online_spot}
The offline pre-plan provides the baseline execution blueprint for online service provision. Nevertheless, MCS dynamics, such as worker dropout, mobility perturbations, and environmental fluctuations, may render some long-term contracted workers unavailable, thereby inducing realized sensing-quality deficits. Accordingly, we activate a temporary-recruitment potential game over idle workers to mitigate these deficits with bounded interaction rounds and low communication overhead.

For the online stage, let $\bm{\mathcal{W}}^{\mathrm{on}}$ denote the set of idle or standby workers that are available for temporary recruitment, and let $\bm{\mathcal{S}}^{\mathrm{on}}$ denote the set of tasks with unmet quality demands, i.e., tasks whose realized sensing quality remains below the required threshold.
We then characterize the practical utilities of workers and tasks in the online adjustment stage.

\noindent\emph{(i) Worker utility.}
The realized utility of worker $w_j$ in the online stage is
\begin{equation}
	u^\mathrm{W,\mathrm{on}}_j
	=
	\sum_{s_i\in\bm{\mathcal{S}}^{\mathrm{on}}}
	x^{\mathrm{on}}_{i,j}\big(p_{i,j}-c^\mathrm{W}_{i,j}\big),
	\label{eq:worker_utility_online_spot}
\end{equation}
where $x^{\mathrm{on}}_{i,j}\in\{0,1\}$ denotes the online execution/participation decision of $w_j$ for task $s_i$.

\noindent\emph{(ii) Task utility.}
By the redundancy-aware quality model in \eqref{eq:Qi_redundancy}, the incremental task-side utility contributed by temporary online recruitment is defined as
\begin{equation}
	u^\mathrm{S,\mathrm{on}}_i
	=
	\omega_3 \,
	\frac{\sum_{w_j\in\bm{\mathcal{W}}^{\mathrm{on}}} x^{\mathrm{on}}_{i,j}\, q_{i,j}}
	{1+\big(n^{\mathrm{on}}_i-1\big)\zeta_i}
	-
	\sum_{w_j\in\bm{\mathcal{W}}^{\mathrm{on}}} x^{\mathrm{on}}_{i,j}\, p_{i,j},
	\label{eq:task_utility_online_spot}
\end{equation}
where $n^{\mathrm{on}}_i \triangleq \sum_{w_j\in\bm{\mathcal{W}}^{\mathrm{on}}} x^{\mathrm{on}}_{i,j}$ is the realized number of participating (temporarily recruited) workers for task $s_i$ in the online stage.

\noindent\emph{(iii) Online social welfare.}
The online social welfare $\mathbb{SW}^{\mathrm{on}}$ measures the incremental welfare brought by temporary recruitment:
\begin{equation}
	\mathbb{SW}^{\mathrm{on}}
	=
	\sum_{s_i\in\bm{\mathcal{S}}^{\mathrm{on}}}
	u^{\mathrm{S,on}}_i
	+
	\sum_{w_j\in\bm{\mathcal{W}}^{\mathrm{on}}}
	u^{\mathrm{W,on}}_j .
	\label{eq:SW_online_spot}
\end{equation}

In the online stage, we aim to maximize the incremental online SW through low-overhead temporary recruitment, subject to perturbed-intent eligibility, budget feasibility, redundancy-aware quality-gap mitigation, and bounded interaction requirements. The detailed online formulation and its feasibility-preserving potential-improvement property are provided in Appx.~B.

\section{Evaluation}\label{sec:evaluation}
We conduct comprehensive evaluations to verify the effectiveness of our iParts.
Experiments are carried out via MATLAB R2025a.

\subsection{Key Parameters}\label{subsec:sim_setting}
Experiments are conducted on the real-world Chicago taxi trip dataset~\cite{qi2023matching,chicago_taxi_trips_2013}, which reports taxi trajectories from 2013 to 2016 over 77 community areas. Following the widely adopted MCS setup, we designate a single community area (e.g., Area 77) as the sensing region and extract the corresponding trajectory samples for market instantiation. Taxis are regarded as workers, and the worker pool $\bm{\mathcal{W}}$ is constructed by sampling taxis with sufficient historical traces. Worker availability is modeled by an arrival probability $\pi_j$, estimated from the historical presence frequency in a reference month, and the realized availability is sampled as $\alpha_j\sim \mathbf{B}(\pi_j)$. To emulate spatially uncertain sensing demands, task locations are uniformly generated within the sensing region.

To quantify $c^{\mathrm{exe}}_{i,j}$ in a data-driven manner, we extract three distance-related factors from the dataset:
\emph{(i)} the traveled distance of taxi ($w_j$) in the considered trip segment;
\emph{(ii)} the distance between the current location of $w_j$ and task location $loc_i^{\mathrm{S}}$;
and \emph{(iii)} the distance between the post-service location (after completing the trip segment) and $loc_i^{s}$.
We set $c^{\mathrm{exe}}_{i,j}$ to be proportional to a weighted sum of the above \textit{(i)-(iii)}.
The sensing quality $q_{i,j}$ is set inversely proportional to the service uncertainty, represented by the sum of factors \textit{(ii)} and \textit{(iii)},
which reflects that sensing quality generally degrades with increasing worker-task distance due to higher noise and latency.
Accordingly, key parameters are set as follows, following common settings in MCS studies~\cite{cai2025towards,qi2023matching,qi2026Accelerating}:
$\pi_j\in[0.56,0.96]$, $\theta_j\in[45,55]$,
$p_{i,j}\in[40,55]$, $B_i\in[200,250]$,
$Q_i^{D}\in[20,28]$, $\zeta_i\in[0.05,0.40]$,
$\mu_j\in[0.2,0.8]$, $\lambda_j\in[1,5]$,
$\varepsilon_j\in[0.1,5.0]$, $\omega_3=7$, and $\rho_1=\rho_2=0.2$.
Unless otherwise specified, ASPIRE-Off evaluates the offline potential and quality-violation risk exactly over the finite Bernoulli arrival-state space induced by each pre-plan profile. Each reported value is obtained by averaging over $100$ independent experiments with different random seeds (e.g., task locations, worker sampling, and Bernoulli arrivals).

\subsection{Benchmark Methods}
\label{subsec:baselines}

We compare iParts with five representative benchmarks that cover heuristic allocation, online re-optimization, learning-based allocation, privacy-preserving recruitment, and matching-based scheduling.

\noindent $\bullet$ \textbf{Myopic:}
A myopic online re-optimization method that reconstructs the task-worker allocation from scratch in each execution slot according to the currently active workers, unfinished tasks, remaining budgets, and realized quality deficits. This baseline follows the online MCS task-allocation setting in~\cite{zhang2025utility}.

\noindent $\bullet$ \textbf{GDRL:}
A graph-attention-based DRL allocation method adapted from~\cite{xu2023gdrl}. It encodes task-worker relationships as a graph and employs a Double-DQN-style policy to sequentially select feasible task-worker pairs under the same budget and feasibility filters.

\noindent $\bullet$ \textbf{DPCMAB:}
A differentially private combinatorial multi-armed bandit (CMAB) recruitment method adapted from~\cite{an2024privacy}. It perturbs sensing-quality or reward feedback and updates confidence-bound scores for worker recruitment. Note that DPCMAB protects recruitment feedback observations rather than workers' task-intent vectors, and thus serves as a privacy-preserving recruitment baseline rather than an intent-privacy baseline.

\noindent $\bullet$ \textbf{OSM:}
An online stable matching method adapted from matching-based dynamic MCS service trading schemes~\cite{qi2026Bank}. It constructs worker/task preference lists and applies many-to-one deferred acceptance under quota and budget constraints.

\noindent $\bullet$ \textbf{Greedy:}
A benefit-cost-ratio greedy method iteratively selects the feasible task-worker pair with the largest marginal benefit per unit cost. This method follows the budget-constrained marginal-density selection principle commonly adopted in MCS incentive and task-allocation mechanisms~\cite{tong2024privacy}.

We also use internal variants for ablation and privacy analysis: \textbf{NoP} removes intent privacy protection, \textbf{NoMem} removes memoization, \textbf{NoR} removes redundancy-aware quality modeling, and \textbf{NoOn} disables online remediation.

\subsection{Performance Metrics}
\label{subsec:metrics}

We adopt the following metrics to quantitatively evaluate welfare, online interaction efficiency, privacy robustness, and constraint satisfaction.

\noindent \textbf{\textit{(i)} Welfare/utility and interaction-efficiency metrics:}

\noindent $\bullet$ \textit{SW:}
SW measures the overall social welfare, obtained by aggregating the realized value of executing offline contracts in~\eqref{eq:SW_def} and the additional contribution of online recruitment in~\eqref{eq:SW_online_spot}. For scale-consistent comparison, we report the normalized value of SW in the figures. A larger SW indicates higher overall service efficiency under dynamic worker participation.

\noindent $\bullet$ \textit{TU:}
TU denotes task-side utility, capturing the effective service benefit and budget efficiency achieved by sensing tasks. A larger TU indicates that tasks obtain higher quality-aware returns under budget constraints.

\noindent $\bullet$ \textit{WU:}
WU denotes worker-side utility, reflecting workers' incentives and long-term willingness to participate. A larger WU indicates that the mechanism provides stronger participation motivation while satisfying individual rationality.

\noindent $\bullet$ \textit{INSW:}
To jointly characterize welfare gain and online interaction efficiency, we introduce interaction-normalized social welfare (INSW), defined as $\mathrm{INSW} \triangleq\frac{\mathrm{SW}}{1+\mathrm{NI}/|\mathcal{S}|}. $
Here, $\mathrm{NI}/|\mathcal{S}|$ denotes the average number of online interactions per task. A larger INSW means that the mechanism can convert limited online interactions into higher social welfare. For visualization only, we plot $\log_{10}(1+\mathrm{INSW})$ in the corresponding figure.

\noindent \textbf{\textit{(ii)} Online interaction/computational overhead metrics:}

\noindent $\bullet$ \textit{RT:}
RT describes the wall-clock running time for obtaining online decisions in each execution epoch under the same software/hardware configuration. It reflects the real-time computational efficiency of the online stage. The offline pre-planning stage is conducted before execution and can be amortized over the memoized planning horizon.

\noindent $\bullet$ \textit{NI:}
NI denotes the number of task-worker message exchanges during online decision-making, including online status querying, recruitment notification, and response collection. It captures both signaling overhead and the observable interaction surface during execution.

\noindent $\bullet$ \textit{IL:}
We model the uplink latency $t^{\mathrm{U}}_{i,j}$ from worker to platform/task within $[0.5,11]$ milliseconds~\cite{3GPP2022,qi2023matching}, and the downlink latency $t^{\mathrm{D}}_{i,j}$ from platform/task to worker within $[0.5,4]$ milliseconds~\cite{3GPP2022,qi2023matching}. Then, IL is computed as $\mathrm{IL} \triangleq \sum_i \sum_j N_{i,j}
\big(t^{\mathrm{U}}_{i,j}+t^{\mathrm{D}}_{i,j}\big), $
where $N_{i,j}$ denotes the number of online message exchanges involving pair $(s_i,w_j)$.

\noindent $\bullet$ \textit{IEC:}
Let $e^{W}_{j}$ and $e^{S}_{i}$ denote the transmission powers of workers and tasks/platform, respectively, where $e^{W}_{j}\in[0.2,0.4]$ Watts~\cite{yi2020Multi-User,qi2026Bank} and $e^{S}_{i}\in[6,20]$ Watts~\cite{yi2020Multi-User,qi2026Bank}. The energy consumption induced by online interactions is computed as $\mathrm{IEC} \triangleq \sum_i \sum_j N_{i,j}
\big(e^{S}_{i}t^{\mathrm{D}}_{i,j}+e^{W}_{j}t^{\mathrm{U}}_{i,j}\big). $
IEC characterizes the communication-energy overhead during online decision-making. For readability, RT, NI, IL, and IEC are plotted using $\log_{10}(1+x)$.

\noindent \textbf{\textit{(iii)} Intent-privacy evaluation:}

\noindent $\bullet$ \textit{AOeIE:}
AOeIE denotes the average one-snapshot intent estimation error. Specifically, we report the adversary's eIE $\xi_j$ in~\eqref{eq:inference_error_intent_rev}, averaged over evaluated workers with the distance metric in~\eqref{eq:distance_intent_rev}. A larger AOeIE indicates that the perturbed reports retain greater uncertainty against optimal Bayesian inference from a single observation.

\noindent $\bullet$ \textit{AOSR:}
AOSR denotes the average one-snapshot success rate. It measures the probability that the adversary correctly recovers workers' intent entries from one observation. Specifically, we calculate the entry-wise recovery accuracy as $\mathrm{OSR} \triangleq
\frac{1}{|\mathcal{S}||\mathcal{W}|}
\sum_j \sum_i \mathbbm{1}\{\hat b_{i,j}=b_{i,j}\}, $
and report its average value. A smaller AOSR indicates lower one-snapshot intent recoverability.

\noindent $\bullet$ \textit{AMFL:}
AMFL denotes the average multi-snapshot frequency deviation. Under repeated observations, the adversary computes the empirical frequency $F_{i,j}$ in~\eqref{eq:freq_stat_intent_rev}. We measure the deviation between $F_{i,j}$ and the true intent $b_{i,j}$ by $\mathrm{MFL} \triangleq
\frac{1}{|\mathcal{S}||\mathcal{W}|}
\sum_j \sum_i |F_{i,j}-b_{i,j}|. $
A larger AMFL means that long-term frequency statistics provide less reliable signals for profiling inference.

\noindent $\bullet$ \textit{FMSR and AMSR:}
We further implement a frequency-threshold profiling attacker, e.g.,~\eqref{eq:majority_inference_intent_rev}, and report its recovery accuracy over intent entries. FMSR denotes the final multi-snapshot success rate at the maximum number of snapshots, while AMSR denotes the average multi-snapshot success rate over all evaluated snapshot numbers. Smaller FMSR and AMSR indicate stronger resistance to long-term averaging attacks.

\noindent \textbf{\textit{(iv)} Individual rationality and risk controllability:}

\noindent $\bullet$ \textit{Task-side budget feasibility:}
This reports each task's realized total payment and assesses task-side budget feasibility in~\eqref{equ.25c_rev}.

\noindent $\bullet$ \textit{IR of workers:}
This reports workers' received payments and realized service costs and assesses the worker-side individual rationality constraint in~\eqref{equ.25b_rev}.

\noindent $\bullet$ \textit{Unsatisfactory-quality risk (QRisk):}
This reports the empirical/estimated probability that a task's realized quality violates the threshold, i.e., $R_i^\mathrm{Qual}$ in~\eqref{eq:R_qual}, directly reflecting reliability under dropout and redundancy discounting.

\noindent $\bullet$ \textit{Intent-mismatch risk (PRisk):}
This reports the intent-mismatch risk $R_j^\mathrm{Pref}$ in~\eqref{eq:R_pref}, capturing the probability of being assigned to undesired tasks due to perturbed intent reports. A lower PRisk implies closer alignment with true preferences and greater participation stability.

\subsection{Performance Evaluations}

\subsubsection{SW, TU, WU and INSW}
\label{sec:perform_utility}

\begin{figure}[t!]
	\centering
	\includegraphics[width=1\linewidth]{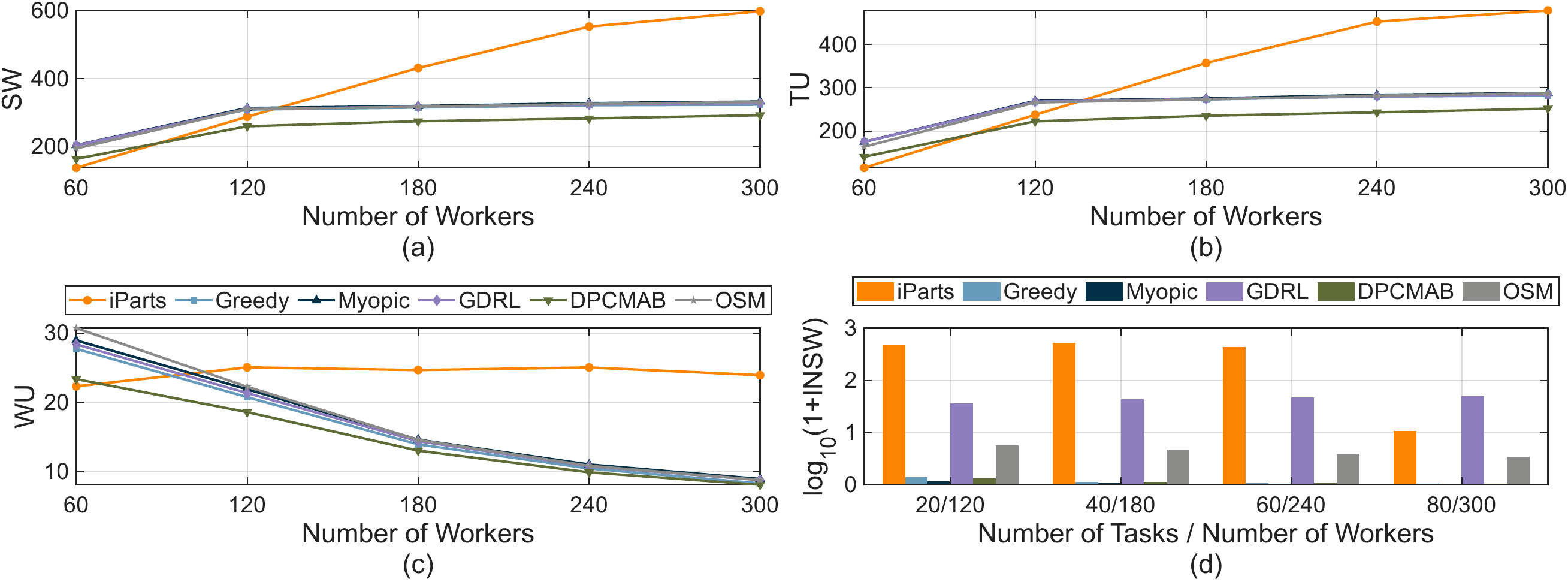}
	\caption{Performance comparison in terms of SW, TU, WU and INSW, where (a)--(c) set the number of tasks to 60.}
	\label{fig:SW_all}
	\vspace{-0.4cm}
\end{figure}

To evaluate economic performance and interaction-aware service efficiency, we consider four metrics in Fig.~\ref{fig:SW_all}: SW, TU, WU, and INSW. Figs.~\ref{fig:SW_all}(a)--\ref{fig:SW_all}(c) show the impact of the number of workers when the number of tasks is fixed to 60, while Fig.~\ref{fig:SW_all}(d) reports INSW under different task/worker scales. For better visualization, Fig.~\ref{fig:SW_all}(d) plots $\log_{10}(1+\mathrm{INSW})$.

As shown in Fig.~\ref{fig:SW_all}(a), iParts achieves a clear SW advantage when the worker population becomes sufficiently large. When the worker pool is small, iParts behaves relatively conservatively because intent perturbation and risk-aware pre-planning restrict overly aggressive assignments under limited worker availability. However, as more workers become available, iParts exploits the enlarged feasible worker pool more effectively and obtains substantially higher SW than Greedy, Myopic, GDRL, DPCMAB, and OSM. This indicates that iParts can better convert redundant worker availability into system-level welfare gains.

Fig.~\ref{fig:SW_all}(b) shows a similar trend for TU. iParts gradually outperforms all baselines as the number of workers increases, suggesting that the proposed risk-aware and redundancy-aware planning mechanism can improve task-side utility under richer worker resources. In contrast, Greedy and DPCMAB exhibit relatively limited TU growth because greedy benefit-cost selection and DP-protected recruitment do not explicitly coordinate long-term task quality, worker availability, and online remediation. Myopic, GDRL, and OSM achieve competitive TU in moderate settings, but their gains become saturated when the worker pool further increases.

Fig.~\ref{fig:SW_all}(c) compares WU under different worker populations. For most baselines, WU decreases as the number of workers increases, indicating stronger worker-side competition and less favorable utility distribution in denser markets. In contrast, iParts maintains a relatively stable and higher WU after the worker pool becomes sufficiently large. This is because iParts jointly considers worker-side participation utility, risk-aware service provisioning, and online remediation, instead of simply selecting workers according to immediate task-side gains.

Finally, Fig.~\ref{fig:SW_all}(d) evaluates INSW. iParts achieves high interaction-normalized welfare in most settings and remains competitive under the heaviest scale. This result is important because high SW can be obtained by frequent online querying, repeated re-optimization, or aggressive rematching. By contrast, INSW measures how efficiently each method converts limited online interactions into social welfare. The observed trend is consistent with the low-interaction two-stage design: offline pre-planning handles most regular service provisioning, while online remediation is triggered only when runtime deficits occur. Although GDRL shows competitive INSW in the largest setting, iParts consistently outperforms Greedy, Myopic, DPCMAB, and OSM in terms of interaction-aware welfare efficiency. Overall, Fig.~\ref{fig:SW_all} suggests that iParts achieves a favorable balance among system welfare, task utility, worker utility, and online interaction efficiency.

\subsubsection{Interaction overhead}
\label{sec:perform_io}

\begin{figure}[t!]
	\centering
	\includegraphics[width=1\linewidth]{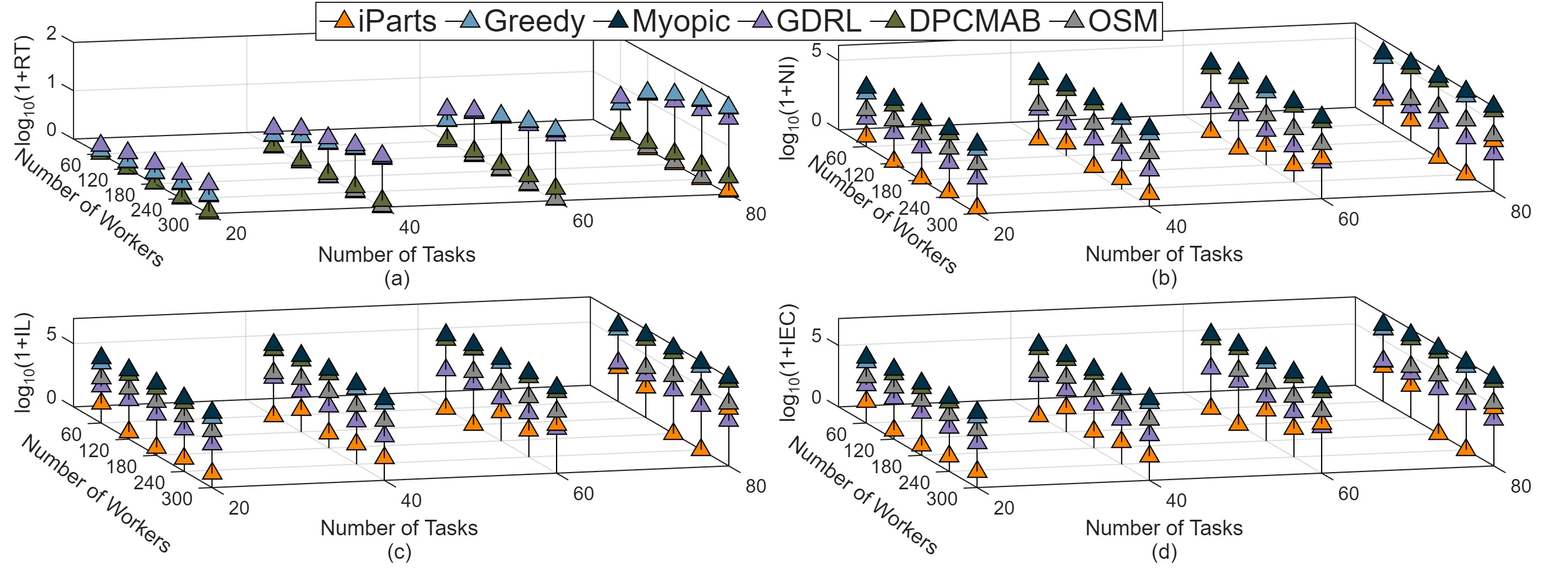}
	\caption{Performance comparisons in terms of online interaction overhead.}
	\label{fig:IO}
\end{figure}

To evaluate the real-time decision overhead and online interaction efficiency of iParts, we compare six methods in terms of RT, NI, IL, and IEC, as shown in Fig.~\ref{fig:IO}. RT denotes online running time, NI denotes the number of online interactions, IL denotes interaction latency, and IEC denotes interaction energy consumption. Since different methods may incur overheads at different orders of magnitude, each subfigure reports $\log_{10}(1+x)$ for better visualization. Fig.~\ref{fig:IO} focuses on online-stage overhead. The offline pre-planning of iParts is conducted before the sensing epoch and can be amortized over the memoized planning horizon; hence, it is not counted as repeated online interaction.

Fig.~\ref{fig:IO}(a) compares RT under different task/worker scales. iParts maintains low online running time across most settings because the major contract construction and risk-aware planning procedures have already been completed offline. During online execution, iParts only identifies realized quality gaps and activates lightweight temporary remediation among idle workers. In contrast, Myopic and OSM usually incur higher RT because they repeatedly perform online re-optimization or stable-matching updates as the market state changes. GDRL benefits from learned policy inference, but it still needs online state construction and decision execution.

Fig.~\ref{fig:IO}(b) presents NI, which directly measures online interactions. iParts consistently incurs fewer online interactions than most baselines because it avoids repeatedly querying all workers after each runtime change. Instead, the offline stage establishes risk-aware service provisioning decisions based on privacy-preserving intent reports, and the online stage only interacts with a limited candidate set when remediation is necessary. By comparison, Myopic and OSM require more frequent online coordination, such as state collection, preference updates, proposals, acceptances, and rematching operations. DPCMAB also requires continuous recruitment feedback and exploration-exploitation updates, increasing online communication.

Figs.~\ref{fig:IO}(c) and~\ref{fig:IO}(d) further report IL and IEC. Since latency and energy consumption are closely related to online interactions, their trends are generally consistent with NI. iParts achieves low IL and IEC in most task/worker settings, indicating that it not only reduces the number of online messages but also lowers the associated communication delay and energy cost. This advantage becomes more evident as the market scale grows, where fully online or frequently updated methods suffer from enlarged coordination spaces.

Overall, Fig.~\ref{fig:IO} illustrates the low-interaction property of iParts. By moving privacy-sensitive intent reporting, risk-aware pre-planning, and most regular service-provisioning decisions to the offline stage, iParts reduces the online coordination scope. The online stage is then limited to lightweight remediation for tasks with realized service deficits. Consequently, iParts achieves lower online RT, NI, IL, and IEC while preserving strong utility performance, suggesting that the proposed two-stage mechanism balances service quality, real-time efficiency, and interaction cost in dynamic MCS.

\subsubsection{Privacy protection and inference robustness}
\label{sec:perform_privacy}

\begin{table}[t]
	\centering
	\caption{Privacy resistance against one-snapshot and multi-snapshot intent inference attacks.}
	\label{tab:privacy_resistance}
	\footnotesize
	\setlength{\tabcolsep}{2pt}
	\renewcommand{\arraystretch}{1}
	\begin{tabular}{lccccc}
		\toprule
		Method & AOeIE $\uparrow$ & AOSR $\downarrow$ & AMFL $\uparrow$ & FMSR $\downarrow$ & AMSR $\downarrow$ \\
		\midrule
		iParts & $\mathbf{0.159 \pm 0.081}$ & $0.840 \pm 0.102$ & $\mathbf{0.165}$ & $\mathbf{0.834}$ & $\mathbf{0.835}$ \\
		NoP    & $0.000 \pm 0.000$ & $1.000 \pm 0.000$ & $0.000$ & $1.000$ & $1.000$ \\
		NoMem  & $0.158 \pm 0.080$ & $\mathbf{0.834 \pm 0.091}$ & $0.165$ & $1.000$ & $0.941$ \\
		\bottomrule
	\end{tabular}
\end{table}

To evaluate intent-privacy preservation, we examine iParts under both \emph{one-snapshot} and \emph{multi-snapshot} inference settings. AOeIE and AOSR evaluate the adversary's average inference error and recovery success rate under one-shot observation, where the inference process follows~\eqref{eq:inference_error_intent_rev}. AMFL, FMSR, and AMSR characterize long-term leakage under repeated observations, where the adversary exploits the empirical frequency in~\eqref{eq:freq_stat_intent_rev} and the frequency-threshold profiling rule in~\eqref{eq:majority_inference_intent_rev}. We compare iParts with two variants: NoP, which directly exposes true intents, and NoMem, which applies the same intent perturbation as iParts but removes memoization.

Table~\ref{tab:privacy_resistance} shows that NoP provides almost no protection. Its AOeIE is zero and its AOSR reaches one, indicating that the adversary can nearly perfectly recover workers' true intent entries once unperturbed reports are exposed. In contrast, iParts achieves a positive AOeIE of $0.159$ and reduces AOSR to $0.840$, suggesting that personalized LDP introduces uncertainty into observable intent reports. NoMem exhibits comparable one-snapshot behavior, with AOeIE of $0.158$ and AOSR of $0.834$. This is expected because iParts and NoMem use the same personalized perturbation mechanism for a single report. Hence, the one-snapshot results mainly support the role of LDP-based intent perturbation.

The advantage of memoization becomes evident under multi-snapshot inference. NoMem has AMFL close to that of iParts, indicating that AMFL alone captures the deviation of empirical frequencies but does not fully reflect the recoverability of true intents after threshold-based profiling. However, its FMSR reaches $1.000$ and its AMSR rises to $0.941$, meaning that independent per-round perturbations still allow the adversary to exploit frequency concentration and recover the underlying intent pattern through threshold-based profiling. By contrast, iParts maintains a much lower FMSR of $0.834$ and AMSR of $0.835$. This improvement is achieved because MIRROR reuses an epoch-stable perturbed intent vector within each memo-epoch, preventing repeated observations within a memo-epoch from becoming independent samples for frequency-based profiling.

Overall, Table~\ref{tab:privacy_resistance} summarizes the complementary roles of personalized LDP and memoization. Personalized LDP provides one-snapshot protection by perturbing intent reports, while memoization strengthens long-term protection by suppressing multi-snapshot frequency convergence. This supports the use of iParts in repeated participation scenarios, where long-term task-selection traces may otherwise expose workers' intent preferences.

\subsubsection{Individual rationality and risk analysis}
\label{sec:IR}

\begin{figure}[t]
	\centering
	\subfigtopskip=2pt
	\subfigbottomskip=10pt
	\subfigcapskip=-0.1cm
	\setlength{\abovecaptionskip}{-0.1cm}
	\hspace{-4.29mm}
	\subfigure[IR of tasks]{
		\begin{minipage}[t]{0.32\columnwidth}
			\centering
			\includegraphics[width=\linewidth, height=2.5cm]{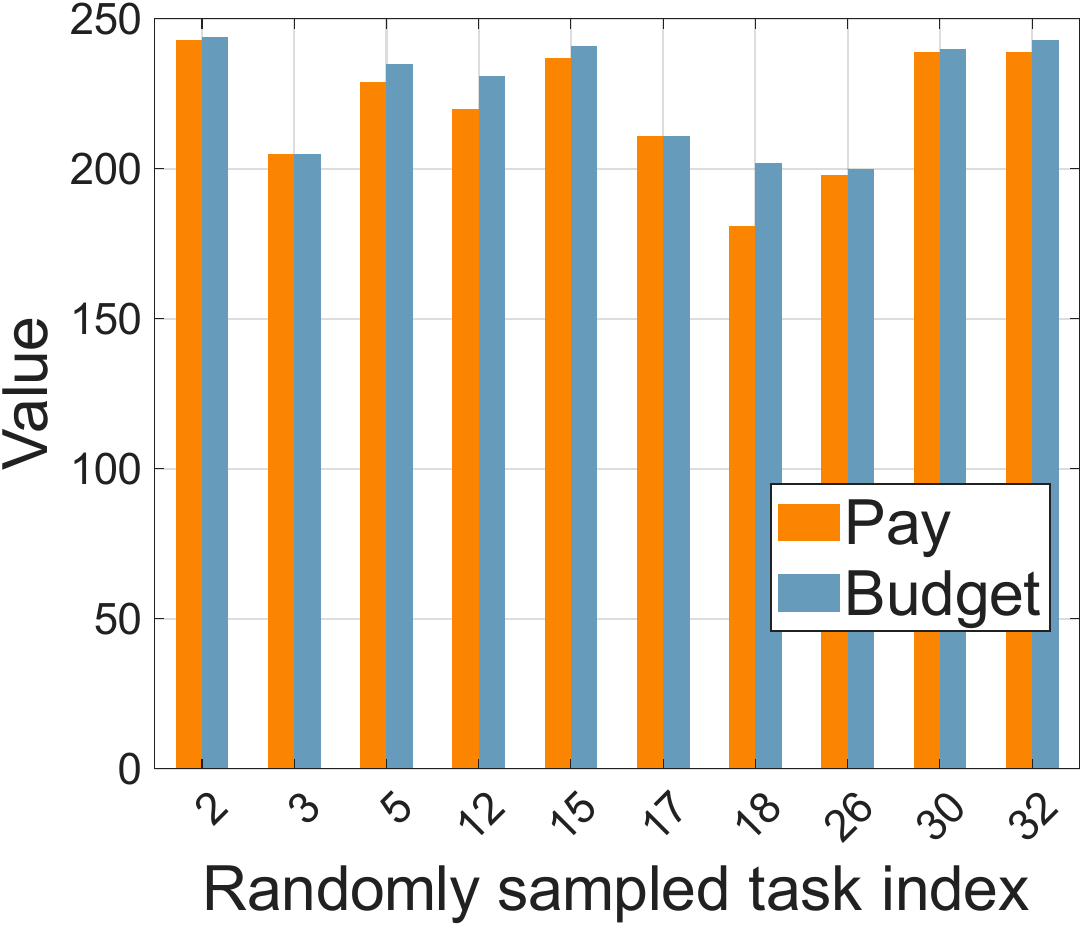}
			\label{fig:IR_tas}
	\end{minipage}}
	\subfigure[IR of workers]{
		\includegraphics[width=0.32\columnwidth, height=2.5cm]{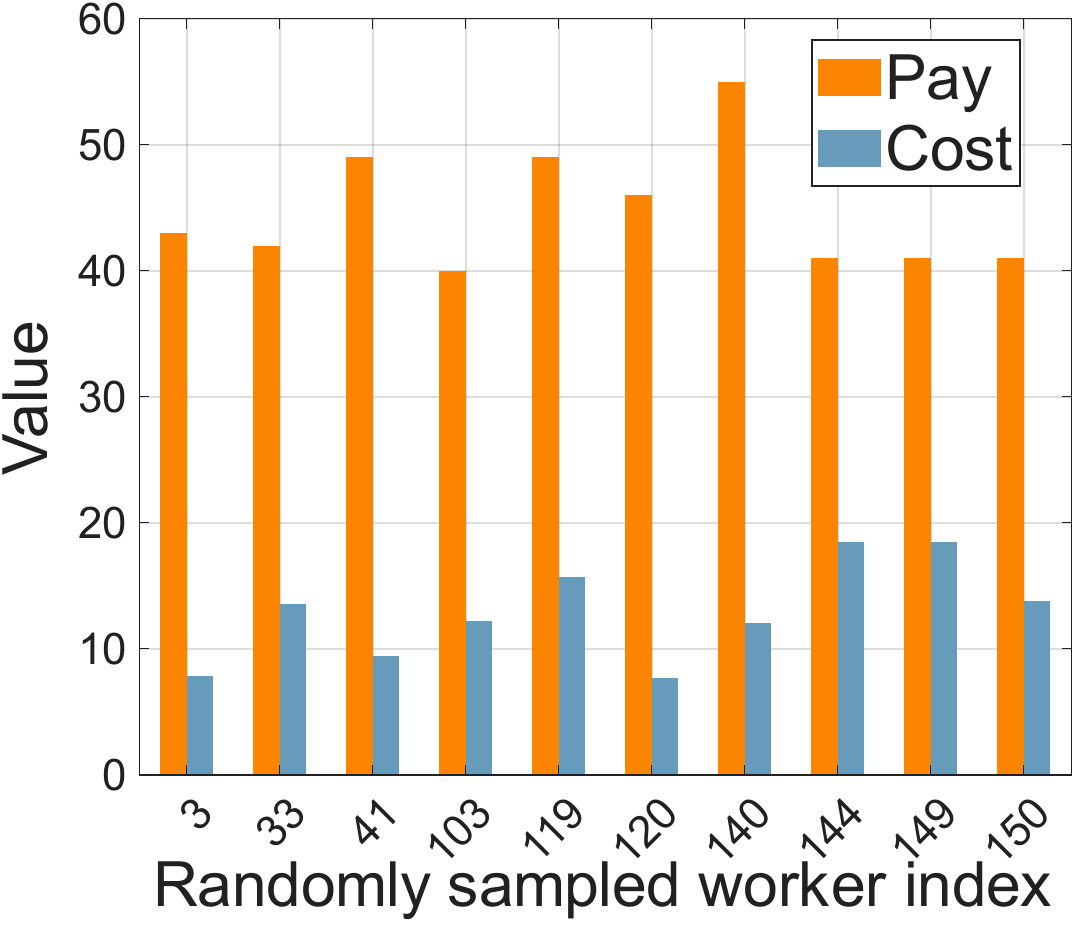}
		\label{fig:IR_worker}
	}
	\hspace{-1.29mm}
	\subfigure[Risk analysis]{
		\includegraphics[width=0.32\columnwidth, height=2.5cm]{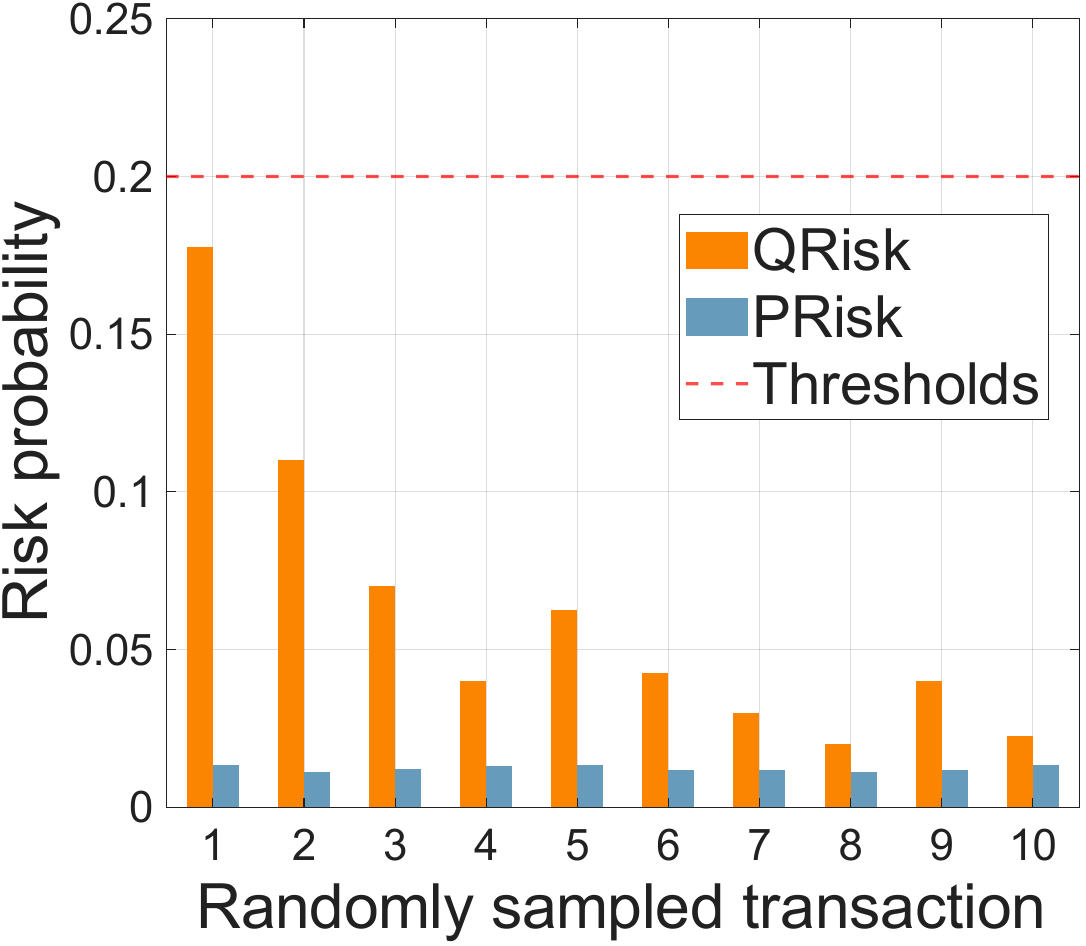}
		\label{fig:Risk}
	}
	\caption{IR and risk analysis of tasks and workers.}
	\label{fig:IR}
	\vspace{-0.4cm}
\end{figure}

To examine the practical deliverability of iParts, we evaluate individual rationality, budget feasibility, and risk controllability under a representative setting with $|\mathcal{W}|=160$ and $|\mathcal{S}|=40$.

Fig.~\ref{fig:IR_tas} reports task-side budget feasibility in~\eqref{equ.25c_rev}. For 10 randomly sampled tasks, the realized total payments never exceed their budgets, indicating strict budget compliance. Fig.~\ref{fig:IR_worker} reports worker-side individual rationality in~\eqref{equ.25b_rev}. For 10 sampled workers, payments always exceed realized costs, indicating non-negative utility after accounting for execution and privacy costs. Fig.~\ref{fig:Risk} examines the explicit risk constraints in~\eqref{equ.25g_rev}--\eqref{equ.25h_rev}. For 10 sampled transactions, both QRisk and PRisk remain below the threshold 0.2, consistent with the quality-violation and preference-mismatch risk bounds.

These results show that iParts satisfies worker-side IR, task-side budget feasibility, and dual risk controllability. Together, these results support the practical reliability and constraint-compliant execution of the proposed intent-preserving service provisioning mechanism.

\section{Conclusion}
This paper studied intent-preserving and efficiency-aware service provisioning in dynamic MCS under an honest-but-curious platform. We proposed \textit{iParts}, a two-stage framework that integrates personalized intent perturbation, redundancy-aware quality aggregation, risk-aware offline pre-planning, and lightweight online remediation. By shifting privacy-sensitive structural decisions to the offline stage and limiting online decisions to temporary recruitment, iParts coordinates task-worker services while reducing repeated preference exposure and interaction overhead. We further formulated the offline pre-planning problem as an exact potential game and established constrained-equilibrium and finite-convergence guarantees for offline feasible improvement dynamics under explicit privacy, budget, quality, and risk constraints. Trace-driven experiments showed that iParts improves service welfare and task reliability, reduces redundant sensing and online overhead, and lowers intent recoverability under one-snapshot and multi-snapshot inference. Future work will explore adaptive privacy-budget calibration and temporal intent protection under evolving worker preferences.

\begin{spacing}{0.99}
\bibliographystyle{ieeetr}
\bibliography{reference}

@article{wang2022bsif,
  title={BSIF: Blockchain-based secure, interactive, and fair mobile crowdsensing},
  author={Wang, Weizheng and Yang, Yaoqi and Yin, Zhimeng and Dev, Kapal and Zhou, Xiaokang and Li, Xingwang and Qureshi, Nawab Muhammad Faseeh and Su, Chunhua},
  journal={IEEE J. Sel. Areas Commun.},
  volume={40},
  number={12},
  pages={3452--3469},
  year={2022},
  publisher={IEEE}
}

@article{cai2025towards,
  title={Towards Personalized Location Privacy Trading for Mobile Crowd Sensing},
  author={Cai, Hui and Lan, Chen and Yang, Yuanyuan and Xiao, Fu and Zhu, Yanmin and Zhou, Jian and Sheng, Biyun},
  journal={IEEE Trans. Depend. Sec. Comput.},
  year={2025},
  publisher={IEEE}
}

@article{an2024privacy,
  title={Privacy-preserving user recruitment with sensing quality evaluation in mobile crowdsensing},
  author={An, Jieying and Ren, Yanbing and Li, Xinghua and Zhang, Man and Luo, Bin and Miao, Yinbin and Liu, Ximeng and Deng, Robert H},
  journal={IEEE Trans. Depend. Sec. Comput.},
  volume={22},
  number={1},
  pages={787--803},
  year={2025},
  publisher={IEEE}
}

@article{qi2026Accelerating,
  author={Qi, Houyi and Liwang, Minghui and Wang, Xianbin and Fu, Liqun and Hong, Yiguang and Li, Li and Cheng, Zhipeng},
  journal={IEEE Trans. Mobile Comput.}, 
  title={Accelerating Stable Matching Between Workers and Spatial-Temporal Tasks for Dynamic MCS: A Stagewise Service Trading Approach}, 
  year={2026},
  volume={25},
  number={2},
  pages={2878-2894},
  publisher={IEEE}
}

@article{xu2022personalized,
  title={Personalized location privacy protection for location-based services in vehicular networks},
  author={Xu, Chuan and Ding, Yingyi and Chen, Chao and Ding, Yong and Zhou, Wei and Wen, Sheng},
  journal={IEEE Trans. Intell. Transp. Syst.},
  volume={24},
  number={1},
  pages={1163--1177},
  year={2022},
  publisher={IEEE}
}

@article{xiao2020privacy,
  title={Privacy-preserving user recruitment protocol for mobile crowdsensing},
  author={Xiao, Mingjun and Gao, Guoju and Wu, Jie and Zhang, Sheng and Huang, Liusheng},
  journal={IEEE/ACM Trans. Netw.},
  volume={28},
  number={2},
  pages={519--532},
  year={2020},
  publisher={IEEE}
}

@article{azhar2022privacy,
  title={Privacy-preserving and utility-aware participant selection for mobile crowd sensing},
  author={Azhar, Shanila and Chang, Shan and Liu, Ye and Tao, Yuting and Liu, Guohua},
  journal={Mobile Netw. Appl.},
  volume={27},
  number={1},
  pages={290--302},
  year={2022},
  publisher={Springer}
}

@ARTICLE{liwang2025long,
  author={Liwang, Minghui and Gao, Zhibin and Hosseinalipour, Seyyedali and Cheng, Zhipeng and Wang, Xianbin and Jiao, Zhenzhen},
  journal={IEEE Trans. Mobile Comput.}, 
  title={Long-Term or Temporary? Hybrid Worker Recruitment for Mobile Crowd Sensing and Computing}, 
  year={2025},
  volume={24},
  number={2},
  pages={1055-1072},
  publisher={IEEE}
  }

@article{qi2023matching,
  title={Matching-based hybrid service trading for task assignment over dynamic mobile crowdsensing networks},
  author={Qi, Houyi and Liwang, Minghui and Hosseinalipour, Seyyedali and Xia, Xiaoyu and Cheng, Zhipeng and Wang, Xianbin and Jiao, Zhenzhen},
  journal={IEEE Trans. Serv. Comput.},
  volume={17},
  number={5},
  pages={2597--2612},
  year={2024},
  publisher={IEEE}
}

@article{zhao2025privacy,
  author={Zhao, Jingcheng and Xue, Kaiping and Li, Ruidong and Zhu, Bin and Li, Meng and Sun, Qibin and Lu, Jun},
  journal={IEEE Trans. Depend. Sec. Comput.}, 
  title={Privacy-Preserving Truth Discovery of Evolving Truths for Mobile Crowdsensing Systems}, 
  year={2025},
  volume={22},
  number={6},
  pages={6406-6421},
  publisher={IEEE}
  }

@article{Liang2025Privacy,
  author={Liang, Yihuai and Li, Yan and Shin, Byeong-Seok},
  journal={IEEE Trans. Depend. Sec. Comput.}, 
  title={Privacy-Preserving and Reliable Truth Discovery for Heterogeneous Fog-Based Crowdsensing}, 
  year={2025},
  volume={22},
  number={3},
  pages={2338-2351},
  publisher={IEEE}
  }

@article{You2025Location,
  author={You, Zhichao and Dong, Xuewen and Liu, Ximeng and Gao, Sheng and Wang, Yongzhi and Shen, Yulong},
  journal={IEEE Trans. Depend. Sec. Comput.}, 
  title={Location Privacy Preservation Crowdsensing With Federated Reinforcement Learning}, 
  year={2025},
  volume={22},
  number={3},
  pages={1877-1894},
  publisher={IEEE}
  }

@article{Zhou2025Unknown,
  author={Zhou, Qihang and Zhang, Xinglin and Yang, Zheng},
  journal={IEEE Trans. Mobile Comput.}, 
  title={Unknown Worker Recruitment With Long-Term Incentive in Mobile Crowdsensing}, 
  year={2025},
  volume={24},
  number={2},
  pages={999-1015},
  publisher={IEEE}
  }

@article{Ouyang2025MWRS,
  author={Ouyang, Yan and Zeng, Feng and Xiong, Neal N. and Liu, Anfeng and Pedrycz, Witold},
  journal={IEEE Trans. Mobile Comput.}, 
  title={MWRS: A MAB-Based Worker Recruitment Scheme With Tripartite Stackelberg Game for Reliable Mobile Crowdsensing}, 
  year={2025},
  volume={24},
  number={7},
  pages={5665-5680},
  publisher={IEEE}
  }

@article{Guang2025Game,
  author={Guang, Xiaoliang and Peng, Yuhuai and Wang, Chenlu},
  journal={IEEE Trans. Mobile Comput.}, 
  title={Game-Based Multi-UAV Dynamic Collaborative with Energy-Efficient Hierarchical Information Sharing for Mobile Crowdsensing}, 
  year={2025},
  volume={},
  number={},
  pages={1-16},
  publisher={IEEE}
  }

@article{Zhang2025An,
  author={Zhang, Jixian and Chen, Peng and Yang, Xuelin and Wu, Hao and Li, Weidong},
  journal={IEEE Trans. Mobile Comput.}, 
  title={An Optimal Reverse Affine Maximizer Auction Mechanism for Task Allocation in Mobile Crowdsensing}, 
  year={2025},
  volume={24},
  number={8},
  pages={7475-7488},
  publisher={IEEE}
  }

@inproceedings{shokri2012protecting,
  title={Protecting location privacy: optimal strategy against localization attacks},
  author={Shokri, Reza and Theodorakopoulos, George and Troncoso, Carmela and Hubaux, Jean-Pierre and Le Boudec, Jean-Yves},
  booktitle={Proc. the 2012 ACM Conf. Comput. commun. security},
  pages={617--627},
  year={2012}
}

@article{wang2019mobile,
  title={Mobile crowdsourcing task allocation with differential-and-distortion geo-obfuscation},
  author={Wang, Leye and Yang, Dingqi and Han, Xiao and Zhang, Daqing and Ma, Xiaojuan},
  journal={IEEE Trans. Depend. Sec. Comput.},
  volume={18},
  number={2},
  pages={967--981},
  year={2019},
  publisher={IEEE}
}

@inproceedings{yu2017dynamic,
  title={Dynamic Differential Location Privacy with Personalized Error Bounds.},
  author={Yu, Lei and Liu, Ling and Pu, Calton},
  booktitle={NDSS},
  volume={17},
  pages={1--15},
  year={2017}
}

@article{warner1965randomized,
  title={Randomized response: A survey technique for eliminating evasive answer bias},
  author={Warner, Stanley L},
  journal={J. the Amer. statistical Assoc.},
  volume={60},
  number={309},
  pages={63--69},
  year={1965},
  publisher={Taylor \& Francis}
}

@article{dong2022gaussian,
  title={Gaussian differential privacy},
  author={Dong, Jinshuo and Roth, Aaron and Su, Weijie J},
  journal={J. the Royal Statistical Soc. Series B: Statistical Methodology},
  volume={84},
  number={1},
  pages={3--37},
  year={2022},
  publisher={Oxford University Press}
}

@inproceedings{wang2016using,
  title={Using randomized response for differential privacy preserving data collection.},
  author={Wang, Yue and Wu, Xintao and Hu, Donghui},
  booktitle={EDBT/ICDT Workshops},
  volume={1558},
  pages={0090--6778},
  year={2016}
}

@article{monderer1996potential,
  title={Potential games},
  author={Monderer, Dov and Shapley, Lloyd S},
  journal={Games and econ. behavior},
  volume={14},
  number={1},
  pages={124--143},
  year={1996},
  publisher={Elsevier}
}

@article{ding2021potential,
  title={A potential game theoretic approach to computation offloading strategy optimization in end-edge-cloud computing},
  author={Ding, Yan and Li, Kenli and Liu, Chubo and Li, Keqin},
  journal={IEEE Trans. Parallel Distrib. Syst.},
  volume={33},
  number={6},
  pages={1503--1519},
  year={2021},
  publisher={IEEE}
}

@article{wang2022cooperative,
  title={Cooperative and competitive multi-agent systems: From optimization to games},
  author={Wang, Jianrui and Hong, Yitian and Wang, Jiali and Xu, Jiapeng and Tang, Yang and Han, Qing-Long and Kurths, J{\"u}rgen},
  journal={IEEE/CAA J. Autom. Sinica},
  volume={9},
  number={5},
  pages={763--783},
  year={2022},
  publisher={IEEE}
}

@article{konda2024optimal,
  title={Optimal utility design of greedy algorithms in resource allocation games},
  author={Konda, Rohit and Chandan, Rahul and Grimsman, David and Marden, Jason R},
  journal={IEEE Trans. Autom. Control},
  volume={69},
  number={10},
  pages={6592--6604},
  year={2024},
  publisher={IEEE}
}

@article{huang2024distributed,
  title={Distributed Nash equilibrium seeking for multicluster aggregative game of Euler--Lagrange systems with coupled constraints},
  author={Huang, Yi and Meng, Ziyang and Sun, Jian},
  journal={IEEE Trans. Cybern.},
  volume={54},
  number={10},
  pages={5672--5683},
  year={2024},
  publisher={IEEE}
}

@article{zhu2022nash,
  title={Nash equilibrium in iterated multiplayer games under asynchronous best-response dynamics},
  author={Zhu, Yuying and Xia, Chengyi and Chen, Zengqiang},
  journal={IEEE Trans. Autom. Control},
  volume={68},
  number={9},
  pages={5798--5805},
  year={2022},
  publisher={IEEE}
}

@misc{chicago_taxi_trips_2013,
  title        = {Taxi Trips 2013},
  howpublished = {[Online]. Available: https://data.cityofchicago.org/Transportation/Taxi-Trips-2013/6h2x-drp2},
  author       = {{City of Chicago}},
  year         = {2013}
}

@article{qi2026Bank,
  author={Qi, Houyi and Liwang, Minghui and Hosseinalipour, Seyyedali and Fu, Liqun and Zou, Sai and Ni, Wei},
  journal={IEEE J. Sel. Areas Commun.}, 
  title={Future Resource Bank for ISAC: Achieving Fast and Stable Win-Win Matching for Both Individuals and Coalitions}, 
  year={2026},
  volume={44},
  number={},
  pages={513-530},
 publisher={IEEE}
}

@techreport{3GPP2022,
  author       = {{3GPP}},
  title        = {{5G; Study on Scenarios and Requirements for Next Generation Access Technologies}},
  institution  = {{3GPP}},
  number       = {{TR 38.913 V17.0.0}},
  type         = {{Technical Report}},
  year         = {2022},
  month        = may,
  note         = {{Release 17}}
}

@article{yi2020Multi-User,
  author={Yi, Changyan and Cai, Jun and Su, Zhou},
  journal={IEEE Trans. Mobile Comput.}, 
  title={A Multi-User Mobile Computation Offloading and Transmission Scheduling Mechanism for Delay-Sensitive Applications}, 
  year={2020},
  volume={19},
  number={1},
  pages={29-43},
  publisher={IEEE}
}

@article{zhang2025utility,
  title={A Utility-Optimal Reverse Posted Pricing Mechanism for Online Mobile Crowdsensing Task Allocation},
  author={Zhang, Jixian and Yang, Xuelin and Chen, Peng and Wang, Zhemin and Li, Weidong and He, Zhenli and Li, Keqin},
  journal={IEEE Trans. Serv. Comput.},
  year={2025},
  volume={18},
  number={5},
  pages={2588-2601},
}

@ARTICLE{tong2024privacy,
	author={Tong, Fei and Zhou, Yuanhang and Wang, Kaiming and Cheng, Guang and Niu, Jianyu and He, Shibo},
	journal={IEEE Trans. Depend. Sec. Comput.}, 
	title={A Privacy-Preserving Incentive Mechanism for Mobile Crowdsensing Based on Blockchain}, 
	year={2024},
	volume={21},
	number={6},
	pages={5071-5085}}

@ARTICLE{xu2023gdrl,
	author={Xu, Chenghao and Song, Wei},
	journal={IEEE Trans. Netw. Sci. Eng.}, 
	title={Intelligent Task Allocation for Mobile Crowdsensing With Graph Attention Network and Deep Reinforcement Learning}, 
	year={2023},
	volume={10},
	number={2},
	pages={1032-1048}
	}
\end{spacing}

\newpage
\clearpage
\appendices
\setcounter{equation}{49}
\setcounter{table}{2}
\section{Key Notations}
Key notations in this paper are summarized in Table 3.

\begin{table*}[b!]
	{\footnotesize
		\caption{\footnotesize{Key notations}}
		\begin{center}
			\begin{tabular}{|l|l|}
				\hline
				\multicolumn{1}{|l|}{\textbf{Notation}} & \multicolumn{1}{l|}{\textbf{Explanation}} \\ \hline
				$\bm{\mathcal{S}},\,\bm{\mathcal{W}},\,\bm{\mathcal{S}}^{\mathrm{on}},\,\bm{\mathcal{W}}^{\mathrm{on}}$ & \makecell[l]{Task set and mobile-worker set in the MCS system; tasks with execution-time quality deficits and \\workers not locked by offline contracts for online temporary recruitment} \\ \hline
				$s_i,\,w_j$ & The $i$-th task in $\bm{\mathcal{S}}$ and the $j$-th worker in $\bm{\mathcal{W}}$ \\ \hline
				$\{loc_i^{\mathrm{S}},\,B_i,\,Q_i^{\mathrm{D}}\}$ & Location, budget, and quality-demand threshold of task $s_i$ \\ \hline
				$\{loc_j^{\mathrm{W}},\,\varepsilon_j,\,\theta_j,\,\pi_j\}$ & \makecell[l]{Location, personalized/trial privacy-budget parameter, sensing capability, and arrival probability\\ of worker $w_j$} \\ \hline
				$\alpha_j\sim\mathbf{B}(\pi_j)$ & Online availability indicator of worker $w_j$ (arrival/dropout uncertainty) \\ \hline
				$x_{i,j}\in\{0,1\},\,\bm{x}$ & Offline assignment indicator and the assignment matrix \\ \hline
				$a_i,\,\bm{a},\,y_{i,j}(\bm{a})$ & Task $s_i$'s worker-set strategy, joint profile, and induced selection indicator \\ \hline
				$\sigma_{i,j}^2,\,q_{i,j}$ & Effective error variance and single-shot sensing quality (precision) of $w_j$ for $s_i$ \\ \hline
				$n_i,\,\zeta_i,\,Q_i$ & \makecell[l]{Realized number of participants for $s_i$, redundancy factor, and redundancy-aware aggregated quality} \\ \hline
				$p_{i,j}$ & Payment from the platform/task $s_i$ to worker $w_j$ \\ \hline
				$c^\mathrm{W}_{i,j},\,c^{\mathrm{exe}}_{i,j},\,c^{\mathrm{priv}}_j$ & \makecell[l]{Worker cost, execution cost (e.g., travel), and privacy cost induced by $\varepsilon_j$} \\ \hline
				$\mu_j,\,\lambda_j,\,\omega_3$ & \makecell[l]{Execution-cost coefficient, privacy-cost scaling coefficient, and quality-to-utility mapping factor} \\ \hline
				$u_j^\mathrm{W},\,u_i^\mathrm{S},\,\mathbb{SW},\,\mathbb{SW}^{\mathrm{on}}$ & \makecell[l]{Worker/task utilities, realized practical SW under a worker-arrival realization, and incremental online SW} \\ \hline
				$\Phi^{\mathrm{off}}(\bm{a})$ & Exact offline potential function (expected SW over stochastic worker arrivals) \\ \hline
				$U_i(\bm{a})$ & Task payoff defined by marginal contribution to $\Phi^{\mathrm{off}}(\bm{a})$ \\ \hline
				$\mathbf{b}_j,\,\tilde{\mathbf{b}}_j,\,\tilde{\mathbf{b}}^{\mathrm{perm}}_j(e)$ & True intent vector, perturbed report, and epoch-stable memoized report \\ \hline
				$f(\cdot\mid\cdot),\,\varepsilon_j^\star$ & IPM and PRIMER-calibrated personalized privacy budget \\ \hline
				$Q^{\mathrm{loss}}_j(\phi,f,\Delta_j),\,\Delta_j(\tilde{\mathbf{b}}_j,\mathbf{b}_j),\,\gamma_j$ & \makecell[l]{EIRD, weighted Hamming distortion, and worker-specific distortion weight} \\ \hline
				$\phi(\mathbf{b}_j),\,\xi_j,\,\beta^0,\,h(\tilde{\mathbf{b}}_j)$ & \makecell[l]{Prior of intents, one-snapshot eIE, eIE threshold, and auxiliary lower-bound function} \\ \hline
				$R_i^{\mathrm{Qual}},\,R_j^{\mathrm{Pref}},\,\rho_1,\,\rho_2,\,\mathrm{mis}_j(\bm{x},\mathbf{b}_j)$ & \makecell[l]{Quality-violation risk, intent-mismatch risk, risk thresholds, and mismatch ratio function} \\ \hline
				$\mathcal{A}^{\mathrm{uni}}_i(t),\,a_i^{\mathrm{br}}(t),\,\Delta_i(t)$ & \makecell[l]{Unilateral feasible strategy set of task $s_i$ under current strategies of other tasks,\\ exact feasible best response, and exact potential gain (ASPIRE-Off)} \\ \hline
				$N_{i,j},\,t^{\mathrm{U}}_{i,j},\,t^{\mathrm{D}}_{i,j},\,e^{\mathrm{W}}_{j},\,e^{\mathrm{S}}_{i}$ & \makecell[l]{Interaction count on pair $(s_i,w_j)$, uplink/downlink latency, and transmission powers of worker \\and task/platform} \\ \hline
			\end{tabular}
		\end{center}
	}
\end{table*}

\section{Online Stage: Potential-Game-based Temporary Recruitment}
\label{app:online}
The offline pre-plan $\bm{y}(\bm{a}^\star)$ serves as an execution guideline by locking a set of long-term contracts for each task.
Nevertheless, due to realized arrivals/dropouts and mobility perturbations, the actually delivered quality may deviate from the offline expectation, leading to execution-time quality deficits for some tasks.
To mitigate such deficits under strict interaction limits, iParts activates an online \emph{temporary recruitment} procedure that \textit{(i)} only involves the unmet-demand tasks and an idle worker pool (i.e., workers not locked by offline contracts), and \textit{(ii)} follows a potential-game-based, low-overhead update consistent with the offline stage.

\subsection{Basic Modeling for Online Temporary Recruitment}
Given the offline profile $\bm{y}(\bm{a}^\star)$, let $\bm{\mathcal{W}}_i^{\mathrm{off}}\subseteq\bm{\mathcal{W}}$ denote the set of workers pre-contracted to task $s_i$.
During online execution, each worker $w_j$ arrives with indicator $\alpha_j\in\{0,1\}$.
Thus, the realized set of \emph{arrived} offline workers for task $s_i$ is
\begin{equation}
	\bm{\mathcal{W}}_{i,\mathrm{arr}}^{\mathrm{off}}
	=
	\{\,w_j\in \bm{\mathcal{W}}_i^{\mathrm{off}} \mid \alpha_j=1 \,\},
~~
	n_i^{\mathrm{off}}
	=
	|\bm{\mathcal{W}}_{i,\mathrm{arr}}^{\mathrm{off}}|.
	\label{eq:Woff_arrived_app}
\end{equation}

\noindent\textit{(i) Baseline realized quality from offline contracts.}
Under the redundancy-aware aggregation in \eqref{eq:Qi_redundancy}, the delivered (baseline) quality contributed by the arrived offline workers is
\begin{equation}
	Q_i^{\mathrm{base}}
	\triangleq
	\frac{\sum_{w_j\in \bm{\mathcal{W}}_{i,\mathrm{arr}}^{\mathrm{off}}} q_{i,j}}
	{1+\big(n_i^{\mathrm{off}}-1\big)\zeta_i}.
	\label{eq:Qi_base_app}
\end{equation}
If $Q_i^{\mathrm{base}}<Q_i^{D}$, task $s_i$ exhibits an execution-time quality deficit and is considered for online remedy.
Accordingly, we define the set of unmet-demand tasks as
\begin{equation}
	\bm{\mathcal{S}}^{\mathrm{on}}
	=
	\{\, s_i\in\bm{\mathcal{S}} \mid Q_i^{\mathrm{base}}<Q_i^{D}\,\}.
	\label{eq:S_on_def_app}
\end{equation}

\noindent\textit{(ii) Idle worker for temporary recruitment.}
To keep the online procedure lightweight, we restrict recruitment to workers that are \emph{not locked} by offline contracts.
Specifically, the idle worker set is defined as
\begin{equation}
	\bm{\mathcal{W}}^{\mathrm{on}}
=
	\{\, w_j\in\bm{\mathcal{W}}\mid \sum_{s_i\in\bm{\mathcal{S}}}y_{i,j}(\bm{a}^\star)=0 \,\}.
	\label{eq:W_on_def_app}
\end{equation}
Workers in $\bm{\mathcal{W}}^{\mathrm{on}}$ are eligible for temporary recruitment only if they pass the same acceptance screening as in the offline stage, i.e., perturbed-intent visibility and individual rationality.

\noindent\textit{(iii) Remaining budgets.}
Since payments to the arrived offline workers consume part of the task budget, the remaining budget available for online remedy is
\begin{equation}
	\bar B_i
	=B_i-\sum_{w_j\in \bm{\mathcal{W}}_{i,\mathrm{arr}}^{\mathrm{off}}} p_{i,j}
	\label{eq:residual_budget_app}
\end{equation}

\subsection{Joint Optimization for Online Temporary Recruitment}
In the online stage, each task $s_i\in\bm{\mathcal{S}}^{\mathrm{on}}$ recruits workers from $\bm{\mathcal{W}}^{\mathrm{on}}$ to reduce the quality gap while respecting worker exclusivity and remaining budgets.
Let $x^{\mathrm{on}}_{i,j}\in\{0,1\}$ denote whether idle worker $w_j\in \bm{\mathcal{W}}^{\mathrm{on}}$ is temporarily recruited by task $s_i\in\bm{\mathcal{S}}^{\mathrm{on}}$.
The online utilities and the incremental social welfare $\mathbb{SW}^{\mathrm{on}}$ follow \eqref{eq:worker_utility_online_spot}--\eqref{eq:SW_online_spot}.

To characterize the final quality after adding online recruits, we define the final realized quality of task $s_i$ as
\begin{equation}
	Q_i^{\mathrm{fin}}
	=
	\frac{
		\sum_{w_j\in \bm{\mathcal{W}}_{i,\mathrm{arr}}^{\mathrm{off}}} q_{i,j}
		+
		\sum_{w_j\in \bm{\mathcal{W}}^{\mathrm{on}}} x^{\mathrm{on}}_{i,j} q_{i,j}
	}{
		1+\Big(\big(n_i^{\mathrm{off}}+n_i^{\mathrm{on}}\big)-1\Big)\zeta_i
	},
	\label{eq:Qi_final_app}
\end{equation}
where $ n_i^{\mathrm{on}}= \sum_{w_j\in\bm{\mathcal{W}}^{\mathrm{on}}} x^{\mathrm{on}}_{i,j}$. Accordingly, the online temporary recruitment problem is formulated as
\begin{align}
	\bm{\mathcal{F}^{\mathrm{on}}}:~&
	\underset{\bm{x}^{\mathrm{on}}}{\max}\quad \mathbb{SW}^{\mathrm{on}}
	\label{ProOn_app}
	\\
	\text{s.t.}\quad
	&\sum_{s_i\in\bm{\mathcal{S}}^{\mathrm{on}}} x^{\mathrm{on}}_{i,j}\le 1,\ \forall w_j\in \bm{\mathcal{W}}^{\mathrm{on}}
	\label{eq:on_excl_app}\tag{56a}
	\\
	&\sum_{w_j\in\bm{\mathcal{W}}^{\mathrm{on}}} x^{\mathrm{on}}_{i,j} p_{i,j}\le \bar B_i,\ \forall s_i\in \bm{\mathcal{S}}^{\mathrm{on}}
	\label{eq:on_budget_app}\tag{56b}
	\\
	& p_{i,j}\ge c^\mathrm{W}_{i,j},\ \forall s_i\in\bm{\mathcal{S}}^{\mathrm{on}},\ \forall w_j\in\bm{\mathcal{W}}^{\mathrm{on}}
	\label{eq:on_IR_app}\tag{56c}
	\\
	& Q_i^{\mathrm{fin}} \ge Q_i^D,\ \forall s_i\in \bm{\mathcal{S}}^{\mathrm{on}}
	\label{eq:on_quality_app}\tag{56d}\\
	& x^{\mathrm{on}}_{i,j}=0,\ \forall (i,j)\ \text{with}\ \tilde b_{i,j}=0,
	\label{eq:on_intent_screen_app}\tag{56e}
\end{align}
where \eqref{eq:on_excl_app} enforces idle-worker exclusivity, \eqref{eq:on_budget_app} enforces residual-budget feasibility, \eqref{eq:on_IR_app} ensures individual rationality, and \eqref{eq:on_quality_app} specifies the quality-restoration requirement of the online remedy.
In ASPIRE-On, since the online process starts from realized quality deficits, \eqref{eq:on_quality_app} is used as an acceptance-screening condition for temporary updates rather than as a requirement on the initial all-empty recruitment profile.
Constraint \eqref{eq:on_intent_screen_app} restricts recruitment to workers whose perturbed intent reports indicate willingness.
Note that this screening does not consume additional mechanism-level LDP budget, since $\tilde b_{i,j}$ is generated by MIRROR and reused as an eligibility signal (post-processing invariance applies).
Next, we reformulate $\bm{\mathcal{F}^{\mathrm{on}}}$ as a task-dominated exact potential game in the online stage, and develop ASPIRE-On as a DP-assisted potential-improvement procedure for bounded-round online remediation.

\subsection{Online Potential-Game Formulation}
To enable low-overhead distributed remedy, we cast $\bm{\mathcal{F}^{\mathrm{on}}}$ as a task-dominated non-cooperative game.
As in the offline stage, each task selects a subset of available idle workers, while feasibility is enforced at the joint system level.
\subsubsection{Game formulation}
During online execution, only a subset of tasks may experience quality deficits (i.e., $\bm{\mathcal{S}}^{\mathrm{on}}$), and the platform recruits additional workers from the idle pool $\bm{\mathcal{W}}^{\mathrm{on}}$ toward satisfying $Q_i^{\mathrm{fin}}\ge Q_i^{D}$ under remaining budgets $\{\bar B_i\}$.
Although the online stage is designed to be \textit{lightweight}, tasks are still \emph{coupled} through the shared idle worker pool: since each idle worker can be temporarily recruited by at most one task, the recruitment decision of a task immediately restricts the feasible choices of other tasks, thereby creating externalities.
Therefore, we model the online temporary recruitment as a task-dominated non-cooperative game:
\begin{equation}
	\mathcal{G}^{\mathrm{on}}
	=
	\Big(
	\bm{\mathcal{S}}^{\mathrm{on}},\
	\{\mathcal{A}^{\mathrm{on}}_i\}_{s_i\in\bm{\mathcal{S}}^{\mathrm{on}}},\
	\{U_i^{\mathrm{on}}\}_{s_i\in\bm{\mathcal{S}}^{\mathrm{on}}}
	\Big),
	\label{eq:Game_task_on_def_app}
\end{equation}
where the key elements are specified as follows.

\vspace{0.2em}
\noindent$\bullet$ \textbf{Players.}
Tasks in $\bm{\mathcal{S}}^{\mathrm{on}}$ act as players.
Each task $s_i\in\bm{\mathcal{S}}^{\mathrm{on}}$ selects a subset of idle workers to perform temporary recruitment.

\vspace{0.2em}
\noindent$\bullet$ \textbf{Worker-side acceptance filtering.}
An idle worker $w_j\in\bm{\mathcal{W}}^{\mathrm{on}}$ can be considered by task $s_i$ only if it satisfies:
\emph{(i) perturbed intent visibility}: $\tilde b_{i,j}=1$; and
\emph{(ii) individual rationality}: $p_{i,j}\ge c^\mathrm{W}_{i,j}$.
Accordingly, the online candidate set of task $s_i$ is
\begin{equation}{
	\begin{aligned}\small
	\tilde{\bm{\mathcal{W}}}^{\mathrm{on}}_i
	\triangleq
	\Big\{
	w_j\in\bm{\mathcal{W}}^{\mathrm{on}}\ \big|\
	\tilde b_{i,j}=1,\ p_{i,j}\ge c^\mathrm{W}_{i,j},\ \varepsilon_j^\star\neq \textsf{INFEASIBLE}
	\Big\}.
\end{aligned}}
	\label{eq:candidate_workers_on_app}
\end{equation}
Namely, $s_i$ can recruit only from $\tilde{\bm{\mathcal{W}}}^{\mathrm{on}}_i$.

\vspace{0.2em}
\noindent$\bullet$ \textbf{Strategy space.}
The strategy of task $s_i$ is selecting a temporary recruitment set from its candidate pool under the remaining budget $\bar B_i$:
\begin{equation}
	\begin{aligned}
			a_i^{\mathrm{on}} \in \mathcal{A}_i^{\mathrm{on}}
		=
		\Big\{
		\bm{\mathcal{W}}_i^{\mathrm{on}}\subseteq \tilde{\bm{\mathcal{W}}}^{\mathrm{on}}_i\ \big|\
		\sum_{w_j\in \bm{\mathcal{W}}_i^{\mathrm{on}}} p_{i,j}\le \bar B_i
		\Big\}.
	\end{aligned}
	\label{eq:Aion_app}
\end{equation}
Let $\bm{a}^{\mathrm{on}}=(a_1^{\mathrm{on}},\ldots,a_{|\bm{\mathcal{S}}^{\mathrm{on}}|}^{\mathrm{on}})$ be the joint strategy profile.
It induces the binary recruitment indicator
\begin{equation}{
	\begin{aligned}
			x_{i,j}^{\mathrm{on}} = y_{i,j}^{\mathrm{on}}(\bm{a}^{\mathrm{on}})
		=
		\mathbbm{1}\{w_j\in a_i^{\mathrm{on}}\},
		~~
		\bm{y}^{\mathrm{on}}(\bm{a}^{\mathrm{on}})=[y_{i,j}^{\mathrm{on}}(\bm{a}^{\mathrm{on}})].
	\end{aligned}}
	\label{eq:y_from_a_on_app}
\end{equation}

\vspace{0.2em}
\noindent$\bullet$ \textbf{Operational feasible set and restoration-screened deviations.}
During online remediation, the platform starts from realized quality deficits.
Therefore, the all-empty temporary recruitment profile should be operationally feasible even though it may not yet meet all unmet quality demands.
We first define the operational feasible set as
\begin{equation}
	\begin{aligned}
		\mathcal{A}^{\mathrm{on,op}}
		\triangleq
		\Big\{
		\bm a^{\mathrm{on}}\in
		\prod_{s_i\in\bm{\mathcal{S}}^{\mathrm{on}}}
		\mathcal{A}^{\mathrm{on}}_i
		\ \Big|\
		&\sum_{s_i\in\bm{\mathcal{S}}^{\mathrm{on}}}
		y^{\mathrm{on}}_{i,j}(\bm a^{\mathrm{on}})
		\le 1,\\
		&\forall w_j\in\bm{\mathcal{W}}^{\mathrm{on}}
		\Big\}.
	\end{aligned}
	\label{eq:A_on_op}
\end{equation}
where the local strategy space $\mathcal{A}^{\mathrm{on}}_i$ already embeds the remaining-budget constraint and the perturbed-intent/individual-rationality screening, while \eqref{eq:A_on_op} enforces idle-worker exclusivity across tasks.

The quality-restoration requirement is then imposed on accepted unilateral remedial moves.
Given $\bm a^{\mathrm{on}}_{-i}$, the restoration-screened unilateral set of task $s_i$ is defined as
\begin{equation}
	\begin{aligned}
		\mathcal{R}^{\mathrm{on}}_i(\bm a^{\mathrm{on}}_{-i})
		\triangleq
		\Big\{
		a^{\mathrm{on}}_i\in\mathcal{A}^{\mathrm{on}}_i
		\ \Big|\
		&(a^{\mathrm{on}}_i,\bm a^{\mathrm{on}}_{-i})\in
		\mathcal{A}^{\mathrm{on,op}},
		\\
		&Q_i^{\mathrm{fin}}(a^{\mathrm{on}}_i,\bm a^{\mathrm{on}}_{-i})\ge Q_i^D
		\Big\}.
	\end{aligned}
	\label{eq:R_on_i}
\end{equation}
Thus, operational feasibility keeps the online recruitment state well-defined, while \eqref{eq:R_on_i} screens whether a candidate update satisfies the updated-task quality requirement.

\vspace{0.2em}
\noindent$\bullet$ \textbf{Online potential function.}
Since online availability is treated as realized, we define the online potential function directly by the online SW:
\begin{equation}
	\Phi^{\mathrm{on}}(\bm{a}^{\mathrm{on}})
	=
	\mathbb{SW}^{\mathrm{on}}\big(\bm{y}^{\mathrm{on}}(\bm{a}^{\mathrm{on}})\big),
	\label{eq:Phi_on_app}
\end{equation}
where $\mathbb{SW}^{\mathrm{on}}$ is given in \eqref{eq:SW_online_spot}.
Similar to the offline stage, payments cancel out when aggregating task and worker utilities in $\mathbb{SW}^{\mathrm{on}}$; hence $\Phi^{\mathrm{on}}$ can be equivalently expressed in a ``quality gain minus worker costs'' form, which facilitates incremental evaluation.

\vspace{0.2em}
\noindent$\bullet$ \textbf{Task payoff: marginal contribution.}
To align unilateral task updates with the global objective, we define each task's payoff as its marginal contribution to $\Phi^{\mathrm{on}}$:
\begin{equation}
	U_i^{\mathrm{on}}(\bm{a}^{\mathrm{on}})
	=
	\Phi^{\mathrm{on}}(\bm{a}^{\mathrm{on}})
	-
	\Phi^{\mathrm{on}}(\varnothing,\bm{a}^{\mathrm{on}}_{-i}),
	\qquad \forall s_i\in\bm{\mathcal{S}}^{\mathrm{on}},
	\label{eq:Ui_on_marginal_app}
\end{equation}
where $(\varnothing,\bm{a}^{\mathrm{on}}_{-i})$ means that task $s_i$ recruits no idle workers while other tasks keep unchanged.

\vspace{0.2em}
\noindent
For the underlying exact online game, we use the following constrained Nash equilibrium (NE) notion.

\begin{Defn}[Constrained NE of the operational online game]
	\label{def:on_constrained_ne_app}
	A joint strategy $\bm a^{\mathrm{on}\star}\in\mathcal{A}^{\mathrm{on,op}}$ is a pure-strategy constrained NE of the operational online game if, for any task $s_i\in\bm{\mathcal{S}}^{\mathrm{on}}$,
	\begin{equation}
		\begin{aligned}
				U_i^{\mathrm{on}}&(a_i^{\mathrm{on}\star},\bm{a}^{\mathrm{on}\star}_{-i})
			\ge
			U_i^{\mathrm{on}}(a_i^{\mathrm{on}},\bm{a}^{\mathrm{on}\star}_{-i}),\forall a_i^{\mathrm{on}}
			\\&
			\text{s.t.}\ (a_i^{\mathrm{on}},\bm{a}^{\mathrm{on}\star}_{-i})\in\mathcal{A}^{\mathrm{on,op}}.
		\end{aligned}
		\label{eq:on_ne_def_app}
	\end{equation}
\end{Defn}
This NE notion characterizes unilateral optimality within the operational online state space; ASPIRE-On further applies restoration screening when accepting remedial updates.

\subsubsection{Potential property of the online game}
We show that $\mathcal{G}^{\mathrm{on}}$ forms an exact potential game on $\mathcal{A}^{\mathrm{on,op}}$, which supports the constrained-NE characterization and FIP under exact feasible-improvement dynamics of the operational online game.

\begin{Defn}[Exact potential game]
	On the operational feasible set $\mathcal{A}^{\mathrm{on,op}}$, if there exists a scalar function $P^{\mathrm{on}}(\bm{a}^{\mathrm{on}})$ such that for any task $s_i\in\bm{\mathcal{S}}^{\mathrm{on}}$ and any two operationally feasible joint strategies
	$\bm{a}^{\mathrm{on}}=(a_i^{\mathrm{on}},\bm{a}^{\mathrm{on}}_{-i})\in\mathcal{A}^{\mathrm{on,op}}$ and
	$\bm{a}^{\mathrm{on}\prime}=(a_i^{\mathrm{on}\prime},\bm{a}^{\mathrm{on}}_{-i})\in\mathcal{A}^{\mathrm{on,op}}$,
	\begin{equation}
		U_i^{\mathrm{on}}(\bm{a}^{\mathrm{on}\prime})-U_i^{\mathrm{on}}(\bm{a}^{\mathrm{on}})
		=
		P^{\mathrm{on}}(\bm{a}^{\mathrm{on}\prime})-P^{\mathrm{on}}(\bm{a}^{\mathrm{on}}),
	\end{equation}
	then the game is an exact potential game and $P^{\mathrm{on}}(\cdot)$ is a potential function.
\end{Defn}

\begin{thm}[Exact potential property of the operational online game]
	\label{thm:on_exact_potential_app}
	On the operational feasible set $\mathcal{A}^{\mathrm{on,op}}$, game $\mathcal{G}^{\mathrm{on}}$ is an exact potential game with potential function
	\begin{equation}
		P^{\mathrm{on}}(\bm{a}^{\mathrm{on}})=\Phi^{\mathrm{on}}(\bm{a}^{\mathrm{on}}).
	\end{equation}
	Hence, the operational online game admits at least one pure-strategy constrained NE and satisfies the finite improvement property under exact feasible unilateral improvement dynamics.
\end{thm}

\begin{proof}
	For any task $s_i\in\bm{\mathcal{S}}^{\mathrm{on}}$ and any two operationally feasible joint strategies
	$\bm{a}^{\mathrm{on}}=(a_i^{\mathrm{on}},\bm{a}^{\mathrm{on}}_{-i})\in\mathcal{A}^{\mathrm{on,op}}$ and
	$\bm{a}^{\mathrm{on}\prime}=(a_i^{\mathrm{on}\prime},\bm{a}^{\mathrm{on}}_{-i})\in\mathcal{A}^{\mathrm{on,op}}$,
	by \eqref{eq:Ui_on_marginal_app} we have
	\begin{equation}
		\begin{aligned}
			&U_i^{\mathrm{on}}(\bm{a}^{\mathrm{on}\prime})-U_i^{\mathrm{on}}(\bm{a}^{\mathrm{on}})
			=
			\Big(\Phi^{\mathrm{on}}(\bm{a}^{\mathrm{on}\prime})-\Phi^{\mathrm{on}}(\varnothing,\bm{a}^{\mathrm{on}}_{-i})\Big)
			\\
			&-
			\Big(\Phi^{\mathrm{on}}(\bm{a}^{\mathrm{on}})-\Phi^{\mathrm{on}}(\varnothing,\bm{a}^{\mathrm{on}}_{-i})\Big)
			=
			\Phi^{\mathrm{on}}(\bm{a}^{\mathrm{on}\prime})-\Phi^{\mathrm{on}}(\bm{a}^{\mathrm{on}})
			\\
			&=
			P^{\mathrm{on}}(\bm{a}^{\mathrm{on}\prime})-P^{\mathrm{on}}(\bm{a}^{\mathrm{on}}),
		\end{aligned}
	\end{equation}
	which verifies the exact potential condition. Therefore, $\mathcal{G}^{\mathrm{on}}$ is an exact potential game on $\mathcal{A}^{\mathrm{on,op}}$ with $P^{\mathrm{on}}=\Phi^{\mathrm{on}}$.
	
	Moreover, since each $\mathcal{A}_i^{\mathrm{on}}$ and $\mathcal{A}^{\mathrm{on,op}}$ are finite, the finite operational game admits a potential maximizer, which corresponds to a constrained NE of the underlying exact game.
	Under exact asynchronous feasible better-response dynamics, the potential function strictly increases along every accepted update, and thus the FIP holds over the finite operational feasible set.
\end{proof}
The above result establishes an exact alignment between unilateral feasible recruitment improvements and the increase of global online welfare, supporting low-overhead remedial coordination via potential-driven asynchronous feasible updates.

\subsubsection{ASPIRE-On: asynchronous self-organized potential improvement for online remedy}
Theorem~\ref{thm:on_exact_potential_app} provides the potential-game basis for online remedial coordination over the operational feasible set.
To keep online remediation lightweight, we propose \textit{ASPIRE-On} (Alg.~\ref{Alg:TaskDPAsyncOn}), a DP-assisted potential-improvement procedure that mitigates execution-time quality deficits using only the unmet-demand task set $\bm{\mathcal{S}}^{\mathrm{on}}$ and the idle worker pool $\bm{\mathcal{W}}^{\mathrm{on}}$.
Unlike ASPIRE-Off, ASPIRE-On does not enumerate the entire operational feasible set.
Instead, it uses \textsc{KnapsackDP} to identify a promising temporary recruitment set and accepts it only after operational feasibility, quality-restoration screening, and exact potential-improvement verification.
Accordingly, ASPIRE-On provides operational feasibility preservation, finite termination, and monotonic potential improvement for bounded-round online remediation.

\begin{algorithm}[t!]
	{\scriptsize \setstretch{0.45}
		\caption{Proposed ASPIRE-On}
		\label{Alg:TaskDPAsyncOn}
		\LinesNumbered
		
		\textbf{Input:}
		online-demand task set $\bm{\mathcal{S}}^{\mathrm{on}}$, idle worker set $\bm{\mathcal{W}}^{\mathrm{on}}$;
		candidate sets $\{\tilde{\bm{\mathcal{W}}}^{\mathrm{on}}_i\}$;
		payments $\bm p=[p_{i,j}]$, remaining budgets $\{\bar B_i\}$;
		quality parameters $\{q_{i,j}\}$, redundancy factors $\{\zeta_i\}$, baseline qualities $\{Q_i^{\mathrm{base}}\}$, thresholds $\{Q_i^D\}$;
		improvement threshold $\varepsilon^{\prime}$, max rounds $R_{\max}$.
		
		\textbf{Initialization:}
		$r\leftarrow 0$; initialize an operationally feasible online profile $\bm a^{\mathrm{on}}(0)=\{a_i^{\mathrm{on}}(0)\}\in\mathcal{A}^{\mathrm{on,op}}$ (e.g., $a_i^{\mathrm{on}}(0)=\varnothing$);
		set $y^{\mathrm{on}}_{i,j}(0)=\mathbbm{1}\{w_j\in a_i^{\mathrm{on}}(0)\}$;
		$\bm{\mathcal{W}}^{\mathrm{idle}}_{\mathrm{on}}(0)\leftarrow \{w_j\in\bm{\mathcal{W}}^{\mathrm{on}}\mid \sum_{s_i\in\bm{\mathcal{S}}^{\mathrm{on}}} y^{\mathrm{on}}_{i,j}(0)=0\}$.
		
		\While{$r<R_{\max}$}{
			Platform broadcasts the current online profile $\bm a^{\mathrm{on}}(r)$ (or equivalently $\bm y^{\mathrm{on}}(r)$) and the idle worker set $\bm{\mathcal{W}}^{\mathrm{idle}}_{\mathrm{on}}(r)$.
			
			\For{$\forall s_i\in\bm{\mathcal{S}}^{\mathrm{on}}$}{
				$\bm{\mathcal{W}}^{\mathrm{avail,on}}_i(r)\leftarrow
				\big(a_i^{\mathrm{on}}(r)\cup \bm{\mathcal{W}}^{\mathrm{idle}}_{\mathrm{on}}(r)\big)\cap \tilde{\bm{\mathcal{W}}}^{\mathrm{on}}_i$;
				
				$a_{i}^{\mathrm{cand,on}}(r)\leftarrow
				\textsc{KnapsackDP}\big(\bm{\mathcal{W}}^{\mathrm{avail,on}}_i(r),\{p_{i,j}\},\{v_{i,j}^{\mathrm{on}}\},\bar B_i\big)$;
				
				$\bm a^{\mathrm{on}\prime}_i(r)\leftarrow (a_{i}^{\mathrm{cand,on}}(r),\bm a^{\mathrm{on}}_{-i}(r))$;
				
				\If{$a_{i}^{\mathrm{cand,on}}(r)$ passes operational feasibility and restoration screening in \eqref{eq:A_on_op}--\eqref{eq:R_on_i}}{
					compute $\Phi^{\mathrm{on}}(\bm a^{\mathrm{on}\prime}_i(r))$ and $\Phi^{\mathrm{on}}(\bm a^{\mathrm{on}}(r))$ exactly by \eqref{eq:Phi_on_app};
					
					$\Delta_i^{\mathrm{on}}(r)\leftarrow \Phi^{\mathrm{on}}(\bm a^{\mathrm{on}\prime}_i(r))-\Phi^{\mathrm{on}}(\bm a^{\mathrm{on}}(r))$;
				}\Else{
					$\Delta_i^{\mathrm{on}}(r)\leftarrow 0$;
				}
			}
			
			$\mathcal{I}^{\mathrm{on}}(r)\leftarrow \{i\in\bm{\mathcal{S}}^{\mathrm{on}}\mid \Delta_i^{\mathrm{on}}(r)>\varepsilon^\prime\}$;
			
			\If{$\mathcal{I}^{\mathrm{on}}(r)=\varnothing$}{
				\textbf{break};
			}
			
			select one task index $i^\star\in\mathcal{I}^{\mathrm{on}}(r)$ (e.g., $\arg\max_{i\in\mathcal{I}^{\mathrm{on}}(r)}\Delta_i^{\mathrm{on}}(r)$);
			
			$a_{i^\star}^{\mathrm{on}}(r+1)\leftarrow a_{i^\star}^{\mathrm{cand,on}}(r)$;
			
			\For{$\forall i\in\bm{\mathcal{S}}^{\mathrm{on}},\, i\neq i^\star$}{
				$a_{i}^{\mathrm{on}}(r+1)\leftarrow a_{i}^{\mathrm{on}}(r)$;
			}
			
			update $y^{\mathrm{on}}_{i,j}(r+1)=\mathbbm{1}\{w_j\in a_i^{\mathrm{on}}(r+1)\}$;
			
			$\bm{\mathcal{W}}^{\mathrm{idle}}_{\mathrm{on}}(r+1)\leftarrow \{w_j\in\bm{\mathcal{W}}^{\mathrm{on}}\mid \sum_{s_i\in\bm{\mathcal{S}}^{\mathrm{on}}} y^{\mathrm{on}}_{i,j}(r+1)=0\}$;
			
			$r\leftarrow r+1$;
		}
		
		\textbf{Return:} $\bm a^{\mathrm{on}\star}=\bm a^{\mathrm{on}}(r)$ and the induced decision $\bm y^{\mathrm{on}}(\bm a^{\mathrm{on}\star})$.
	}
\end{algorithm}

\vspace{0.25em}
\noindent\textbf{Step 1. Online candidate construction and operational initialization} (line 2, Alg.~\ref{Alg:TaskDPAsyncOn}):
Using the same acceptance screening as in the offline stage, the platform constructs the online candidate set $\tilde{\bm{\mathcal{W}}}^{\mathrm{on}}_i$ for each unmet-demand task $s_i$ as in \eqref{eq:candidate_workers_on_app}.
It then initializes an operationally feasible online recruitment profile $\bm a^{\mathrm{on}}(0)\in\mathcal{A}^{\mathrm{on,op}}$ (e.g., all-empty), derives $y^{\mathrm{on}}_{i,j}(0)$, and forms the initial idle set $\bm{\mathcal{W}}^{\mathrm{idle}}_{\mathrm{on}}(0)$.

\vspace{0.35em}
\noindent\textbf{Step 2. Round-wise broadcast and available-set formation} (lines 3-4, Alg.~\ref{Alg:TaskDPAsyncOn}):
We discretize online remediation into rounds $r=0,1,2,\ldots$.
At each round, the platform broadcasts only aggregate state information, including the current profile $\bm a^{\mathrm{on}}(r)$ (equivalently $\bm y^{\mathrm{on}}(r)$) and the current idle set $\bm{\mathcal{W}}^{\mathrm{idle}}_{\mathrm{on}}(r)$.
Each task $s_i$ then constructs its available set
\[
\bm{\mathcal{W}}^{\mathrm{avail,on}}_i(r)
=
\big(a_i^{\mathrm{on}}(r)\cup \bm{\mathcal{W}}^{\mathrm{idle}}_{\mathrm{on}}(r)\big)\cap \tilde{\bm{\mathcal{W}}}^{\mathrm{on}}_i,
\]
meaning that $s_i$ can only (re)select workers already recruited by itself or recruit from the currently idle pool, while still respecting the perturbed-intent candidate constraint.

\vspace{0.35em}
\noindent\textbf{Step 3. DP-based candidate construction as a surrogate potential-gain search} (lines 6-8, Alg.~\ref{Alg:TaskDPAsyncOn}):
Selecting a subset of idle workers from $\bm{\mathcal{W}}^{\mathrm{avail,on}}_i(r)$ under the residual budget $\bar B_i$ admits a 0--1 knapsack structure.
We thus run \textsc{KnapsackDP} to generate a candidate set $a_{i}^{\mathrm{cand,on}}(r)$.
Specifically, each candidate worker $w_j$ is assigned a DP weight and value as
\begin{equation}
	w^{\mathrm{on}}_{i,j}=p_{i,j},\qquad 
v^{\mathrm{on}}_{i,j}=\omega_3 q_{i,j}-c^{\mathrm{W}}_{i,j},
\end{equation}
where $v^{\mathrm{on}}_{i,j}$ approximates the marginal online utility gain (quality gain minus worker cost).
We emphasize that the DP objective is used only to efficiently search a promising subset under the budget,
and does not replace the exact potential evaluation in Step~4.

\vspace{0.35em}
\noindent\textbf{Step 4. Feasibility screening and strict potential-improvement verification} (lines 9-13, Alg.~\ref{Alg:TaskDPAsyncOn}):
The candidate is first checked against the operational feasible set $\mathcal{A}^{\mathrm{on,op}}$ and restoration-screened unilateral set $\mathcal{R}^{\mathrm{on}}_i(\bm a^{\mathrm{on}}_{-i})$.
Only candidates that preserve idle-worker exclusivity, satisfy the remaining budget, and pass the updated-task quality-restoration screening are further evaluated by exact online potential.
Hence, the platform forms $\bm a_i^{\mathrm{on}\prime}(r)=(a_{i}^{\mathrm{cand,on}}(r),\bm a^{\mathrm{on}}_{-i}(r))$ and computes
\begin{equation}
	\Delta_i^{\mathrm{on}}(r)=\Phi^{\mathrm{on}}(\bm a_i^{\mathrm{on}\prime}(r))-\Phi^{\mathrm{on}}(\bm a^{\mathrm{on}}(r)).
\end{equation}
A candidate is considered effective only when $\Delta_i^{\mathrm{on}}(r)>\varepsilon'$; thus, every accepted update strictly increases $\Phi^{\mathrm{on}}$, yielding monotonic potential improvement.

\vspace{0.35em}
\noindent\textbf{Step 5. Asynchronous permission control and single-task unilateral execution} (lines 17-21, Alg.~\ref{Alg:TaskDPAsyncOn}):
Let $\mathcal{I}^{\mathrm{on}}(r)=\{i\in\bm{\mathcal{S}}^{\mathrm{on}}\mid \Delta_i^{\mathrm{on}}(r)>\varepsilon'\}$.
To prevent simultaneous contention for the same idle worker, the platform enforces asynchronous updates:
at most one task $i^\star\in\mathcal{I}^{\mathrm{on}}(r)$ is selected per round (e.g., max-gain) to execute the update
$a_{i^\star}^{\mathrm{on}}(r{+}1)\leftarrow a_{i^\star}^{\mathrm{cand,on}}(r)$,
while all other tasks keep unchanged.
The platform then refreshes $\bm y^{\mathrm{on}}(r{+}1)$ and $\bm{\mathcal{W}}^{\mathrm{idle}}_{\mathrm{on}}(r{+}1)$ accordingly.

\vspace{0.35em}
\noindent\textbf{Step 6. Termination and remedial output} (lines 3-24, Alg.~\ref{Alg:TaskDPAsyncOn}):
If $\mathcal{I}^{\mathrm{on}}(r)=\varnothing$, no verified DP-assisted remedial candidate yields a potential gain larger than $\varepsilon'$, and ASPIRE-On terminates early.
Otherwise, the process continues until reaching $R_{\max}$ rounds.
Since each accepted update remains in $\mathcal{A}^{\mathrm{on,op}}$, satisfies the updated-task restoration screening, and strictly improves $\Phi^{\mathrm{on}}$, ASPIRE-On returns an operationally feasible temporary recruitment profile with monotonic potential improvement.

\section{Supplementary Proofs of Key Properties}
\label{app:proofs}

This appendix provides proof details for the key theoretical results used throughout iParts. 
We first verify the \emph{mechanism-level} privacy guarantee of MIRROR under $\varepsilon_j$-personalized LDP, then show that platform-side quantities derived from perturbed reports enjoy \emph{mechanism-level post-processing invariance}. 
Next, we formalize why memoization neutralizes the adversary’s multi-snapshot frequency gain. 
Finally, we prove the finite termination and feasibility guarantees of PRIMER and ASPIRE-Off, which together ensure that iParts remains deployable under explicit privacy, budget, and risk constraints.

\subsection{Personalized LDP Guarantee of MIRROR via Randomized Response}
\label{app:proofs:ldp_mirror}

We start from the elementary building block: the binary RR used in MIRROR. 
After proving $\varepsilon$-LDP for a single binary intent entry, we extend the result to the full intent vector under the one-entry adjacency notion in Definition~\ref{def:personalized_ldp_intent}.

\begin{lem}[Binary RR satisfies $\varepsilon$-LDP]
	\label{lem:rr_ldp_single}
	Let $b\in\{0,1\}$ and $\tilde b\in\{0,1\}$ be generated by $\textsc{RR}(b,\varepsilon)$ in \eqref{eq:RR_def}. 
	Then for any two adjacent inputs $b^{(1)}\neq b^{(2)}$ and any output $\tilde b\in\{0,1\}$,
	\begin{equation}
		\frac{\Pr(\tilde b \mid b^{(1)})}{\Pr(\tilde b \mid b^{(2)})}\le e^{\varepsilon}.
	\end{equation}
\end{lem}

\begin{proof}
	By symmetry of the binary domain, it suffices to consider the adjacent pair
	$b^{(1)}=1$ and $b^{(2)}=0$ (the reverse case is identical).
	We next verify the LDP likelihood-ratio bound for the two possible outputs $\tilde b\in\{0,1\}$.
	
	\noindent\textbf{Case 1: $\tilde b=1$.}
	Using \eqref{eq:RR_def}, we have
	\begin{equation}
		\frac{\Pr(\tilde b=1\mid b=1)}{\Pr(\tilde b=1\mid b=0)}
		=
		\frac{\frac{e^{\varepsilon}}{e^{\varepsilon}+1}}{\frac{1}{e^{\varepsilon}+1}}
		=
		e^{\varepsilon}.
	\end{equation}
	
	\noindent\textbf{Case 2: $\tilde b=0$.}
	Similarly, by \eqref{eq:RR_def},
	\begin{equation}
		\frac{\Pr(\tilde b=0\mid b=1)}{\Pr(\tilde b=0\mid b=0)}
		=
		\frac{\frac{1}{e^{\varepsilon}+1}}{\frac{e^{\varepsilon}}{e^{\varepsilon}+1}}
		=
		e^{-\varepsilon}
		\le
		e^{\varepsilon}.
	\end{equation}
	
	Since the likelihood ratio is bounded by $e^{\varepsilon}$ in both exhaustive cases,
	the randomized response mechanism satisfies $\varepsilon$-LDP for a single binary entry.
\end{proof}

\begin{lem}[Vector-level $\varepsilon_j$-personalized LDP via entry-wise randomized response]
	\label{lem:rr_ldp_vector}
	Let $\mathbf b_j\in\{0,1\}^{|\bm{\mathcal{S}}|}$ be worker $w_j$'s intent vector, and let $\tilde{\mathbf b}_j$ be generated by applying $\textsc{RR}(\cdot,\varepsilon_j)$ independently to each entry (as in Alg.~\ref{Alg:MSIPM}). 
	Under the adjacency notion in Definition~\ref{def:personalized_ldp_intent} (i.e., $\|\mathbf b_j^{(1)}-\mathbf b_j^{(2)}\|_0=1$), the resulting mechanism satisfies $\varepsilon_j$-personalized LDP: for any adjacent $\mathbf b_j^{(1)},\mathbf b_j^{(2)}$ and any $\tilde{\mathbf b}^\ast\in\tilde{\mathcal{B}}$,
	\begin{equation}
		\frac{\Pr(\tilde{\mathbf b}^\ast\mid \mathbf b_j^{(1)})}{\Pr(\tilde{\mathbf b}^\ast\mid \mathbf b_j^{(2)})}\le e^{\varepsilon_j}.
	\end{equation}
\end{lem}

\begin{proof}
	Consider any two adjacent intent vectors $\mathbf b_j^{(1)}$ and $\mathbf b_j^{(2)}$ that differ in exactly one task dimension, denoted by $s_{i^\circ}$, and coincide on all remaining dimensions.
	Since the randomized response is applied independently across entries, the conditional likelihood of observing a specific output vector $\tilde{\mathbf b}^\ast$ factorizes as
$	\Pr(\tilde{\mathbf b}^\ast\mid \mathbf b_j)
	=
	\prod_{i\in\bm{\mathcal{S}}}\Pr(\tilde b^\ast_{i,j}\mid b_{i,j}).$

	Taking the likelihood ratio under $\mathbf b_j^{(1)}$ and $\mathbf b_j^{(2)}$, all terms corresponding to the identical dimensions cancel out, leaving only the single differing entry:
\begin{equation}
		\frac{\Pr(\tilde{\mathbf b}^\ast\mid \mathbf b_j^{(1)})}{\Pr(\tilde{\mathbf b}^\ast\mid \mathbf b_j^{(2)})}
	=
	\frac{\Pr(\tilde b^\ast_{i^\circ,j}\mid b^{(1)}_{i^\circ,j})}{\Pr(\tilde b^\ast_{i^\circ,j}\mid b^{(2)}_{i^\circ,j})}.
\end{equation}

	By Lemma~\ref{lem:rr_ldp_single}, the above single-entry ratio is upper bounded by $e^{\varepsilon_j}$ for any $\tilde b^\ast_{i^\circ,j}\in\{0,1\}$. 
	Therefore, the entire vector-level ratio can also be bounded by $e^{\varepsilon_j}$.
\end{proof}

In MIRROR (i.e., Alg.~\ref{Alg:MSIPM}), worker $w_j$ first samples an epoch-stable report $\tilde{\mathbf b}^{\mathrm{perm}}_j(e)$ via entry-wise RR, which satisfies $\varepsilon_j$-personalized LDP by Lemma~\ref{lem:rr_ldp_vector}. 
All subsequent within-epoch transmissions simply output the same already-randomized $\tilde{\mathbf b}^{\mathrm{perm}}_j(e)$.
Hence, each observable report remains distributed exactly as the original LDP mechanism output, and the per-round distinguishability bound in \eqref{eq:personalized_ldp_def_rev} is preserved unchanged.

\subsection{Post-processing Invariance of $\varepsilon_j$-Personalized LDP}
\label{app:proofs:post_processing}

After MIRROR outputs the perturbed intent report $\tilde{\mathbf b}_j$, the platform uses $\tilde{\mathbf b}_j$ together with public system parameters and protocol-defined execution states to construct candidate sets, compute the offline pre-plan (RAPCoD/ASPIRE-Off), and trigger online temporary recruitment.
We show that the intent-report-dependent components of these computations do not consume additional LDP budget beyond $\varepsilon_j$.

\begin{lem}[Post-processing invariance for $\varepsilon_j$-personalized LDP]
	\label{lem:post_processing}
	Let $f$ be an $\varepsilon_j$-personalized LDP mechanism that maps $\mathbf b_j$ to $\tilde{\mathbf b}_j$.
	For any (possibly randomized) mapping $g$ defined on the output space of $f$, the composed mechanism $g\circ f$
	is also $\varepsilon_j$-personalized LDP.
\end{lem}

\begin{proof}
	Consider any two adjacent inputs $\mathbf b_j^{(1)},\mathbf b_j^{(2)}$ and any output $o$ in the range of $g$.
	By the law of total probability over $\tilde{\mathbf b}$, we have
\begin{equation}
		\Pr(o\mid \mathbf b_j)
	=
	\sum_{\tilde{\mathbf b}} \Pr(o\mid \tilde{\mathbf b})\,\Pr(\tilde{\mathbf b}\mid \mathbf b_j).
\end{equation}
	Thus, we have
\begin{equation}
	\begin{aligned}
		\frac{\Pr(o\mid \mathbf b_j^{(1)})}{\Pr(o\mid \mathbf b_j^{(2)})}
	&	=
		\frac{\sum_{\tilde{\mathbf b}} \Pr(o\mid \tilde{\mathbf b})\Pr(\tilde{\mathbf b}\mid \mathbf b_j^{(1)})}
		{\sum_{\tilde{\mathbf b}} \Pr(o\mid \tilde{\mathbf b})\Pr(\tilde{\mathbf b}\mid \mathbf b_j^{(2)})}
		\\&\le
		\max_{\tilde{\mathbf b}}
		\frac{\Pr(\tilde{\mathbf b}\mid \mathbf b_j^{(1)})}{\Pr(\tilde{\mathbf b}\mid \mathbf b_j^{(2)})}
		\le
		e^{\varepsilon_j},
	\end{aligned}
\end{equation}
	where the last inequality follows directly from the $\varepsilon_j$-personalized LDP guarantee of $f$
	in \eqref{eq:personalized_ldp_def_rev}.
\end{proof}

\noindent
Lemma~\ref{lem:post_processing} implies that the intent-report-dependent components of candidate filtering, the RAPCoD/ASPIRE-Off contract profile, and online temporary recruitment are post-processing of MIRROR outputs. Therefore, these components do not consume additional mechanism-level LDP budget beyond the original $\varepsilon_j$.
\subsection{Mitigating Multi-snapshot Frequency Attacks via Memoization}
\label{app:proofs:memoization}

We next clarify why MIRROR remains robust under repeated participation.
The key point is that, within each memo-epoch, the worker reuses an epoch-stable perturbed intent vector.
Hence, an adversary observing multiple rounds within the same epoch does \emph{not} obtain additional independent samples,
and frequency averaging cannot further improve inference accuracy.

\begin{Prop}[No frequency-averaging gain within a memo-epoch]
	\label{Prop:memo_no_variance_reduction}
	Fix any memo-epoch $e$.
	Under Alg.~\ref{Alg:MSIPM}, the reported sequence
	$\{\tilde{\mathbf b}_j^{(1)},\ldots,\tilde{\mathbf b}_j^{(T)}\}$ within epoch $e$ satisfies
	\begin{equation}
		\tilde{\mathbf b}_j^{(1)}=\cdots=\tilde{\mathbf b}_j^{(T)}=\tilde{\mathbf b}^{\mathrm{perm}}_j(e).
	\end{equation}
	Consequently, for any adversary that relies on empirical frequency statistics (e.g., \eqref{eq:freq_stat_intent_rev} and \eqref{eq:majority_inference_intent_rev}),
	all $T$ observations within epoch $e$ are equivalent to observing a \emph{single} perturbed report.
\end{Prop}

\begin{proof}
	By Alg.~\ref{Alg:MSIPM}, once $\tilde{\mathbf b}^{\mathrm{perm}}_j(e)$ is generated and stored, the worker reports
	$\tilde{\mathbf b}_j^{(\tau)} \leftarrow \tilde{\mathbf b}^{\mathrm{perm}}_j(e)$
	for every round $\tau$ within the same epoch $e$.
	Therefore, the observed sequence can stay constant.
	
	For any task $s_i$, the empirical frequency in \eqref{eq:freq_stat_intent_rev} reduces to
$	F_{i,j}=\frac{1}{T}\sum_{\tau=1}^T \tilde b^{(\tau)}_{i,j}
	=
	\tilde b^{\mathrm{perm}}_{i,j}(e)$,
	which is independent of $T$.
	Hence, repeated observations within epoch $e$ do not create additional samples for averaging, and cannot yield any frequency-based accuracy gain.
	As a result, any frequency-threshold inference rule (including \eqref{eq:majority_inference_intent_rev}) is fully determined by the single epoch-stable report
	$\tilde{\mathbf b}^{\mathrm{perm}}_j(e)$.
\end{proof}

Proposition~\ref{Prop:memo_no_variance_reduction} directly explains why the MSR of iParts remains low and stable as the number of snapshots increases:
unlike i.i.d. per-round perturbation, MIRROR prevents the adversary from accumulating independent evidence within a memo-epoch.

\subsection{Finite Termination and Constraint Satisfaction of PRIMER}
\label{app:proofs:primer}

We next establish two basic properties of PRIMER.
First, the budget scan terminates after a finite number of steps.
Second, any returned $\varepsilon_j^\star$ is guaranteed to satisfy the distortion constraint \eqref{equ.25e_rev}
and the one-snapshot inference constraint \eqref{equ.25f_rev} by construction.

\begin{Prop}[Finite termination of PRIMER]
	\label{Prop:primer_termination}
	PRIMER (Alg.~\ref{Alg:PIRD}) terminates in finite steps.
	In particular, the number of scanned budgets is at most
	$1+\left\lceil\frac{\varepsilon_j^{\max}-\varepsilon_j^{\min}}{\Delta\varepsilon}\right\rceil$.
\end{Prop}

\begin{proof}
	PRIMER initializes $\varepsilon\leftarrow \varepsilon_j^{\min}$ and increases it by a fixed step $\Delta\varepsilon>0$.
	The algorithm stops once either a feasible $\varepsilon$ is found or $\varepsilon>\varepsilon_j^{\max}$.
	Therefore, the number of iterations is upper bounded by the size of the discrete scan grid.
\end{proof}

\begin{Prop}[Constraint satisfaction for calibrated workers]
	\label{Prop:primer_feasibility}
	If PRIMER returns a calibrated $\varepsilon_j^\star$ for worker $w_j$, then the instantiated IPM $f(\cdot\mid\cdot;\varepsilon_j^\star)$ satisfies
	$Q^{\mathrm{loss}}_j(\phi,f,\Delta_j)\le Q^{\mathrm{loss}}_{\max}$ in \eqref{equ.25e_rev} and
	$\xi_j\ge \beta^0$ (equivalently \eqref{equ.25f_rev}) by construction.
	If PRIMER returns \textsf{INFEASIBLE}, worker $w_j$ is excluded from the candidate sets, and thus constraints \eqref{equ.25e_rev}--\eqref{equ.25f_rev}
	hold for all workers that participate in RAPCoD.
\end{Prop}

\begin{proof}
	PRIMER outputs $\varepsilon_j^\star$ only when the acceptance test in Alg.~\ref{Alg:PIRD} is satisfied, i.e.,
	both $Q^{\mathrm{loss}}_j(\phi,f,\Delta_j)\le Q^{\mathrm{loss}}_{\max}$ and $\xi_j\ge \beta^0$ hold under $f(\cdot\mid\cdot;\varepsilon_j^\star)$.
	If no $\varepsilon$ in $[\varepsilon_j^{\min},\varepsilon_j^{\max}]$ passes the test, PRIMER returns \textsf{INFEASIBLE};
	the worker is then removed from the effective candidate sets, enforcing \eqref{equ.25e_rev}--\eqref{equ.25f_rev} over the participating workers.
\end{proof}

\subsection{Feasibility Preservation and Finite-step Properties of ASPIRE-Off and ASPIRE-On}
\label{app:proofs:aspireoff}

We finally provide proof details for ASPIRE-Off and ASPIRE-On.
For ASPIRE-Off, exact feasible best-response updates are accepted only when the resulting profile remains feasible and yields a strict potential improvement, which implies feasibility preservation and finite-step convergence to a constrained NE.
For ASPIRE-On, because Alg.~\ref{Alg:TaskDPAsyncOn} adopts DP-assisted remedial search over bounded rounds, we state feasibility preservation, finite termination, and monotonic potential improvement for online execution.

\begin{Prop}[Feasibility preservation of ASPIRE-Off]
	\label{Prop:aspireoff_feasibility}
	Suppose the initial profile $\bm a(0)$ is feasible, i.e., $\bm a(0)\in\mathcal{A}^{\mathrm{feas}}$ in \eqref{eq:A_feas_task_en}.
	Then every accepted update in ASPIRE-Off preserves feasibility:
	$\bm a(t)\in\mathcal{A}^{\mathrm{feas}}$ for all iterations $t$.
\end{Prop}

\begin{proof}
	In Alg.~\ref{Alg:TaskDPAsyncOff}, a best-response update $\big(a_i^{\mathrm{br}}(t),\bm a_{-i}(t)\big)$ is accepted only if $a_i^{\mathrm{br}}(t)\in\mathcal{A}^{\mathrm{uni}}_i(t)$,
	which by definition satisfies the feasibility conditions in \eqref{eq:A_feas_task_en}, including local budget feasibility, worker exclusivity, and the risk constraints.
	Since ASPIRE-Off commits at most one accepted update per round, the resulting profile $\bm a(t{+}1)$ remains feasible.
	The claim then follows by induction on $t$.
\end{proof}

\begin{Prop}[Finite-step convergence of ASPIRE-Off to a constrained NE]
	\label{Prop:aspireoff_convergence}
	ASPIRE-Off terminates in finite steps and returns a constrained Nash equilibrium on $\mathcal{A}^{\mathrm{feas}}$.
\end{Prop}

\begin{proof}
	By Theorem~\ref{thm:task_exact_potential_en}, the offline game is an exact potential game on the finite feasible set $\mathcal{A}^{\mathrm{feas}}$ and satisfies FIP.
	ASPIRE-Off accepts an exact feasible best-response update only when $\Delta_i(t)>0$, which yields a strict increase of the exact potential value.
	Because $\mathcal{A}^{\mathrm{feas}}$ is finite, strict potential ascent cannot continue indefinitely.
	Thus, ASPIRE-Off terminates in finite steps when no feasible unilateral improvement exists.
	At termination, no task admits a feasible unilateral deviation that improves its payoff (equivalently, increases the potential), which is exactly the definition of a constrained NE on $\mathcal{A}^{\mathrm{feas}}$.
\end{proof}
\begin{Prop}[Feasibility preservation and finite termination of ASPIRE-On]
	\label{Prop:aspireon_convergence}
	Starting from an operationally feasible online profile $\bm a^{\mathrm{on}}(0)\in\mathcal{A}^{\mathrm{on,op}}$, ASPIRE-On preserves operational feasibility for every accepted temporary recruitment update and terminates in a finite number of rounds. Moreover, every accepted update passes the updated-task restoration screening and strictly improves the exact online potential $\Phi^{\mathrm{on}}$ by more than $\varepsilon'$.
\end{Prop}

\begin{proof}
	In Alg.~\ref{Alg:TaskDPAsyncOn}, an online update is accepted only after the candidate profile remains in $\mathcal{A}^{\mathrm{on,op}}$ and the updated task's candidate belongs to $\mathcal{R}^{\mathrm{on}}_i(\bm a^{\mathrm{on}}_{-i})$. Hence, every accepted update preserves operational feasibility and satisfies the updated-task restoration screening. In addition, ASPIRE-On accepts an update only when $\Delta_i^{\mathrm{on}}(r)>\varepsilon'$, which implies a strict increase of the exact online potential $\Phi^{\mathrm{on}}$. Since the algorithm either stops when no verified DP-assisted remedial candidate yields sufficient potential gain or reaches the prescribed $R_{\max}$ rounds, it terminates in a finite number of rounds.
\end{proof}
\end{document}